\documentclass[11pt]{article}
\usepackage[margin=0.9in]{geometry}
\usepackage{natbib}
\usepackage{epsfig,psfrag,url}
\usepackage{graphics}
\usepackage{graphicx}
\usepackage{amsmath, amssymb}
\usepackage{amsthm}
\usepackage{verbatim}

\usepackage{amsfonts}
\usepackage{algorithmic}
\usepackage{algorithm}

\newcommand{\imply}{\Longrightarrow}

\newcommand{\rnk}{\mathrm{rank}}

\DeclareMathOperator*{\argmin}{arg\,min}

\DeclareMathOperator*{\mini}{minimize}

\newcommand{\SI}{\textsc{Soft-Impute}}
\newcommand{\NSI}{\textsc{NC-Impute}}
\newcommand{\diag}{\mathrm{diag}}
\newcommand{\B}{\boldsymbol}
\newcommand{\M}{\mathbf}
\newcommand{\vv}{\text{v}}

\newcommand{\var}{\mbox{var}}

\usepackage{subfigure}

\usepackage{bm}
\usepackage{bbm}
\usepackage{url}
\usepackage{hyperref}

\newtheorem{proposition}{Proposition}
\newtheorem{remark}{Remark}
\newtheorem{mydef}{Definition}
\newtheorem{lemma}{Lemma}

\begin{document}

\title{Matrix Completion with Nonconvex Regularization: Spectral Operators and Scalable Algorithms}
\author{Rahul Mazumder, Diego F. Saldana, Haolei Weng}



\date{}
\maketitle

\begin{abstract}
In this paper, we study the popularly dubbed matrix completion problem, where the task is to ``fill in'' the unobserved entries of a matrix from a small subset of observed entries, under the assumption that the underlying matrix is of low-rank. Our contributions herein, enhance our prior work on nuclear norm regularized problems for matrix completion~\citep{MazumderEtal2010} by incorporating a continuum of nonconvex penalty functions between the convex nuclear norm and nonconvex rank functions. Inspired by \SI~\citep{MazumderEtal2010,HastieEtal2015}, we propose \NSI~--- an EM-flavored algorithmic framework for computing a family of nonconvex penalized matrix completion problems with warm-starts. We present a systematic study of the associated spectral thresholding operators, which play an important role in the overall algorithm. We study convergence properties of the algorithm. Using structured low-rank SVD computations, we demonstrate the computational scalability of our proposal for problems up to the Netflix size (approximately, a $500,000 \times 20, 000$ matrix with $10^8$ observed entries). We demonstrate that on a wide range of synthetic and real data instances, our proposed nonconvex regularization framework leads to low-rank solutions with better predictive performance when compared to those obtained from nuclear norm problems. Implementations of algorithms proposed herein, written in the {\texttt{R}} language, are made available on {\texttt{github}}.
\end{abstract}

\section{Introduction}\label{sec1}
In several problems of contemporary interest, arising for instance, in recommender system applications, for example, the Netflix Prize competition~\citep{netflix}, observed data is in the form of a large sparse matrix, $Y_{ij}, (i,j) \in \Omega$, where $\Omega \subset \{ 1, \ldots, m\} \times \{ 1, \ldots, n\}$, with $|\Omega| \ll mn$. Popularly dubbed as the matrix completion problem~\citep{CandesRecht2009,MazumderEtal2010}, the task is to predict the unobserved entries, under the assumption that the underlying matrix is of low-rank. This leads to the natural rank regularized optimization problem:
\begin{equation}\label{rank-prob-1}
\mini_{X}  \;\;  \frac12\| \mathcal{P}_{\Omega}( X  - Y ) \|_{F}^2 + \lambda~\rnk(X),
\end{equation}
where, $\mathcal{P}_{\Omega}(X)$ denotes the projection of $X_{m \times n}$ onto the observed indices $\Omega$ and is zero otherwise; and $\| \cdot \|_F$ denotes the usual Frobenius norm of a matrix. 
 Problem~\eqref{rank-prob-1}, however, is computationally difficult due to the presence of the combinatorial rank constraint~\citep{chistov1984complexity}. A natural convexification~\citep{fazel-thes,RechtEtal2010} of  $\rnk(X)$ is $\|X\|_{*}$, the nuclear norm of $X$, which leads to the  following surrogate of Problem~\eqref{rank-prob-1}:
\begin{equation}\label{conv-prob-1}
\mini_{X}  \;\;  \frac12\| \mathcal{P}_{\Omega}( X  - Y ) \|_{F}^2 + \lambda \| X\|_*.
\end{equation}
\cite{CandesRecht2009,CandesPlan2010} show that under some assumptions on the underlying ``population'' matrix, a solution to Problem~\eqref{conv-prob-1} approximates a solution to Problem~\eqref{rank-prob-1} reasonably well.
The estimator obtained from Problem~\eqref{conv-prob-1} works quite well: the nuclear norm shrinks the singular values and simultaneously sets many of the singular values to zero, thereby encouraging low-rank solutions. It is thus not surprising that Problem~\eqref{conv-prob-1} has enjoyed a significant amount of attention in the wider statistical community over the last decade. There have been impressive advances in understanding its statistical
properties~\citep{CandesPlan2010, CandesTao2010, RechtEtal2010, Recht2011, gross2011recovering, RohdeTsybakov2011, koltchinskii2011nuclear, NegahbanWainwright2011, chen2015incoherence, lecue2018regularization, chen2019noisy}. Motivated by the work of~\cite{CandesRecht2009,CaiEtal2010}, the authors in~\cite{MazumderEtal2010} proposed \SI, an EM-flavored~\citep{DempsterEtal1977} algorithm for optimizing Problem~\eqref{conv-prob-1}. For some other computational work in developing scalable algorithms for Problem~\eqref{conv-prob-1},
see the papers~\cite{jaggi2010simple, FGM_2015_FW, HastieEtal2015}, and references therein.
Typical assumptions under which the nuclear norm works as a good proxy for the low-rank problem require the entries of the singular vectors of the ``true'' low-rank matrix to be sufficiently spread, and the missing pattern to be roughly uniform. The proportion of observed entries needs to be sufficiently larger than the number of parameters of the matrix $O \left((m+n)r\right)$, where, $r$ denotes the rank of the true underlying matrix. Some extensions under general sampling distribution has been made in \cite{klopp2014noisy, alquier2015bayesian}. \cite{NegahbanWainwright2012} proposes improvements  with a (convex) weighted nuclear norm penalty in addition to spikiness constraints for the noisy matrix completion problem.

The nuclear norm penalization framework, however, has limitations. If some conditions mentioned above fail, Problem~\eqref{conv-prob-1} may fall short of delivering reliable low-rank estimators with good prediction performance (on the missing entries). Since the nuclear norm shrinks the singular values, in order to obtain an estimator with good explanatory power, it often results in a matrix estimator with high numerical rank --- thereby leading to models that have higher rank than what might be desirable.   The limitations mentioned above, however, should not come as a surprise to an expert --- especially, if one draws a parallel connection to the 
\textsc{Lasso}~\citep{Ti96}, a popular sparsity inducing shrinkage mechanism effectively used in the context of sparse linear modeling and regression.
In the linear regression context, the \textsc{Lasso} often leads to dense models and suffers when the features are highly correlated --- the limitations of the \textsc{Lasso} are quite well known in the statistics literature, and there have been major strides in moving beyond the convex $\ell_{1}$-penalty to more aggressive forms of nonconvex penalties~\citep{FanLi2001,zouli08,MazumderEtal2011,Zhang2010,zhang2012general,loh2015regularized,bertsimas2015best,zheng2017does, feng2017sorted}. The key principle in these methods is the use of nonconvex regularizers that better approximate the $\ell_{0}$-penalty, leading to possibly nonconvex estimation problems. Thusly motivated, we study herein, the following family of nonconvex regularized estimators for the task of (noisy) matrix completion:
 \begin{equation}\label{nonconv-prob-1}
 \begin{aligned}
\mini_{X} \underbrace{ \; \frac12\| \mathcal{P}_{\Omega}( X  - Y ) \|_{F}^2 +  \sum_{i=1}^{\min \{ m, n \}} P(\sigma_{i}(X) ; \lambda,\gamma)}_{:=f(X)},
\end{aligned}
\end{equation}
where, $\sigma_{i}(X), i \geq 1$ are the singular values of $X$ and $ \sigma \mapsto P(\sigma ; \lambda,\gamma)$ is a concave penalty function on $[0, \infty)$ that takes the value 
$\infty$ whenever $\sigma<0$. We will denote an estimator obtained from Problem~\eqref{nonconv-prob-1} by $\hat{X}_{\lambda, \gamma}$.
The family of penalty functions $P(\sigma ; \lambda,\gamma)$ is indexed by the parameters $(\lambda, \gamma)$ --- these parameters together control the amount of nonconvexity and shrinkage --- see for example~\cite{MazumderEtal2011,zhang2012general} and also Section~\ref{sec2}, herein, for examples of such nonconvex families.

A caveat in considering problems of the form~\eqref{nonconv-prob-1} is that they lead to nonconvex optimization problems and thus obtaining a certifiably optimal global minimizer is generally difficult. Fairly recently, ~\cite{bertsimas2015best,mazumder2015discrete} have shown that subset selection problems in sparse linear regression can be computed using advances in mixed integer quadratic optimization. Such global optimization methods, however, do not apply to matrix variate problems involving spectral\footnote{We say that a function is a spectral function of a matrix $X$, if it depends only upon the singular values of $X$. The state of the art algorithmics in mixed integer Semidefinite optimization problems is in its nascent stage; and not even comparable to the technology for mixed integer quadratic optimization.} penalties, as in Problems~\eqref{rank-prob-1} or~\eqref{nonconv-prob-1}. The main focus in our work herein is to develop a computationally scalable algorithmic framework that allows us to obtain high quality stationary points or upper bounds\footnote{Since the problems under consideration are nonconvex, our methods are not guaranteed to reach the global minimum -- we thus refer to the solutions obtained as 
\emph{upper bounds}. In many synthetic examples, however, the solutions are indeed seen to be globally optimal. We do show rigorously, however, that these solutions are first order stationary points for the optimization problems under consideration.}
for Problem~\eqref{nonconv-prob-1} ---  we obtain a path of solutions $\hat{X}_{\lambda, \gamma}$ across a grid of values of $(\lambda, \gamma)$ for Problem~\eqref{nonconv-prob-1} by employing warm-starts, following the path-following scheme proposed in~\cite{MazumderEtal2011}.
Leveraging problem structure, modern advances in computationally scalable low-rank SVDs and appropriately advancing the tricks successfully employed in~\cite{MazumderEtal2010,HastieEtal2015},
we empirically demonstrate the computational scalability of our method for problems of the size of the Netflix dataset, a matrix of size (approx.) $480,000 \times 18,000$ with $\sim 10^8$ observed entries. Perhaps most importantly, we demonstrate empirically that the resultant estimators lead to better statistical properties (i.e., the estimators have lower rank and enjoy better prediction performance) over nuclear norm based estimates, on a variety of 
problem instances.

Some recent works~\citep{jain2010guaranteed, jain2013low, hardt2014understanding, hardt2014fast, chen2015fast, ma2017implicit, chen2019nonconvex} study the scope of alternating minimization or (projected) gradient stylized algorithmic strategies for the rank constrained optimization problem, similar to Problem~\eqref{rank-prob-1} --- see also~\cite{HastieEtal2015}
for related discussions. We should emphasize that our work herein, studies the \emph{entire} family of nonconvex spectral penalized problems of the form of 
Problem~\eqref{nonconv-prob-1}, and is hence more general than the class of estimation problems considered in those works. We establish empirically that this flexible family of nonconvex penalized estimators leads to solutions with better statistical properties than those available from particular instantiations of the penalty function --- nuclear norm regularization~\eqref{conv-prob-1} and rank regularization~\eqref{rank-prob-1}. Along the lines of the aforementioned works, there exists an active stream of research on characterizing the global optimality of local algorithms for various matrix factorization based formulations~\citep{bhojanapalli2016global, ge2016matrix, sun2016guaranteed, zheng2016convergence, ge2017no, shapiro2018matrix}. Our paper focuses on a more general family of nonconvex regularization, with admittedly less strong algorithmic guarantees. Finally, a series of iterative reweighted algorithms have been proposed and discussed \citep{MazumderEtal2010, mohan2010reweighted, fornasier2011low, mohan2012iterative, gu2017weighted}, largely motivated by the reweighting ideas from sparse recovery problems \citep{zou2006adaptive, candes2008enhancing, daubechies2010iteratively}. Different weight formulas have been suggested to improve the statistical and computational efficiency. These are, however, beyond the scope of the current paper.

\subsection{Contributions and Outline}

The main contributions of our paper can be summarized as follows:

\begin{itemize}
\item We propose a computational framework for nonconvex penalized matrix completion problems of the form~\eqref{nonconv-prob-1}.
Our algorithm: \NSI, may be thought of as a novel adaptation (with important enhancements and modifications) of the EM-stylized procedure \SI\\ \citep{MazumderEtal2010} to more general nonconvex penalized thresholding operators.

\item We present an in-depth investigation of  nonconvex spectral thresholding operators, which form the main building block of our algorithm. We also study their effective degrees of freedom (\emph{df}), which provide a simple and intuitive way to calibrate the two-dimensional grid of tuning parameters, extending the scope of the 
method proposed in nonconvex penalized (least squares) regression by~\cite{MazumderEtal2011} to spectral thresholding operators. We propose computationally efficient methods to approximate the \emph{df} using
tools from random matrix theory. 

\item We provide comprehensive computational guarantees of our algorithm --- this includes the number of iterations needed to reach a first order stationary point and the asymptotic convergence of the sequence of estimates produced by \NSI.

\item Every iteration of \NSI~requires the computation of a low-rank SVD of a structured matrix, for which we propose new methods.
Using efficient warm-start tricks to speed up the low-rank computations, we demonstrate the effectiveness of our proposal to large scale instances up to the Netflix size in reasonable computation times.

\item Over a wide range of synthetic and real-data examples, we show that our proposed nonconvex penalized framework leads to high quality solutions with excellent statistical properties, which are often found to be significantly better than nuclear norm regularized solutions in terms of producing low-rank solutions with good predictive performances.

\item Implementations of our algorithms in the {\texttt{R}} programming language have been made publicly available on {\texttt{github}} at:~\url{https://github.com/diegofrasal/ncImpute}.

\end{itemize}

The remainder of the paper is organized as follows. Section~\ref{sec2} studies several properties of nonconvex spectral penalties and associated spectral thresholding operators, including their effective degrees of freedom. Section~\ref{sec3} describes our algorithmic framework~\NSI~and studies the convergence properties of the algorithm. Section~\ref{sec4} presents numerical experiments demonstrating the usefulness of nonconvex penalized estimation procedures in terms of superior statistical properties on several synthetic datasets --- we also show the usefulness of these estimators on several real data instances. Section~\ref{sec-concl} contains the conclusions and discusses several important future research directions. To improve readability, some technical materials and empirical results are relegated to Section~\ref{last:appen}.

\paragraph{Notation:}
For a matrix $A_{m \times n}$, we denote its $(i,j)$th entry by $a_{ij}$. 
$\mathcal{P}_{\Omega}(A)$ is a matrix with its $(i,j)$th entry given by $a_{ij}$ for $(i,j) \in \Omega$ and zero otherwise, with $\Omega \subset \{ 1, \ldots, m \} \times \{ 1, \ldots, n\}$.
We use the notation $\mathcal{P}_{\Omega}^\perp(A) = A - \mathcal{P}_{\Omega}(A)$ to denote the projection of $A$ onto the complement of $\Omega$.
Let $\sigma_{i}(A), i = 1, \ldots, \max \{m,n\}$ denote the singular values of $A$, with $\sigma_{i}(A) \geq \sigma_{i+1}(A)$ (for all $i$) -- we will use the notation $\B\sigma(A)$ to denote the vector of singular values. When clear from the context, we will simply write $\B\sigma$ instead of $\B\sigma(A)$. For a vector $\M{a}=(a_{1}, \ldots, a_{n})\in \mathbb{R}^{n}$, we will use the notation $\diag(\M{a})$ to denote an $n \times n$ diagonal matrix with $i$th diagonal entry being $a_{i}$.

\section{Spectral Thresholding Operators}\label{sec2}
We begin our analysis by considering the fully observed version of Problem~\eqref{nonconv-prob-1}, given by:
\begin{equation}\label{gen-thresh-1}
\min_{X}~\underbrace{\frac 12 \| X - Z \|_{F}^2 + \sum_{i=1}^{\min \{ m, n \}} P(\sigma_{i}(X); \lambda,\gamma)}_{:=g(X)}
\end{equation}
where, for a given matrix $Z$, a minimizer of the function $g(X)$, denoted by $S_{\lambda, \gamma}(Z)$, is the \textit{spectral thresholding operator} induced by the spectral penalty $\sum_{i} P(\sigma_{i}(X); \lambda,\gamma).$
Suppose $U \diag(\B\sigma)V'$ denotes the SVD of $Z$.
For the nuclear norm regularized problem with the penalty function $P(\sigma_{i}(X); \lambda,\gamma) = \lambda \sigma_{i}(X),$ the corresponding thresholding operator, denoted by $S_{\lambda, \ell_1}(Z)$ (say), is given by the familiar soft-thresholding operator~\citep{CaiEtal2010,MazumderEtal2010}:
\begin{equation}\label{soft-1}
\begin{aligned}
  S_{\lambda, \ell_{1}}(Z) := U \diag (s_{\lambda, \ell_{1}}(\B\sigma)) V' 
\end{aligned}
 \end{equation}
where, $s_{\lambda, \ell_{1}}(\sigma_{i}) :=  (\sigma_{i} - \lambda)_+$, $(\cdot)_+ = \max \{ \cdot, 0\}$ and $s_{\lambda, \ell_{1}}(\sigma_{i})$ is the $i$th entry of  $s_{\lambda, \ell_{1}}(\B\sigma)$ (due to separability of the thresholding operator).
Here, $S_{\lambda, \ell_{1}}(Z)$ is the the soft-thresholding operator on the singular values of $Z$ and plays a crucial role in the~\SI~algorithm~\citep{MazumderEtal2010}.
For the rank regularized problem, with 
$$P(\sigma_{i}(X); \lambda,\gamma) = \lambda \mathbbm{1}(\sigma_{i}(X) > 0),$$ the thresholding operator denoted by $S_{\lambda, \ell_{0}}(Z)$ is given by the hard-thresholding operator~\citep{MazumderEtal2010}:
\begin{equation}\label{hard-1}
S_{\lambda, \ell_{0}}(Z) := U \diag (s_{\lambda, \ell_{0}}(\B\sigma)) V'
 \end{equation}
with $s_{\lambda, \ell_{0}}(\sigma_{i}) = \sigma_{i} \mathbbm{1}(\sigma_{i} > \sqrt{2\lambda}).$ A closely related thresholding operator that retains the top $r$ singular values and sets the remaining to zero formed the basis of the \textsc{Hard-Impute} algorithm in~\cite{MazumderEtal2010,troyanskaya2001missing}. The results in~\eqref{soft-1} and \eqref{hard-1} suggest a curious link --- the spectral thresholding operators (for the two specific choices of the spectral penalty functions given above) are tied to the corresponding thresholding functions that operate only on the singular values of the matrix --- in other words, the operators $S_{\lambda, \ell_{1}}(Z),S_{\lambda, \ell_{0}}(Z)$
do \emph{not} change the singular vectors of the matrix $Z$. It turns out that a similar result holds true for more general spectral penalty functions $P(\cdot; \lambda,\gamma)$ as the following proposition illustrates.

\begin{proposition}\label{prop1}
Let $Z=U\diag(\B\sigma)V'$ denote the SVD of $Z$, and $s_{\lambda, \gamma}( \B\sigma )$ denote the following thresholding operator on the singular values of $Z$:
\begin{equation}\label{gen-thresh-11}
s_{\lambda, \gamma}( \B\sigma ) \in \argmin_{\B\alpha \geq \M{0} }\underbrace{ \frac12 \| \B\alpha - \B\sigma  \|_{2}^2 + \sum_{i=1}^{\min \{ m, n \}} P(\alpha_{i}; \lambda,\gamma)}_{:= \bar{g}(\B\alpha)}.
\end{equation}
Then  $S_{\lambda, \gamma}(Z) = U \diag(s_{\lambda, \gamma}(\B\sigma)) V'.$  
\end{proposition}
\begin{proof}
Note that by the Wielandt-Hoffman inequality~\citep{horn2012matrix} we have that:
$\| X - Z \|_{F}^2 \geq \| \B\sigma(X) - \B\sigma(Z) \|_{2}^2,$ where, for a vector $\M{a}$, $\|\M{a}\|_{2}$ denotes the standard Euclidean norm. Equality holds when $X$ and $Z$ share the same left and right singular vectors.
This leads to:
\begin{equation*}
\begin{aligned}
\frac12 \| X - Z \|_{F}^2 +\sum_{i=1}^{\min \{ m, n \}}P(\sigma_{i}(X) ; \lambda,\gamma)  \geq \frac12 \| \B\sigma(X) -\B\sigma(Z) \|_{2}^2 + \sum_{i=1}^{\min \{ m, n \}} P(\sigma_{i}(X) ; \lambda,\gamma).
\end{aligned}
\end{equation*}
In the above inequality, note that the left hand side is $g(X)$ (defined in~\eqref{gen-thresh-1}) and right hand side is $\bar{g}(\B\sigma(X))$ (defined in~\eqref{gen-thresh-11}). It follows that 
\begin{equation}\label{defn-ineq-1-1}
\min_{X} ~~ g(X) \geq \min_{\B\sigma(X)}~~ \bar{g}(\B\sigma(X)) = \bar{g} \left(s_{\lambda, \gamma}(\B\sigma) \right),
\end{equation}
where, we used the observation that $\B\sigma(X) \geq \M{0}$ and
 $s_{\lambda, \gamma}( \B\sigma),$ as defined in~\eqref{gen-thresh-11} minimizes $\bar{g}(\B\sigma(X))$.
In addition, this minimum is attained by the function $g(X)$, at the choice $X= U \diag(s_{\lambda, \gamma}(\B\sigma)) V'$.  This completes the proof of the proposition. 
\end{proof}

Due to the separability of the optimization Problem~\eqref{gen-thresh-11} across the coordinates, i.e., $\bar{g}(\B\alpha) = \sum_{i} \bar{g}_i(\alpha_i)$ (where, $\bar{g}_i(\cdot)$ is defined in~\eqref{univ-thresh-funtion}), it suffices to consider each of the subproblems separately. Let $s_{\lambda, \gamma}(\sigma_i)$ denote a minimizer of $\bar{g}_i(\alpha)$, i.e., 
\begin{equation}\label{univ-thresh-funtion}
s_{\lambda, \gamma}(\sigma_i) \in \argmin_{\alpha \geq 0} \; \bar{g}_{i} (\alpha) := \frac12 (\alpha - \sigma_{i})^2 + P(\alpha;\lambda,\gamma).
\end{equation}
It is easy to see that the $i$th coordinate of $s_{\lambda, \gamma}(\B\sigma)$ is given by $s_{\lambda, \gamma}(\sigma_i)$. This discussion suggests that 
our understanding of the spectral thresholding operator 
$S_{\lambda, \gamma}(Z)$ is intimately tied to the univariate thresholding operator~\eqref{univ-thresh-funtion}. 
Thusly motivated, in the following, we present a concise discussion about univariate penalty functions and the resultant thresholding operators. 
We begin with some examples of concave penalties that are popularly used in statistics in the context of sparse linear modeling.
\paragraph{Families of  Nonconvex Penalty Functions:}
Several types of nonconvex penalties are popularly used in high-dimensional regression frameworks---see for example,~\cite{Nikolova2000,LvFan2009,zhang2012general}.
For our setup, since these penalty functions operate on the singular values of a matrix,
it suffices to consider nonconvex functions that are defined only on the nonnegative real numbers. We present a few examples below:
\begin{itemize}
\item The $\ell_{\gamma}$ penalty \citep{FrankFriedman1993} given by 
\[
P(\sigma; \lambda, \gamma)=\lambda \sigma^{\gamma},
\]
where $\lambda > 0$ and $0\leq \gamma <1$.
\item The SCAD penalty \citep{FanLi2001} is defined via:
 $$P'(\sigma;\lambda,\gamma)=\lambda \mathbbm{1}(\sigma \leq \lambda) + \frac{(\gamma\lambda -\sigma)_+}{\gamma-1}\mathbbm{1}(\sigma >\lambda),$$
where $\lambda > 0, \gamma>2$, and $P'(\sigma;\lambda,\gamma)$ denotes the derivative of $\sigma \mapsto P(\sigma;\lambda,\gamma)$ on $\sigma \geq 0$ with
$P(0; \lambda, \gamma)=0$.
\item The MC+ penalty \citep{Zhang2010,MazumderEtal2011} defined as
\begin{align*}
P(\sigma;\lambda,\gamma)=\lambda \left(\sigma-\frac{\sigma^2}{2\lambda \gamma}\right )\mathbbm{1}(0 \leq \sigma< \lambda \gamma) +\frac{\lambda^2 \gamma}{2}\mathbbm{1}(\sigma \geq \lambda \gamma),
\end{align*}
with $\lambda > 0, \gamma >0$.
\item The log-penalty, with 
\[
P(\sigma; \lambda,\gamma)= \lambda\log(\gamma \sigma + 1)/ \log(\gamma +1)
\] 
on $\lambda>0$ and $\gamma > 0$. 
\end{itemize}

\begin{figure*}[htb!]
\begin{center}
\includegraphics[width= \textwidth, height= 0.25\textheight]{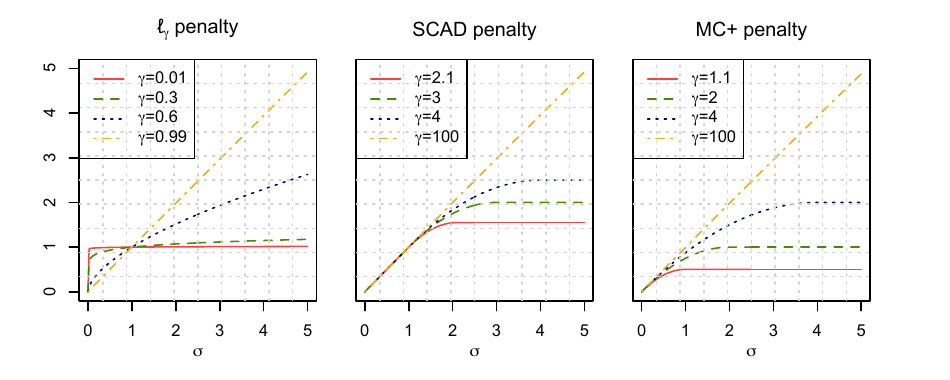}
\includegraphics[width= \textwidth, height= 0.25\textheight]{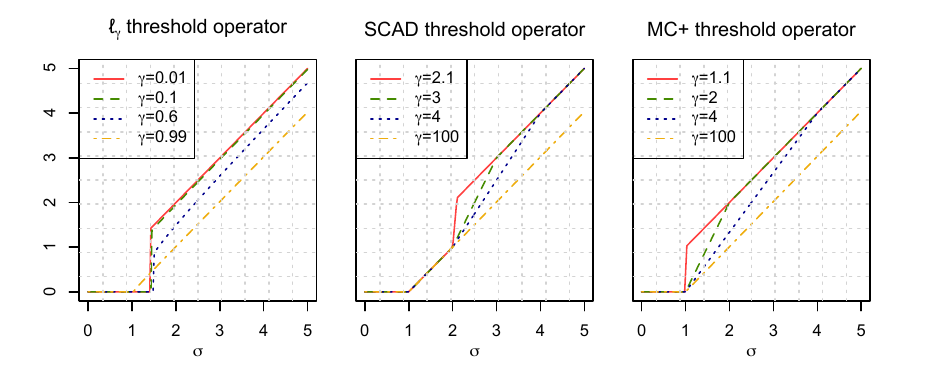}
\caption{\small {[Top panel] Examples of nonconvex penalties $\sigma \mapsto P(\sigma; \lambda, \gamma)$ with $\lambda=1$ for different values of $\gamma$. [Bottom Panel] The corresponding scalar thresholding operators: $\sigma \mapsto s_{\lambda, \gamma}(\sigma)$.
At $\sigma =1$, some of the thresholding operators corresponding to the $\ell_{\gamma}$ penalty function are discontinuous, and some of the other thresholding functions are ``close'' to being so.}} \label{fig1}
\end{center}
\end{figure*}


Figure~\ref{fig1} shows some members of the above nonconvex penalty families. The $\ell_{\gamma}$ penalty function is non differentiable at $\sigma = 0$, due to the unboundedness of
$P'(\sigma; \lambda,\gamma)$ as $\sigma \rightarrow 0+$.
The nonzero derivative at $\sigma = 0+$ encourages sparsity. The $\ell_{\gamma}$ penalty functions show a clear transition from the
$\ell_1$ penalty  to the $\ell_{0}$ penalty  --- similarly, the resultant thresholding operators show a passage from the soft-thresholding
to the hard-thresholding operator. Let us examine the analytic form of the thresholding function induced by the MC+ penalty (for any $\gamma > 1$):
\begin{equation}\label{mcp_thresh}
s_{\lambda, \gamma}(\sigma)=
\begin{cases}
0, & \text{if } \sigma \leq \lambda \\
 \Big( \frac{\sigma -\lambda}{1-1/\gamma} \Big), & \text{if } \lambda <  \sigma  \leq \lambda \gamma \\
\sigma, & \text{if } \sigma > \lambda \gamma.
\end{cases}
\end{equation}
It is interesting to note that for the MC+ penalty, the derivatives are all bounded and the thresholding functions are continuous for all $\gamma > 1$.
As $\gamma \rightarrow \infty$, the threshold operator~\eqref{mcp_thresh} coincides with the soft-thresholding operator. 
However, as $\gamma \rightarrow 1+$ the threshold operator approaches the discontinuous hard-thresholding operator $\sigma \mathbbm{1}(\sigma \geq \lambda)$ --- this is illustrated in Figure~\ref{fig1} and can also be observed by inspecting~\eqref{mcp_thresh}.
Note that the $\ell_{1}$ penalty penalizes small and large singular values in a similar fashion, thereby incurring an increased bias in estimating the larger coefficients.
For the MC+ and SCAD penalties, we observe that they penalize the larger coefficients less severely than the $\ell_{1}$ penalty --- simultaneously, they penalize the smaller coefficients in a manner similar to that of the $\ell_{1}$ penalty. On the other hand, the $\ell_{\gamma}$ penalty (for small values of $\gamma$) imposes a more severe penalty for values of $\sigma \approx 0$, quite different from the behavior of other penalty functions. In general, for a given family of nonconvex penalties $P(\sigma;\lambda, \gamma)$, the effect of $(\lambda, \gamma)$ on the nonconvexity can be characterized through the general concavity quantity $\phi_P$ that is to be introduced in \eqref{def-conc-1}.

\subsection{Properties of Spectral Thresholding Operators}

The nonconvex penalty functions described in the previous section are concave functions on the nonnegative real line. We will now discuss measures that 
may be thought (loosely speaking) to measure the amount of concavity in the functions. For a univariate penalty function $\alpha \mapsto P(\alpha ;\lambda,\gamma)$ on $\alpha \geq 0$, assumed to be differentiable on $(0, \infty)$, we introduce the following quantity ($\phi_P$) that measures the amount of concavity (see also,~\cite{Zhang2010}) of $P(\alpha ;\lambda,\gamma)$:
\begin{equation}\label{def-conc-1}
\phi_P := \; \inf_{\alpha, \alpha' >0} \;\;  \frac{P'(\alpha ;\lambda,\gamma)  - P'(\alpha';\lambda,\gamma)}{\alpha - \alpha'},
\end{equation}
where $P'(\alpha ;\lambda,\gamma)$ denotes the derivative of $P(\alpha ;\lambda,\gamma)$ wrt $\alpha$ on $\alpha >0$. 

We say that the function $g(X)$ (as defined in~\eqref{gen-thresh-1})  is \emph{$\tau$-strongly convex} if the following condition holds:
\begin{equation}\label{str-conv-g}
g(X) \geq g(\widetilde{X}) + \langle \nabla g (\widetilde{X}),  X - \widetilde{X} \rangle + \frac{\tau}{2} \|  X - \widetilde{X} \|_{F}^2,
\end{equation}
for some $\tau \geq 0$ and all $X, \widetilde{X}$. In inequality~\eqref{str-conv-g}, $\nabla g (\widetilde{X})$ denotes \emph{any} subgradient (assuming it exists) of $g(X)$ at $\widetilde{X}$.
If $\tau=0$ then the function is simply convex\footnote{Note that we consider  $\tau \geq 0$  in the definition so that it includes the case of (non strong) convexity.}.
Using standard properties of spectral functions~\citep{borwein-conv,Lewis1995}, it follows that $g(X)$ is $\tau$-strongly convex iff the vector function:
\begin{equation}\label{bar-defn1}
\bar{g}(\B\alpha) =  \frac12 \| \B\alpha - \B\sigma(Z) \|_{2}^2 + \sum_{i=1}^{\min \{ m, n \}} P(\alpha_i;\lambda,\gamma)
\end{equation}
is $\tau$-strongly convex on $\{ \B\alpha: \B\alpha \geq \M{0}\}$, where $ \B\sigma(Z)$ denotes the singular values of $Z$. 
Let us recall the separable decomposition of $\bar{g}(\B\alpha)  = \sum_{i} \bar{g}_{i} (\alpha_{i})$, with
$\bar{g}_{i} (\alpha)$ as defined in~\eqref{univ-thresh-funtion}.
Clearly, the function $\B\alpha \mapsto \bar{g}(\B\alpha)$ is $\tau$-strongly convex (on the nonnegative reals) iff each summand
$\bar{g}_{i} (\alpha)$ is $\tau$-strongly convex on $\alpha \geq 0$. Towards this end, notice that $\bar{g}_{i} (\alpha)$ is convex on $\alpha \geq 0$ iff
$1 +  \phi_{P}  \geq 0 $ --- in particular, $\bar{g}_{i} (\alpha)$ is $\tau$-strongly convex with parameter $\tau = 1 + \phi_{P}$, provided this number is nonnegative.
In this vein, we have the following proposition:
\begin{proposition}\label{conv-prox-map-1}
Suppose $\phi_P > -1$, then the function $X \mapsto g(X)$ is $\tau$-strongly convex  with $\tau = 1 + \phi_P$.
\end{proposition}
For the MC+ penalty, the condition $ \tau= 1 + \phi_P > 0$ is equivalent to $\gamma > 1$. For the $\ell_{\gamma}$ penalty function, with $\gamma<1$, the parameter $\tau = -\infty$, and thus the
function $g(X)$ is not strongly convex. 

\begin{proposition}\label{prop-liphs-1}
Suppose $1 + \phi_P > 0$, then $Z \mapsto S_{\lambda, \gamma}(Z) $ is Lipschitz continuous with constant $\frac{1}{1 + \phi_P}$, i.e, for all $Z_1, Z_2$ we have:
\begin{equation}\label{spec-map-Lipsh}
\| S_{\lambda, \gamma}(Z_1)  - S_{\lambda, \gamma}(Z_{2}) \|_{F} \leq   \frac{1}{1 + \phi_P} \| Z_{1} -  Z_{2} \|_{F}.
\end{equation}
\end{proposition}
\begin{proof}
We rewrite $g(X)$ as:
\begin{align}\label{g-x-line-1}
g(X) = \left \{ \frac12 \| X - Z \|_{F}^2 - \frac{\psi}{2} \| X\|_{F}^2 \right \} + \left \{ \sum_{i=1}^{\min \{ m, n \}} P(\sigma_i(X);\lambda,\gamma) + \frac{\psi}{2} \| X\|_{F}^2 \right \}.
\end{align}
We have that $\| X\|_{F}^2  = \sum_{i=1}^{\min \{ m, n \}} \sigma^2_{i}(X)$. Using the shorthand 
notation $\widetilde{P}(\sigma_i(X)) = P(\sigma_i(X);\lambda,\gamma) + \frac{\psi}{2} \sigma^2_i(X)$, and rearranging the terms in~\eqref{g-x-line-1}, 
 it follows that $S_{\lambda, \gamma}(Z)$, a minimizer of $g(X)$, is given by:
\begin{equation}\label{prox-map-tau-1}
\begin{aligned}
S_{\lambda, \gamma}(Z) \in \argmin_{X}\bigg\{ \frac{1- \psi}{2} \| X - \frac{1}{1 - \psi} Z \|_{F}^2  + \sum_{i=1}^{\min \{ m, n \}} \widetilde{P}(\sigma_i(X)) \bigg \}.
\end{aligned}
\end{equation}
If $\psi + \phi_P > 0$, the function $\sigma_{i} \mapsto \widetilde{P}(\sigma_i)$ is convex for every $i$. If $1 - \psi > 0$, then the first term appearing in the objective function in~\eqref{prox-map-tau-1} is convex. Thus, assuming $\psi + \phi_P > 0, 1 - \psi > 0$ both summands in the above objective function are convex. In particular, the optimization problem~\eqref{prox-map-tau-1} is convex and $Z \mapsto S_{\lambda, \gamma}(Z)$ can be viewed as a 
convex proximal map~\citep{rock-conv-96}. Using standard contraction properties of proximal maps, we have that:
\begin{align*}
 \|S_{\lambda, \gamma}(Z_1) - S_{\lambda, \gamma}(Z_2) \|_{F} &\leq  \left \| \frac{Z_1}{1 - \psi} - \frac{Z_2}{1 - \psi} \right\|_{F} \leq \frac{1}{1 - \psi} \| Z_{1} - Z_{2} \|_{F}.
\end{align*}
Since the above holds true for any $\psi$ as chosen above, optimizing over the value of $\psi$ such that Problem~\eqref{prox-map-tau-1} remains convex gives us $\hat{\psi} = -\phi_P$, i.e., $1/(1 - \hat{\psi}) = 1/(1+ \phi_P)$, thereby leading to~\eqref{spec-map-Lipsh}.
\end{proof}


\subsection{Effective Degrees of Freedom for Spectral Thresholding Operators}\label{sec-edf-spec-ops}
In this section, to better understand the statistical properties of spectral thresholding operators,
we  study their degrees of freedom.
The effective degrees of freedom or \emph{df} is a popularly used statistical notion that measures the amount of ``fitting'' performed by an estimator~\citep{LARS,hastie09:_elemen_statis_learn_II,stein1981estimation}. In the case of classical linear regression, for example, \emph{df} is simply given by  the number of features used in the linear model. This notion applies more generally to additive fitting procedures. 
Following~\cite{LARS,stein1981estimation}, let us consider an additive model of the form:
\begin{equation}\label{iid-add-1}
Z_{ij} = \mu_{ij} + \varepsilon_{ij} ~~~\text{with}~~~ \varepsilon_{ij} \stackrel{\text{iid}}{\sim} N(0, \vv^2),
\end{equation}
for $i = 1, \ldots, m, j = 1, \ldots, n.$ The \emph{df} of $\hat{\mu}:=\hat{\mu}(Z)$, for the fully observed model above, i.e.,~\eqref{iid-add-1} is given by:
$$ \emph{df}(\hat{\mu}) = \sum_{ij} \text{Cov}(\hat{\mu}_{ij}, Z_{ij})/\vv^2,$$
where $\mu_{ij}$ denotes the $(i,j)$th entry of the matrix $\mu$.
For the particular case of a spectral thresholding operator we have $\hat{\mu} = S_{\lambda,\gamma}(Z).$
When $ Z \mapsto \hat{\mu}(Z)$ satisfies a weak differentiability condition, the \emph{df} may be computed via a divergence formula~\citep{stein1981estimation,LARS}:
\begin{equation}\label{eq-df-div1}
df (\hat{\mu}) = {\mathbb E} \left( \left(\nabla \cdot \hat{\mu}(Z)\right) \cdot (Z) \right),
\end{equation}
where $  \left(\nabla \cdot \hat{\mu}\right) \cdot (Z)   = \sum_{ij} \partial \hat{\mu}(Z_{ij})/\partial Z_{ij}.$
For the spectral thresholding operator $S_{\lambda, \gamma}(\cdot)$, expression~\eqref{eq-df-div1} holds if the map $Z \mapsto S_{\lambda, \gamma}(Z)$ is Lipschitz and hence weakly differentiable --- see for example,~\cite{candes2013unbiased}. In the light of Proposition~\ref{prop-liphs-1}, the map
$Z \mapsto S_{\lambda, \gamma}(Z)$ is Lipschitz when $\phi_P + 1 > 0$. Under the model~\eqref{iid-add-1}, the singular values of $Z$ will have a multiplicity of one with probability one. 
We assume that the univariate thresholding operators 
are differentiable, i.e., $s'_{\lambda,\gamma}(\cdot)$ exists. With these assumptions in place, the divergence formula for $S_{\lambda, \gamma}(Z)$ can be obtained following~\cite{candes2013unbiased, mazumder2019computing}, as presented in the following proposition. 
\begin{proposition}\label{propfour}
Assume that $1+\phi_P>0$ and the model~\eqref{iid-add-1} is in place. Then the degrees of freedom of the estimator $S_{\lambda, \gamma}(Z)$ is given by:
\begin{align}\label{df_explicit}
df (S_{\lambda, \gamma}(Z))=&{\mathbb E}\sum_i \Big (s'_{\lambda,\gamma}(\sigma_i)+|m-n|\frac{s_{\lambda,\gamma}(\sigma_i)}{\sigma_i} \Big)+ 2{\mathbb E}\sum_{i \neq j}\frac{\sigma_is_{\lambda,\gamma}(\sigma_i)}{\sigma^2_i-\sigma^2_j},
\end{align}
where the $\sigma_i$'s are the singular values of $Z$.
\end{proposition}
We note that the above expression is true for any value of $1+\phi_P>0$. For the MC+ penalty function, expression~\eqref{df_explicit} holds for $\gamma > 1$. As soon as $\gamma \leq 1$, the above method of deriving \emph{df} does not apply due to the discontinuity in the map $ Z \mapsto S_{\lambda, \gamma}(Z)$. Values of $\gamma$ close to one (but larger), however, give an expression for the \emph{df} near the hard-thresholding spectral operator, which corresponds to $\gamma = 1$.


\begin{figure}[htb!]
\begin{center}
\includegraphics[height= 0.25\textheight, trim =0cm .5cm 1cm 1cm, clip = true]{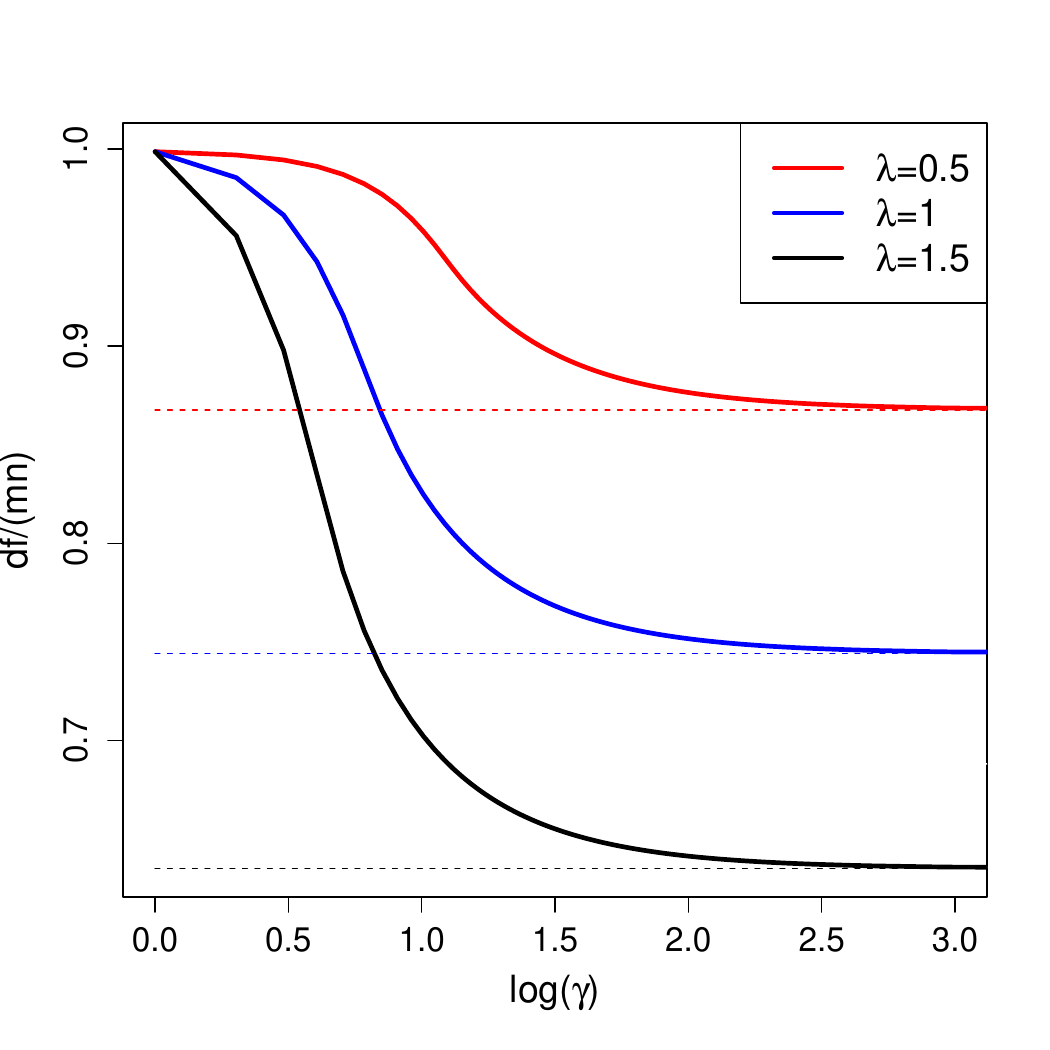}
\caption{\small{Figure showing the \emph{df} for the MC+ thresholding operator for a matrix with $m=n=10$, $\mu = 0$ and $\vv=1$. The \emph{df} profile as a function of $\gamma$ (in the log scale) is shown for three values of $\lambda$. The dashed lines correspond to the \emph{df} of the spectral soft-thresholding operator, corresponding to $\gamma=\infty$. We propose calibrating the $(\lambda, \gamma)$ grid to a $(\widetilde{\lambda},\widetilde{\gamma})$ grid such that the \emph{df} corresponding to every value of $\widetilde{\gamma}$ matches the \emph{df} of the soft-thresholding operator --- as shown in Figure~\ref{fig-df-2}.}}
\label{fig-df-1}
\end{center}
\end{figure}

To understand the behavior of the \emph{df} as a function of $(\lambda, \gamma)$, let us consider the null model with $\mu = 0$ and the MC+ penalty function. In this case, for a fixed $\lambda$ (see Figure~\ref{fig-df-1} with a fixed $\lambda>0$), 
the \emph{df} is seen to increase with smaller $\gamma$ values: the soft-thresholding function shrinks the large coefficients and sets all coefficients smaller than $\lambda$ to be zero; the more aggressive (closer to the hard thresholding operator) shrinkage operators ($s_{\lambda, \gamma}(\sigma)$) shrink less for larger values of $\sigma$ and set all coefficients smaller than $\lambda$ to zero. Thus, intuitively, the more aggressive thresholding operators should have larger \emph{df} since they do more ``fitting'' --- this is indeed observed  in Figure~\ref{fig-df-1}. \cite{MazumderEtal2011} studied the \emph{df} of the univariate thresholding operators in the linear regression problem, and observed a similar pattern in the behavior of the \emph{df} across $(\lambda, \gamma)$ values. For the linear regression problem, \cite{MazumderEtal2011} argued that it is desirable to choose a parametrization for $(\lambda, \gamma)$ such that for a fixed $\lambda$, as one moves across $\gamma$, the \emph{df} should be the same. We follow the same strategy for the spectral regularization problem considered herein --- we reparametrize a two-dimensional grid of $(\lambda, \gamma)$ values to a two-dimensional grid of $(\widetilde{\lambda},\widetilde{\gamma})$ values, such that the \emph{df} remain calibrated in the sense described above --- this is illustrated in Figure~\ref{fig-df-1} (see the horizontal dashed lines 
corresponding to the constant \emph{df} values, after calibration). Figure~\ref{fig-df-2} shows the lattice of $(\widetilde{\lambda},\widetilde{\gamma})$ after calibration. The values of $(\tilde{\lambda},\tilde{\gamma})$ on each curve induce the same \emph{df}. As $\tilde{\gamma}$ moves down (the penalty becomes more ``nonconvex"), the corresponding $\tilde{\lambda}$ (the shrinkage) has to increase to maintain the same \emph{df}.

\begin{figure}[htb!]
\begin{center}
\includegraphics[height= 0.31\textheight, trim = 0cm .5cm .1cm .1cm, clip = true]{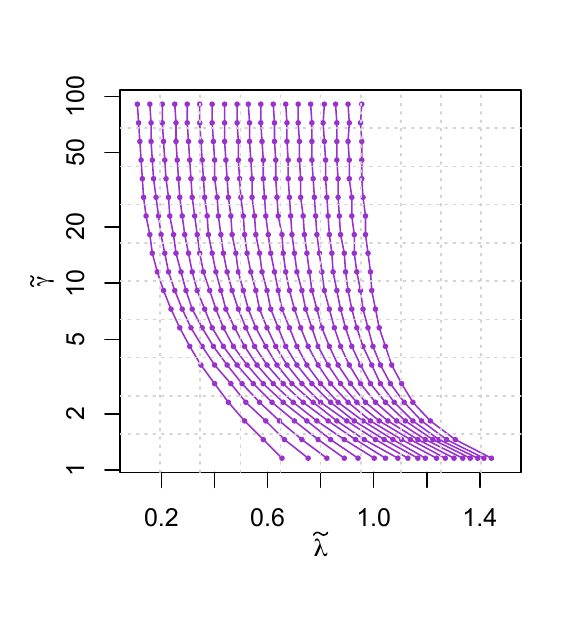}
\vspace{-0.4cm}
\caption{\small{Figure showing the calibrated $(\widetilde{\lambda}, \widetilde{\gamma})$ lattice --- for every fixed value of $\widetilde{\lambda}$, the \emph{df} of the MC+ spectral threshold operators are the same across different $\widetilde{\gamma}$ values. The \emph{df} computations have been performed on a null model using Proposition~\ref{propfive}.
}}
\label{fig-df-2}
\end{center}
\end{figure}

The study of \emph{df} presented herein provides a simple and intuitive explanation about the roles played by the parameters $(\lambda, \gamma)$
 for the fully observed problem. The notion of calibration provides a new parametrization of the family of penalties. 
 From a computational viewpoint, since the general algorithmic framework presented in this paper (see Section~\ref{sec3}) computes a regularization surface using warm-starts across adjacent $(\lambda, \gamma)$ values on a two-dimensional grid; it is desirable for the adjacent values to be close --- the \emph{df} calibration also ensures this in a simple and intuitive manner.


\paragraph{Computation of \emph{df}:} The \emph{df} estimate as implied by Proposition~\ref{propfour} depends only upon singular values (and not the singular vectors) of a matrix and can hence be computed with cost $O(\min \{m , n \}^2)$. The expectation can be approximated via Monte-Carlo simulation --- these computations are easy to parallelize and can be done offline. Since we compute the \emph{df} for the null model, for larger values of $m,n$ we recommend using the Marchenko-Pastur law for
iid Gaussian matrix ensembles to approximate the \emph{df} expression~\eqref{df_explicit}. We illustrate the method using the MC+ penalty for $\gamma >1$. Towards this end, let us define a function on $ \beta \geq 0$
\begin{equation*}
g_{\zeta, \gamma}(\beta)=
\begin{cases}
0, & \text{if } \sqrt{\beta} \leq \zeta \\
\frac{\gamma}{\gamma-1} (1- \frac{\zeta}{\sqrt{\beta}}), & \text{if } \zeta < \sqrt{\beta} \leq \zeta\gamma \\
1, & \text{if } \sqrt{\beta} > \zeta\gamma \,.
\end{cases}
\end{equation*}

For the following proposition, we will assume (for simplicity) that $m \geq n$.
\begin{proposition} \label{propfive}
Let $m, n \rightarrow \infty$ with $\frac{n}{m} \rightarrow \alpha \in (0, 1]$, then under the model $Z_{ij} \overset{\text{iid}}{\sim}N(0,1)$, we have
\begin{align*}\label{df_asymp}
\lim_{m,n\rightarrow \infty}\;\; \frac{df( S_{\lambda, \gamma}(Z))}{mn} = 
\begin{cases}
0 & \text{if }   \frac{\lambda}{\sqrt{m}} \rightarrow \infty \\
(1-\alpha) \mathbb{E} g_{\zeta, \gamma}(T_1) + \alpha \mathbb{E} \left( \frac{T_1g_{\zeta, \gamma}(T_1)-T_2g_{\zeta, \gamma}(T_2)}{T_1-T_2} \right) & \text{if } \frac{\lambda}{\sqrt{m}} \rightarrow \zeta \\
1 & \text{if } \frac{\lambda}{\sqrt{m}}\rightarrow 0 
\end{cases}
\end{align*}
where $S_{\lambda, \gamma}(Z)$ is the thresholding operator corresponding to the MC+ penalty with $\lambda \geq0, \gamma>1$ and
the expectation is taken with respect to $T_1$ and $T_2$ independently generated from the Marchenko-Pastur distribution (see Lemma~\ref{lemma3}, Section~\ref{app}).
\begin{proof}
For a proof, see Section~\ref{proof-prop-five}.
\end{proof}
\end{proposition}

Note that the variance $\vv^2$ in model~\eqref{iid-add-1} can always be assumed to be one (by adjusting the value of the tuning parameter accordingly\footnote{This follows from the simple observation that $s_{a\lambda, \gamma}(ax)=a s_{\lambda, \gamma}(x)$ and $s'_{a\lambda, \gamma}(ax)=s'_{\lambda, \gamma}(x)$.}).

%


\section{The \NSI~Algorithm} \label{sec3}

In this section, we present algorithm \NSI. The algorithm is inspired by an EM-stylized procedure, similar to \SI~\citep{MazumderEtal2010}, but has important 
innovations, as we will discuss shortly. It is helpful to recall that, for 
observed data: $\mathcal{P}_{\Omega}(Y)$, the algorithm \SI~relies on the following update sequence
\begin{equation}\label{si-update-2}
X_{k+1} = S_{\lambda, \ell_{1}} \left (  \mathcal{P}_{\Omega}(Y) + \mathcal{P}_{\Omega}^\perp(X_{k})\right),
\end{equation}
which can be interpreted as computing the nuclear norm regularized spectral thresholding operator for the following ``fully observed'' problem:
\begin{equation*}
\begin{aligned}
X_{k+1} \in \argmin_{X}  \left \{ \frac12 \left\| X -  \left (  \mathcal{P}_{\Omega}(Y) + \mathcal{P}_{\Omega}^\perp(X_{k})\right) \right\|_{F}^2 + \lambda \| X \|_{*}\right \},
\end{aligned}
\end{equation*}
where, the missing entries are filled in by the current estimate, i.e., $\mathcal{P}_{\Omega}^\perp(X_{k})$. We refer the reader to~\cite{MazumderEtal2010} for a detailed study of the algorithm. 
\cite{MazumderEtal2010} suggest, in passing, the notion of extending \SI~to more general thresholding operators; however, such generalizations were not pursued by the authors. In this paper, we present a thorough investigation about nonconvex generalized thresholding operators --- we study their convergence properties, scalability aspects and demonstrate their superior statistical performance across a wide range of numerical experiments. 

Update~\eqref{si-update-2} suggests a natural generalization to more general nonconvex penalty functions, by simply replacing the  spectral soft thresholding operator $S_{\lambda, \ell_{1}}(\cdot)$ with more general spectral operators~$S_{\lambda, \gamma}(\cdot)$:
\begin{equation}\label{nsi-update-2}
X_{k+1}  =  S_{\lambda, \gamma} \left (  \mathcal{P}_{\Omega}(Y) + \mathcal{P}_{\Omega}^\perp(X_{k})\right).
\end{equation}
While the above update rule works quite well in our numerical experiments, it enjoys limited computational guarantees, as suggested by our convergence analysis in Section~\ref{sec:conv-ana}. We thus propose and study a seemingly minor generalization of the rule~\eqref{nsi-update-2} --- this modified rule enjoys superior finite time convergence rates to a first order stationary point. We develop our algorithmic framework below. 

Let us define the following function:
\begin{align}
F_{\ell}(X; X_{k}):= \frac12 \left\|  \mathcal{P}_{\Omega}(X - Y) \right\|_F^2+ \frac12\| \mathcal{P}_{\Omega}^\perp ( X - X_{k}) \|_{F}^2 + \frac{\ell}{2} \| X - X_{k} \|_{F}^2 + \sum_{i=1}^{\min \{ m, n \}} P(\sigma_i(X);\lambda,\gamma),\label{def-xkp1-1}
\end{align}
for $\ell  \geq 0$. Note that $F_{\ell}(X; X_{k})$ majorizes the objective function $f(X)$ defined in \eqref{nonconv-prob-1}, i.e., $ F_{\ell}(X; X_{k}) \geq f(X)$ for any $X$ and $X_{k}$, with equality holding at $X=X_{k}$. In an attempt to obtain a minimum of Problem~\eqref{nonconv-prob-1}, we propose to iteratively minimize $F_{\ell}(X; X_{k})$, an upper bound to $f(X)$, to obtain $X_{k+1}$ --- more formally, this leads to the following update sequence:
\begin{equation}\label{def-xk+1-ell}
X_{k+1} \in \argmin_{X} \; F_{\ell}(X; X_{k}).
\end{equation}
Note that $X_{k+1}$ is easy to compute; by some rearrangement of~\eqref{def-xkp1-1} we see:
\begin{equation}\label{def-Xk+1-S}
\begin{aligned}
X_{k+1} \in \argmin_{X}\;  \underbrace{\frac{\ell+1}{2} \| X - \widetilde{X}_{k} \|_{F}^2 + \sum_{i=1}^{\min \{ m, n \}}
P(\sigma_{i}(X); \lambda,\gamma)}_{:=S^{\ell}_{\lambda, \gamma} \left( \widetilde{X}_{k} \right)},
\end{aligned}
\end{equation}
where $\widetilde{X}_{k} = \left( \mathcal{P}_{\Omega}(Y) + \mathcal{P}_{\Omega}^\perp(X_{k}) + \ell X_{k}\right)/(\ell+1)$. Note that~\eqref{def-Xk+1-S} is a minor modification of~\eqref{nsi-update-2} --- in particular, if $\ell = 0$, then these two update rules coincide.

The sequence $X_{k}$ defined via~\eqref{def-Xk+1-S} has desirable convergence properties, as we discuss in Section~\ref{sec:conv-ana}. In particular, as $k \rightarrow \infty$, the sequence reaches (in a sense that will be made more precise later) a first order stationary point for Problem~\eqref{nonconv-prob-1}. We also provide a finite time 
convergence analysis of the update sequence~\eqref{def-Xk+1-S}.

We intend to compute an entire regularization surface of solutions to Problem~\eqref{nonconv-prob-1} over a two-dimensional grid of $(\lambda, \gamma)$-values, using warm-starts. We take the MC+ family of functions as a running example, with $(\lambda, \gamma) \in \{ \lambda_1 > \lambda_2 > \dots > \lambda_N\} \times \{\infty := \gamma_{1} > \gamma_{2} > \ldots > \gamma_{M} \}$. At the beginning, we compute a path of solutions for the nuclear norm penalized problem, i.e., Problem~\eqref{nonconv-prob-1} with $\gamma= \infty$ on a grid of $\lambda$ values. For a fixed value of $\lambda$, we compute solutions to Problem~\eqref{nonconv-prob-1} for smaller values of $\gamma$, gradually moving away from the convex problems. In this continuation scheme, we found the following strategies useful: 
\begin{itemize}
\item For every value of $(\lambda_{i}, \gamma_{j})$, we apply two copies of the iterative scheme~\eqref{def-xk+1-ell} initialized with solutions obtained from
its two neighboring points $(\lambda_{i-1},\gamma_j)$ and  $(\lambda_i,\gamma_{j-1})$. From these two candidates, we select the one that 
leads to a smaller value of the objective function $f(\cdot)$ at $(\lambda_{i}, \gamma_{j})$.
\item Instead of using a two-dimensional rectangular lattice, one can also use the recalibrated lattice, suggested in Section~\ref{sec-edf-spec-ops}, as the two-dimensional grid of tuning parameters. 
\end{itemize}
The algorithm outlined above, called \NSI~is summarized as Algorithm~\ref{aone}.

\begin{algorithm}[htb!]
\caption{\textsc{NC-Impute}} \label{aone}
\begin{itemize}
\item[1.] Input: A search grid $\,\lambda_1 > \lambda_2 > \dots > \lambda_N; \; +\infty:= \gamma_1 > \gamma_2 > \dots > \gamma_M$. Tolerance $\varepsilon$.
\item[2.] Compute solutions $\hat{X}_{\lambda_i,\gamma_1}$ for $i=1, \ldots, N$, for the nuclear norm regularized problem.
\item[3.] For every $(\gamma, \lambda) \in \{ \gamma_{2}, \ldots, \gamma_{M}\} \times \{ \lambda_{1}, \ldots, \lambda_N\} $:
\begin{itemize}
\item[(a)] Initialize $$X^{\text{old}}=\arg\min\limits_{X}  \left \{ f(X), X \in \left \{ \hat{X}_{\lambda_{i-1},\gamma_j}, \; \hat{X}_{\lambda_i,\gamma_{j-1}} \right\}  \right \}.$$
\item[(b)] Repeat until convergence, i.e., $\|X^{\text{new}}-X^{\text{old}}\|^2_F < \varepsilon  \|X^{\text{old}}\|^2_F$:
\begin{itemize}
\item[(i)] Compute $X^{\text{new}} \in \argmin\limits_{X} \; F_{\ell}(X; X^{\text{old}}).$ 
\item[(ii)]Assign $X^{\text{old}} \leftarrow X^{\text{new}}$.
\end{itemize}
\item[(c)] Assign $\hat{X}_{\lambda_i,\gamma_j} \leftarrow X^{\text{new}}$.
\end{itemize}
\item[4.] Output: $\hat{X}_{\lambda_i,\gamma_j}$ for $i =1, \ldots,N$, $j = 1, \ldots, M$.
\end{itemize}
\end{algorithm}

We now present an elementary convergence analysis of the update sequence~\eqref{def-Xk+1-S}. Since the problems under investigation herein are nonconvex, our analysis requires new ideas and techniques beyond those used in~\cite{MazumderEtal2010} for the convex nuclear norm regularized problem.


\subsection{Convergence Analysis}\label{sec:conv-ana}

By the definition of $X_{k+1}$ we have that:
\begin{equation*}\label{ineq-ell-1}
F_{\ell}(X_{k+1}; X_{k}) = \min_{X} \; F_{\ell}(X; X_{k}) \leq  F_{\ell}(X_{k}; X_{k}) = f(X_{k}).
\end{equation*}
Let us define the quantities:
\begin{equation*}\label{mu-ell-defn}
\nu(\ell):=1 + \phi_P + \ell \;\;\; \text{ and } \;\;\; \nu^{\dagger}(\ell) := \max \left \{ \nu(\ell), 0  \right\},
\end{equation*}
where, if $\nu(\ell) \geq 0$, then the function $ X \mapsto F_{\ell}(X; X_{k})$ is $\nu(\ell)$-strongly convex. In particular, from~\eqref{def-xk+1-ell}, it follows that $\nabla F_{\ell}(X_{k+1}; X_{k})$, a subgradient of the map $ X \mapsto F_{\ell}(X; X_{k})$ (evaluated at $X_{k+1}$) equals zero. We thus have:
\begin{equation}\label{str-conv-ineq-0}
\begin{aligned}
F_{\ell}(X_{k}; X_{k}) - F_{\ell}(X_{k+1}; X_{k})  \geq \frac{\nu(\ell)}{2} \|  X_{k+1} - X_{k} \|_{F}^2.
\end{aligned}
\end{equation}
Now note that, by the definition of $X_{k+1}$, we always have: $F_{\ell}(X_{k}; X_{k}) \geq F_{\ell}(X_{k+1}; X_{k}),$ which combined with~\eqref{str-conv-ineq-0} 
leads to (replacing $\nu(\ell)$ by $\nu^\dagger(\ell)$):
\begin{equation}\label{str-conv-ineq-1}
\begin{aligned}
F_{\ell}(X_{k}; X_{k}) - F_{\ell}(X_{k+1}; X_{k})  \geq \frac{\nu^{\dagger}(\ell)}{2} \|  X_{k+1} - X_{k} \|_{F}^2.
\end{aligned}
\end{equation}
In addition, we have:
\begin{align}
F_{\ell}(X_{k+1}; X_{k}) =& \frac12 \left\|  \mathcal{P}_{\Omega}(X_{k+1} - Y) \right\|_F^2+ \sum_{i=1}^{\min \{ m, n \}} P(\sigma_i(X_{k+1});\lambda,\gamma) \nonumber \\
& +    \frac12\| \mathcal{P}_{\Omega}^\perp ( X_{k+1} - X_{k}) \|_{F}^2 +  \frac{\ell}{2} \| X_{k+1} - X_{k} \|_{F}^2 \nonumber \\
= & f(X_{k+1}) + \frac12\| \mathcal{P}_{\Omega}^\perp ( X_{k+1} - X_{k}) \|_{F}^2 +  \frac{\ell}{2} \| X_{k+1} - X_{k} \|_{F}^2. \label{line-bb-1}
\end{align}
Combining~\eqref{str-conv-ineq-1} and~\eqref{line-bb-1}, and observing that $F_{\ell}(X_{k};X_{k}) = f(X_k)$, we have:
\begin{align}
\label{suff-decrease-1}
f(X_{k}) - f(X_{k+1}) \geq&  \frac{\nu^{\dagger}(\ell)}{2} \|  X_{k+1} - X_{k} \|_{F}^2 +  \frac{\ell}{2} \| X_{k+1} - X_{k} \|_{F}^2 + \frac12\| \mathcal{P}_{\Omega}^\perp ( X_{k+1} - X_{k}) \|_{F}^2 \nonumber \\
= & \underbrace{\frac{\nu^{\dagger}(\ell) + \ell}{2} \|  X_{k+1} - X_{k} \|_{F}^2 + \frac12\| \mathcal{P}_{\Omega}^\perp ( X_{k+1} - X_{k}) \|_{F}^2}_{:= \Delta_{\ell} (X_{k}; X_{k+1})}.
\end{align}
Since $\Delta_{\ell} (X_{k}; X_{k+1}) \geq 0$, the above inequality immediately implies that $ f(X_{k}) \geq f(X_{k+1})$ for all $k$; and the
improvement in objective values is at least as large as the quantity $\Delta_{\ell} (X_{k}; X_{k+1})$. 
The term $\Delta_{\ell} (X_{k}; X_{k+1})$ is a measure of progress of the algorithm, as formalized by the following proposition.
\begin{proposition}\label{prop-fixed-pt1}
{\bf {(a)}}: Let $\nu^{\dagger}(\ell) + \ell >0$ and for any $X_{a}$, let us consider the update $X_{a+1} \in \argmin_{X} F_{\ell}(X; X_{a})$. Then the following are equivalent: 
\begin{itemize}
\item[(i)] $f(X_{a+1}) = f(X_{a})$
\item[(ii)]  $\Delta_{\ell} (X_{a}; X_{a+1})=0$
\item[(iii)] $X_{a}$ is a fixed point, i.e., $X_{a+1} = X_{a}$.
\end{itemize}

{\bf {(b)}}: If $\nu^{\dagger}(\ell), \ell=0$ and $\Delta_{\ell} (X_{a}; X_{a+1})=0$ then
$X_{a+1}$ is a fixed point.
\end{proposition}
\begin{proof}
Proof of Part (a):\\
We will show that (i) $\implies$ (ii) $\implies$ (iii) $\implies$ (i); by analyzing~\eqref{suff-decrease-1}. If $f(X_{a+1}) = f(X_{a})$ then $\Delta_{\ell} (X_{a}; X_{a+1})=0$. Since $\nu^{\dagger}(\ell)+\ell>0$, we have that  $X_{a+1} = X_{a}$, which trivially implies (i).

Proof of Part (b):\\
If $\nu^{\dagger}(\ell) + \ell = 0,$ Part (a) needs to be slightly modified. Note that $\Delta_{\ell} (X_{a}; X_{a+1})=0$ iff $\mathcal{P}_{\Omega}^\perp(X_{a+1}) = \mathcal{P}_{\Omega}^\perp(X_{a})$. Since $\ell = 0$, we 
have that $X_{a+2} = S_{\lambda, \gamma} \left(  \mathcal{P}_{\Omega}(Y) + \mathcal{P}_{\Omega}^\perp(X_{a+1}) \right)$. The condition $\mathcal{P}_{\Omega}^\perp(X_{a+1}) = \mathcal{P}_{\Omega}^\perp(X_{a}),$ implies 
that $$S_{\lambda, \gamma} (  \mathcal{P}_{\Omega}(Y) + \mathcal{P}_{\Omega}^\perp(X_{a+1})) = S_{\lambda, \gamma} \left( \mathcal{P}_{\Omega}(Y) + \mathcal{P}_{\Omega}^\perp(X_{a}) \right),$$ where the term on the right 
equals $X_{a+1}$. Thus, $X_{a+1}= X_{a+2} = \cdots$, i.e., $X_{a+1}$ is a fixed point. 
\end{proof}


Since the $f(X_{k})$'s form a decreasing sequence which is bounded from below, they converge to $\hat{f}$, say --- this implies that $\Delta_{\ell} (X_{k}; X_{k+1}) \rightarrow 0$ as $k \rightarrow \infty$. Let us now consider two cases, depending upon the value of $\nu^{\dagger}(\ell) + \ell$. If $\nu^{\dagger}(\ell) + \ell > 0$, then we have $X_{k+1} - X_{k} \rightarrow 0$ as $k \rightarrow \infty$. On the other hand, if the quantities $\nu^{\dagger}(\ell)=0,\ell=0$, the conclusion needs to be modified: $\Delta_{\ell} (X_{k}; X_{k+1}) \rightarrow 0$ implies that $\mathcal{P}_{\Omega}^\perp (X_{k+1} - X_{k}) \rightarrow 0$ as $k \rightarrow \infty$.

Motivated by  the above discussion, we make the following definition of a first order stationary point for Problem~\eqref{nonconv-prob-1}.
\begin{mydef}\label{def-fo-pt1}
$X_{a}$ is said to be a first order stationary point for Problem~\eqref{nonconv-prob-1} if $\Delta_{\ell}(X_{a}; X_{a+1})=0$. $X_{a}$ is said to be an $\epsilon$-accurate first order stationary point for Problem~\eqref{nonconv-prob-1} if $\Delta_{\ell}(X_{a}; X_{a+1}) \leq \epsilon$.
\end{mydef}

\begin{proposition}\label{conv-rate-1}
The sequence $f(X_k)$ is decreasing and suppose it converges to $\hat{f}$. Then the rate of convergence of $X_k$ to this first order stationary point is given by:
\begin{equation}\label{rate-of-conv1}
\min\limits_{1\leq k \leq {\mathcal K}} \Delta_{\ell} (X_{k}; X_{k+1}) \leq \frac{1}{\mathcal K} \left(f(X_{1}) - \hat{f}\right).
\end{equation}
\end{proposition}
\begin{proof}
The arguments presented preceding Proposition~\ref{conv-rate-1} establish that the sequence $f(X_k)$ is decreasing and converges to $\hat{f}$, say. Consider~\eqref{suff-decrease-1} for any $1 \leq k \leq {\mathcal K}$. We have that $ \Delta_{\ell}(X_{k}; X_{k+1}) \leq f(X_{k}) - f(X_{k+1})$ --- summing this inequality for $k=1, \ldots, {\mathcal K}$ we obtain:
\begin{align*}
{\mathcal K}\min\limits_{1\leq k \leq {\mathcal K}} \Delta_{\ell} (X_{k}; X_{k+1}) \leq \sum_{1\leq k\leq \mathcal K} \Delta_{\ell} (X_{k}; X_{k+1})  \leq f(X_1) - f(X_{\mathcal K + 1}) \leq f(X_{1}) - \hat{f},
\end{align*}
where in the last inequality we used the simple fact that $f(X_{k}) \downarrow \hat{f}$. Gathering the left and right parts of the above chain of inequalities leads to~\eqref{rate-of-conv1}.
\end{proof}

Proposition~\ref{conv-rate-1} shows that the sequence $X_{k}$ reaches an $\epsilon$-accurate first order stationary point within 
$K_{\epsilon} =  (f(X_{1}) - \hat{f})/\epsilon$ many iterations. The number of iterations $K_{\epsilon}$, depends upon how close the initial estimate 
$f(X_{1})$ is to the eventual solution $\hat{f}$. Since \NSI~employs warm-starts, the constant appearing in the rhs of~\eqref{rate-of-conv1} suggests that the number of iterations required to a reach 
an approximate first order stationary point is quite low --- this is indeed observed in our experiments, and this feature of using warm-starts makes our algorithm particularly attractive
from a practical viewpoint. 


\subsubsection{Rank Stabilization}
\label{rankstable:sec}
Let us consider the thresholding function $S^{\ell}_{\lambda, \gamma}(\widetilde{X}_{k})$ defined in~\eqref{def-Xk+1-S}, which expresses $X_{k+1}$ as a function of $X_{k}$.
Using the development in Section~\ref{sec2}, it is easy to see that the spectral operator $S^{\ell}_{\lambda, \gamma}(\widetilde{X}_{k})$ is closely tied to the
following vector thresholding operator~\eqref{tilde-xk-spe-1}, acting on the singular values of $\widetilde{X}_{k}$. Formally, for a given nonnegative vector $\widetilde{\M{x}}$, if we denote:
\begin{equation}\label{tilde-xk-spe-1}
\begin{aligned}
s^{\ell}_{\lambda, \gamma}(\widetilde{\M{x}}) \in& \argmin_{\B\alpha \M{\geq} \M{0} } \bigg \{  \frac{\ell+1}{2} \| \B\alpha - \widetilde{\M{x}} \|_{2}^2 +\sum_{i=1}^{\min \{ m, n \}} P(\alpha_{i}; \lambda,\gamma) \bigg\},
\end{aligned}
\end{equation}
then $$S^{\ell}_{\lambda, \gamma}(\widetilde{X}) = \widetilde{U} \diag(s^{\ell}_{\lambda, \gamma}(\widetilde{\M{x}}) )\widetilde{V}',$$ where $\widetilde{X}= \widetilde{U} \diag(\widetilde{\M{x}}) \widetilde{V}'$ is the SVD of $\widetilde{X}$. Thus, properties of the thresholding function $S^{\ell}_{\lambda, \gamma}(\widetilde{X})$ are closely related to those of the vector thresholding operator $s^{\ell}_{\lambda, \gamma}(\widetilde{\M{x}})$. Due to the separability of the vector thresholding operator $s^{\ell}_{\lambda, \gamma}(\widetilde{\M{x}})$, across each coordinate of $\widetilde{\M{x}}$, we denote by $s^{\ell}_{\lambda,\gamma}(\widetilde{x}_{i})$, the $i$th coordinate of $s^{\ell}_{\lambda, \gamma}(\widetilde{\M{x}})$.

We now investigate what happens to the rank of the sequence $X_k$ as defined via~\eqref{def-xk+1-ell}.
In particular, does this rank converge? We show that the rank stabilizes after finitely many iterations, under an additional assumption --- namely 
the spectral thresholding operator is discontinuous --- see Figure~\ref{fig1} for examples of discontinuous thresholding functions.

\begin{proposition}\label{prop-rank-stab-1}
Consider the update sequence $$X_{k+1} = S^{\ell}_{\lambda, \gamma}(\widetilde{X}_{k})$$ as defined in~\eqref{def-Xk+1-S}; and let $\nu^{\dagger}(\ell) + \ell >0$.
Suppose that there is a $\lambda_{S} >0$ such that, for any scalar $\widetilde{x}\geq0$, the following holds:
$s^{\ell}_{\lambda, \gamma}(\widetilde{x}) \neq 0 \imply |s^{\ell}_{\lambda, \gamma}(\widetilde{x})| > \lambda_{S}$ --- 
i.e., the scalar thresholding operator $\widetilde{x} \mapsto s^{\ell}_{\lambda, \gamma}(\widetilde{x})$ is discontinuous. Then there exists an integer ${\mathcal K}^*$ such that for all $k \geq {\mathcal K}^*$, we have $\rnk(X_{k}) = r$, i.e., the rank stabilizes after finitely many iterations.
\end{proposition}
\begin{proof}
Using~\eqref{suff-decrease-1} it follows that
\begin{align*}
f(X_{k}) - f(X_{k+1}) &\geq \frac{\nu^{\dagger}(\ell) + \ell}{2} \| X_{k+1} - X_{k}\|_{F}^2 \geq \frac{\nu^{\dagger}(\ell) + \ell}{2} \| \B\sigma_{k+1} - \B\sigma_{k} \|_{2}^2,
\end{align*}
where the last inequality follows from Wielandt-Hoffman inequality~\citep{horn2012matrix} and $\B\sigma_{k}:= \B\sigma(X_{k})$ denotes the vector of singular values of $X_{k}$. Let $\mathbbm{1}(\B\sigma)$ be an indicator vector with $i$th coordinate being equal to $1(\sigma_{i}\neq 0)$. We will prove the result of rank stabilization via the method of contradiction. Suppose the rank does \emph{not} stabilize, then $\mathbbm{1}(\B\sigma_{k+1}) \neq \mathbbm{1}(\B\sigma_{k})$ for infinitely many $k$ values. Thus there are infinitely many $k'$ values such that:
$$ \| \B\sigma_{k'+1} - \B\sigma_{k'}\|_{2}^2  \geq \sigma^2_{k'+1,i}\,,$$
where $i$ is taken such that $\sigma_{k'+1,i} \neq 0$ but $ \sigma_{k',i}=0$.
Note that by the property of the thresholding function $s^{\ell}_{\lambda, \gamma}(\cdot)$ we have that $s^{\ell}_{\lambda, \gamma}(\widetilde{x}) \neq 0 \imply 
|s^{\ell}_{\lambda, \gamma}(\widetilde{x})| > \lambda_{S}$. This implies that $\| \B\sigma_{k'+1} - \B\sigma_{k'}\|_{2}^2 \geq \lambda^2_{S}$ for infinitely many $k'$ values, which is a contradiction to the convergence: $f(X_{k+1}) - f(X_{k}) \rightarrow 0$. Thus the support of $\B\sigma(X_{k})$ converges, and necessarily after finitely many iterations --- leading to the existence of an iteration number ${\mathcal K}^*$, after which the rank of $X_{k}$ remains fixed. This completes the proof of the proposition.
\end{proof}

\begin{remark}
If $\ell =0$, the discontinuity of the thresholding operator $s_{\lambda, \gamma}(\cdot)$ (as demanded by Proposition~\ref{prop-rank-stab-1})
 occurs for the MC+ penalty function as soon as $\gamma \leq 1$. For a general $\ell>0$, discontinuity in $s^\ell_{\lambda, \gamma}(\cdot)$ occurs as soon as 
 $\gamma \leq \frac{1}{\ell+1}$. 
\end{remark}


\subsubsection{Subspace Stabilization}\label{sec:subsp-stab-1}

We study herein, the properties of the left and right singular subspaces associated with the sequence $X_{k}$. The \emph{stabilization} of subspaces has
important implications in the main bottleneck of the \NSI~algorithm, i.e., the SVD computations --- we discuss this in further detail in Section~\ref{sec:compute-th-op1}.
The study of singular subspace stabilization requires subtle analysis based on matrix perturbation theory~\citep{stewart1990matrix}, since the left (and right) singular subspace, corresponding to the top $r$ singular values of a matrix is not a continuous function of the matrix argument.

Towards this end, we first recall a standard notion of distance between two subspaces (with same dimension) in terms of canonical angles.
\begin{mydef}\label{def-subsp-dist-2}
Let $S_{1} \in \mathbb{R}^{m \times \ell}$ and $S_{2} \in \mathbb{R}^{m \times \ell}$ be two orthonormal matrices and let us define $S_{1}^\perp$ such that
$[S_{1}, S_{1}^\perp]$ forms an orthonormal basis for $\mathbb{R}^{m}$.
The canonical angles between these two subspaces denoted by the vector $\Theta(S_1,S_2)$ are defined as:
$$\Theta(S_1,S_2) : = \sin^{-1} \left(\sigma_1(\texttt{X}),\ldots,\sigma_\ell(\texttt{X})\right),$$
where,  $\sigma_{i}(\texttt{X}), i \leq \ell$ are the singular values of the matrix  $\texttt{X}:=(S_{1}^\perp)'S_{2}$.
\end{mydef}

We now present a result regarding perturbation of singular subspaces, taken from~\cite{stewart1990matrix}. Before stating the proposition, we introduce some notation.
Let $U_{1} \in \mathbb{R}^{m \times r_{1}}$ ($V_{1}\in \mathbb{R}^{n \times r_{1}}$) denote a matrix of the $r_{1}$ left singular vectors (respectively, right) of a matrix $A$ --- with $\Sigma_{1}$ being a diagonal matrix of the corresponding 
top $r_{1}$ singular values. Similarly, we use the notation $\widetilde{U}_{1}, \widetilde{V}_{1},\widetilde{\Sigma}_{1}$ to denote the triplet of left and right singular vectors and singular values (corresponding to the top $r_1$ singular values) for a matrix $\widetilde{A}$.
We use the following matrices
\begin{equation*}\label{Defn-R1}
R= A\widetilde{V}_{1}  - \widetilde{U}_{1}\widetilde{\Sigma}_{1} , ~~~~~~~ Q = A'\widetilde{U}_{1} - \widetilde{V}_{1}\widetilde{\Sigma}_{1},
\end{equation*}
to measure a notion of proximity between $A$ and $\widetilde{A}$. The distance between the left (and also right) singular subspaces (corresponding to the top $r_1$ singular values) of 
$A$ and $\widetilde{A}$ may be measured by the following quantity: 
\begin{align}\label{ineq-rho-r}
\rho_{r_1} (A, \widetilde{A}) :=  \max \left \{ \left\|\sin \left (\Theta(U_{1}, \widetilde{U}_{1} )\right) \right\|_2 , \left\|\sin \left(\Theta(V_{1}, \widetilde{V}_{1} ) \right) \right\|_2  \right \}, 
\end{align}
where, the notation $\| A \|_{2}$  denotes the spectral norm of $A$.
With the above notations in place, we present the following proposition~\citep{stewart1990matrix} regarding the perturbation of singular subspaces of matrices.
\begin{proposition}\label{mat-pert-svd-2}
Suppose there exists $\alpha, \delta >0$ such that 
$$\min (\widetilde{\Sigma}_{1}) \geq \alpha + \delta , \;\; \text{and} \; \; \max \left(\Sigma_{2} \right) \leq \alpha,$$
where, $\Sigma_{2}$ is a diagonal matrix with the remaining singular values of $A$.
Then,
\begin{equation*}\label{ineq-rho-r1}
\rho_{r_1} (A, \widetilde{A}) \leq  \max \left \{  \| R \|_2 , \| Q\|_2   \right \}/\delta.
\end{equation*}
\end{proposition}

The above proposition informs us about the proximity of the left (and also right) singular subspaces across successive iterates $X_k$, as presented in the following proposition:
\begin{proposition}\label{prop-subsp-stab1}
Suppose $\nu^\dagger(\ell) + \ell >0$ and let 
$$\delta_{k,p} =   \sigma_{p+1}(X_{k}) -  \sigma_{p}(X_{k+1}),$$  for $1\leq p \leq \min \{ m, n \}$.
If $\liminf_{k \rightarrow \infty} \delta_{k,p}>0$ then 
$\rho_{p}(X_{k}, X_{k+1}) \rightarrow 0$ as $k \rightarrow \infty$. 
\end{proposition}
\begin{proof}
The proof is presented in Section~\ref{proof-prop-subsp-stab1}.
\end{proof}
\begin{remark}
Let the assumptions of Proposition~\ref{prop-rank-stab-1} be in place
-- this implies that there exists an integer ${\mathcal K}^*$ such that 
$$\rnk(X_k)=r, \quad \mbox{~~for all~~} k \geq {\mathcal K}^*. $$
Hence, in particular, there is a
separation between  $\sigma_{r}(X_k)$ and $\sigma_{r+1}(X_{k+1})$ for all $k$ sufficiently large. 
This implies that $\rho_{r}(X_{k}, X_{k+1}) \rightarrow 0$ as $k \rightarrow \infty$, i.e., in words: the distance between the left (and right) singular subspaces corresponding to the 
top $r$ singular values of $X_{k}$ and $X_{k+1}$ converges to zero, as $k \rightarrow \infty$.
\end{remark}

\subsubsection{Asymptotic Convergence.}
We now investigate the asymptotic convergence properties of the sequence $X_{k}, k \geq 1.$ Proposition~\ref{prop-rank-stab-1} shows that under suitable assumptions, 
the sequence $\rnk(X_{k}), k \geq 1$ converges. The existence of a limit point of $X_{k}$ is guaranteed if the singular values of 
$\B\sigma(X_k)$ remain bounded. It is not immediately clear whether the sequence $\B\sigma(X_k)$ will remain bounded since 
 several spectral penalty functions (like the MC+ penalty) are bounded\footnote{Due to the boundedness of the penalty function, the boundedness of the objective function does not necessarily imply that the sequence $\B\sigma(X_k)$ will remain bounded.}. 
 We address herein, the existence of a limit point of the sequence $\B\sigma(X_k)$, and hence the sequence $X_k$.

For the following proposition, we will assume that the concave penalty function $\sigma \mapsto P(\sigma; \lambda, \gamma)$ on $\sigma \geq 0$ is differentiable and the gradient is bounded.
\begin{proposition}\label{prop-asympt-conv}
Let $U_{k} \diag(\B\sigma_{k}) V_{k}'$ denote the rank-reduced SVD of $X_{k}$. Let $\bar{U}_{ m \times r}, \bar{V}_{m \times r}$ denote a limit point of the sequence $\{U_{k}, V_{k} \}, k \geq 1$, such that $(U_{n_k}, V_{n_k}) \rightarrow (\bar{U}, \bar{V})$ along a subsequence $n_{k} \rightarrow \infty$.
Let $\bar{u}_{i}$ denote the $i$th column of $\bar{U}$ (and similarly for $\bar{v}_{i}, \bar{V}$)  and let us denote $\bar{\Theta} = [ \text{vec}(\mathcal{P}_{\Omega}(\bar{u}_{1}\bar{v}_{1}')), \ldots,   \text{vec}(\mathcal{P}_{\Omega}(\bar{u}_{r}\bar{v}_{r}'))]$. We have the following:
\begin{itemize}
\item[(a)] If $\rnk(\bar{\Theta})=r$, then the sequence $X_{n_k}$ has a limit point which is a first order stationary point.
\item[(b)] If $\lambda_{\min}(\bar{\Theta}' \bar{\Theta}) + \phi_P > 0$, then the sequence $X_{n_k}$ converges to a first order stationary point: $\bar{X}=\bar{U} \diag(\bar{\B\sigma})\bar{V}'$, where $\B\sigma_{n_k} \rightarrow \bar{\B\sigma}$.
\end{itemize}
\begin{proof}
See Section~\ref{pf-asympt-conv}
\end{proof}
\end{proposition}

Proposition~\ref{prop-rank-stab-1} describes sufficient conditions under which the rank of the sequence $X_{k}$ stabilizes after finitely many iterations ---
it does \emph{not} describe the boundedness of the  sequence $X_k$, which is addressed in Proposition~\ref{prop-asympt-conv}. 
Note that Proposition~\ref{prop-asympt-conv} does not imply that the rank of the sequence $X_{k}$ stabilizes after finitely many iterations (recall that Proposition~\ref{prop-asympt-conv} does not assume that the thresholding operators are discontinuous, an assumption required by Proposition~\ref{prop-rank-stab-1}).

\subsection{Computing the Thresholding Operators}\label{sec:compute-th-op1}

The operator~\eqref{def-Xk+1-S} requires computing a thresholded SVD of the matrix $\widetilde{X}_{k}$, as demonstrated by Proposition~\ref{prop1}. The thresholded singular values 
$s^\ell_{\lambda, \gamma}( \cdot)$ as in~\eqref{tilde-xk-spe-1}
will have many zero coordinates due to the ``sparsity promoting'' nature of the concave penalty.
Thus, computing the thresholding operator~\eqref{def-Xk+1-S} will typically require performing a low-rank SVD on the matrix $\widetilde{X}_{k}$. While direct factorization based SVD methods can be used for smaller problems where $\min \{ m, n\}$ is of the order of a thousand or so; for larger matrices, such methods become computationally prohibitive --- we thus resort to iterative methods for computing low-rank SVDs for large scale problems. Algorithms such as the block power method; also known as block QR iterations, or those based on the Lanczos method~\citep{GVL83} are quite effective in computing the top few singular value and vectors of a matrix $A$, especially when the operations of multiplying ${A}{b}_{1}$ and ${A}'{b}_{2}$ (for vectors $b_{1},b_{2}$ of matching dimensions) can be done efficiently. Indeed, such matrix-vector multiplications turn out to be quite computationally attractive for our problem, since the computational cost of multiplying $\widetilde{X}_{k}$ and $\widetilde{X}_{k}'$ with vectors of matching dimensions is quite low. This is due to the structure of:
\begin{equation}\label{svd-tildex-k}
\begin{aligned}
\widetilde{X}_{k} = & \left( \mathcal{P}_{\Omega}(Y) + \mathcal{P}_{\Omega}^\perp(X_{k}) + \ell X_{k}\right)/(\ell+1)\\
= & \frac{1}{\ell+1}\underbrace{\mathcal{P}_{\Omega} (Y - X_{k})}_{\text{Sparse}} ~~+~~~ \underbrace{X_{k}}_{\text{Low-rank}},
\end{aligned}
\end{equation}
which admits a decomposition as the sum of a sparse matrix  and a low-rank matrix\footnote{We note that it is not guaranteed that the ${X}_k$'s will be of low-rank across the iterations of the algorithm for $k \geq 1$, even if they are eventually, for $k$ sufficiently large. However, in the presence of warm-starts across $(\lambda,\gamma)$ they are indeed, empirically, found to have low-rank as long as the regularization parameters are large enough to result in a small rank solution. Typically, as we have observed in our experiments, in the presence of warm-starts, the rank of $X_{k}$ is found to remain low across all iterations.}. Note that the sparse matrix has the same sparsity pattern as the observed indices $\Omega$. Decomposition~\eqref{svd-tildex-k} is inspired by a similar decomposition that was exploited effectively in the algorithm \SI \citep{MazumderEtal2010}, where the authors use PROPACK \citep{propack} to compute the low-rank SVDs.
In this paper, we use the Alternating Least Squares (ALS)-stylized procedure, which computes a low-rank SVD by solving the following nonlinear optimization problem:
\begin{equation}\label{als-prob-1}
 \mini_{U_{m \times \tilde{r}}, V_{n \times \tilde{r}}}\;\; \frac12 \| \widetilde{X}_{k} - U V' \|_{F}^2, 
 \end{equation}
using alternating least squares---this is in fact, equivalent to the block power method~\citep{GVL83}, in computing a rank $\tilde{r}$ SVD of the matrix $\widetilde{X}_{k}$. Across the iterations of \NSI, we pass the warm-start information in the $U,V$'s obtained from a low-rank SVD of $\widetilde{X}_{k}$ to compute the low-rank SVD for $\widetilde{X}_{k+1}$. Empirically, this warm-start strategy is found to be significantly more advantageous than a black-box low-rank SVD stylized approach, as used in the \SI~algorithm (for example), where, at every iteration, a new low-rank SVD is computed from scratch via PROPACK. This strategy quite naturally leads to a loss of useful information about the left and right singular vectors, 
which become closer to each other along the course of the \SI~iterations (as formalized by Section~\ref{sec:subsp-stab-1}).  Using warm-start information across successive iterations (i.e., $k$ values) leads to notable gains in computational speed (often reduces the total time to compute a family of solutions by orders of magnitude), when compared to black-box SVD stylized methods that do not rely on such warm-start strategies. This improvement is also supported by theory --- the computational guarantee of block power iterations~\citep{GVL83} states that the subspace spanned by the $U$ matrix (in the factorization $UV'$ in~\eqref{als-prob-1})
converges to that of the top $\tilde{r}$ left singular vectors at the rate: $C \gamma^q$, where, $q$ denotes the number of 
power iterations, $\gamma$ depends upon the ratio between the $\tilde{r}+1$ and $\tilde{r}$ singular values of the matrix $\widetilde{X}_{k}$; and $C$ depends upon the distance between: the initial estimate
 of (the subspace spanned by) $U$ and the 
 left top-$\tilde{r}$ set of singular vectors of $\tilde{X}_{k}$. The constant $C$ is smaller with a good warm-start, when compared to a random initialization.  A similar argument applies for the right set of singular vectors.


\section{Numerical Experiments}\label{sec4}
In this section, we present a systematic experimental study of the statistical properties of estimators obtained from~\eqref{nonconv-prob-1} for different choices of penalty functions. 
We perform our experiments on a wide array of synthetic and real data instances. Recall that the majority of the algorithmic guarantees proved in Section \ref{sec3} rely on the condition $\nu^{\dagger}(\ell)+\ell>0$. For MC+ penalty functions, it is straightforward to verify that $\nu^{\dagger}(0)>0$ as long as $\gamma \in (1,\infty]$. Hence we will use $\ell=0$ in \textsc{NC-Impute} throughout this section.


\subsection{Synthetic Examples}
\label{simulation:study}

\begin{figure*}[htb!]
\begin{center}
{\bf {Example-A}} (Low SNR, less missing entries) \\
\subfigure[ROM, $90$\% missing, $\text{SNR}=1$, $\text{true rank}=10$]{
\scalebox{.9}{\includegraphics[width=2.6in, height=2.in]{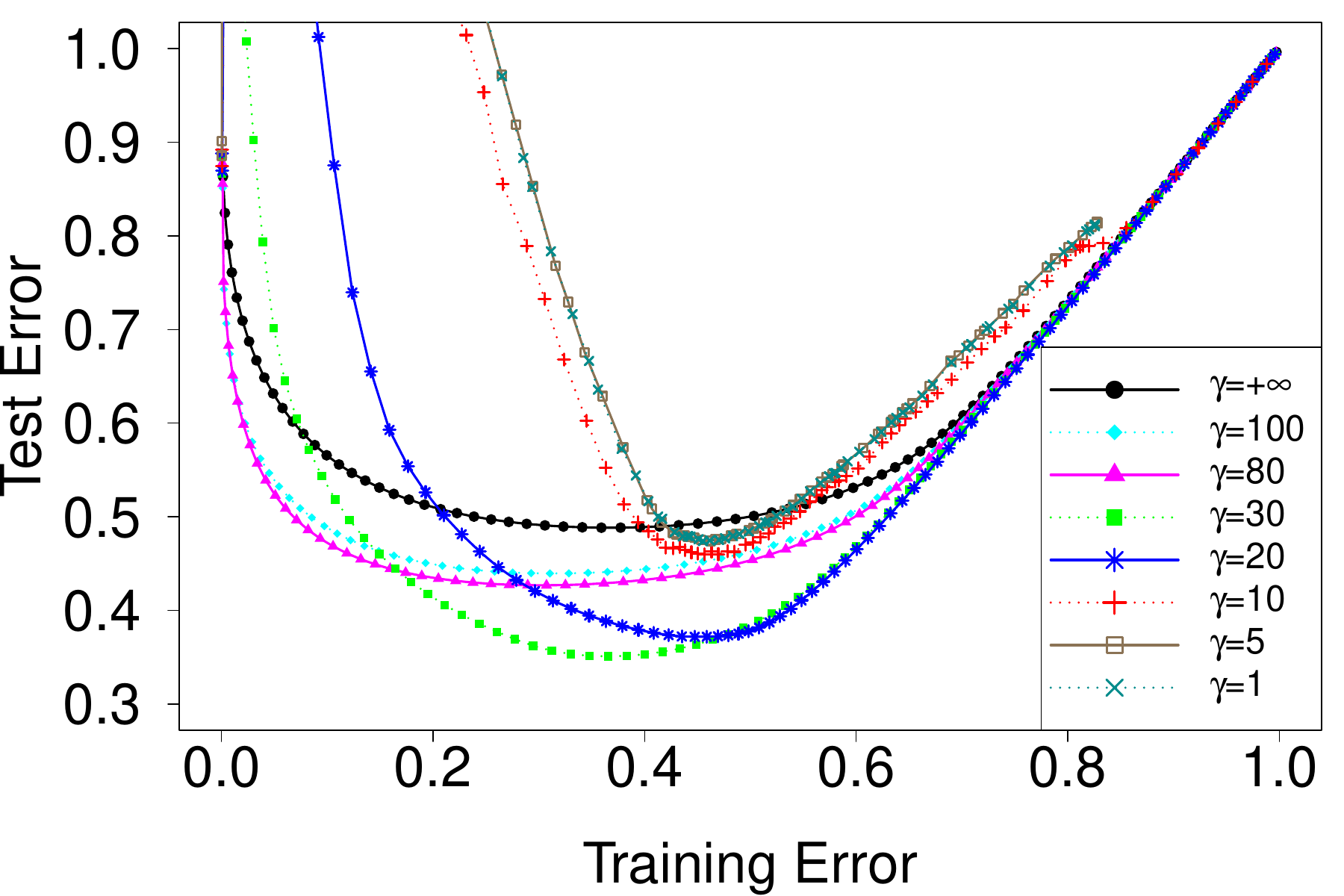}}}
\subfigure[ROM, $90$\% missing, $\text{SNR}=1$, $\text{true rank}=5$]{
\scalebox{.9}{\includegraphics[width=2.6in, height=2.in]{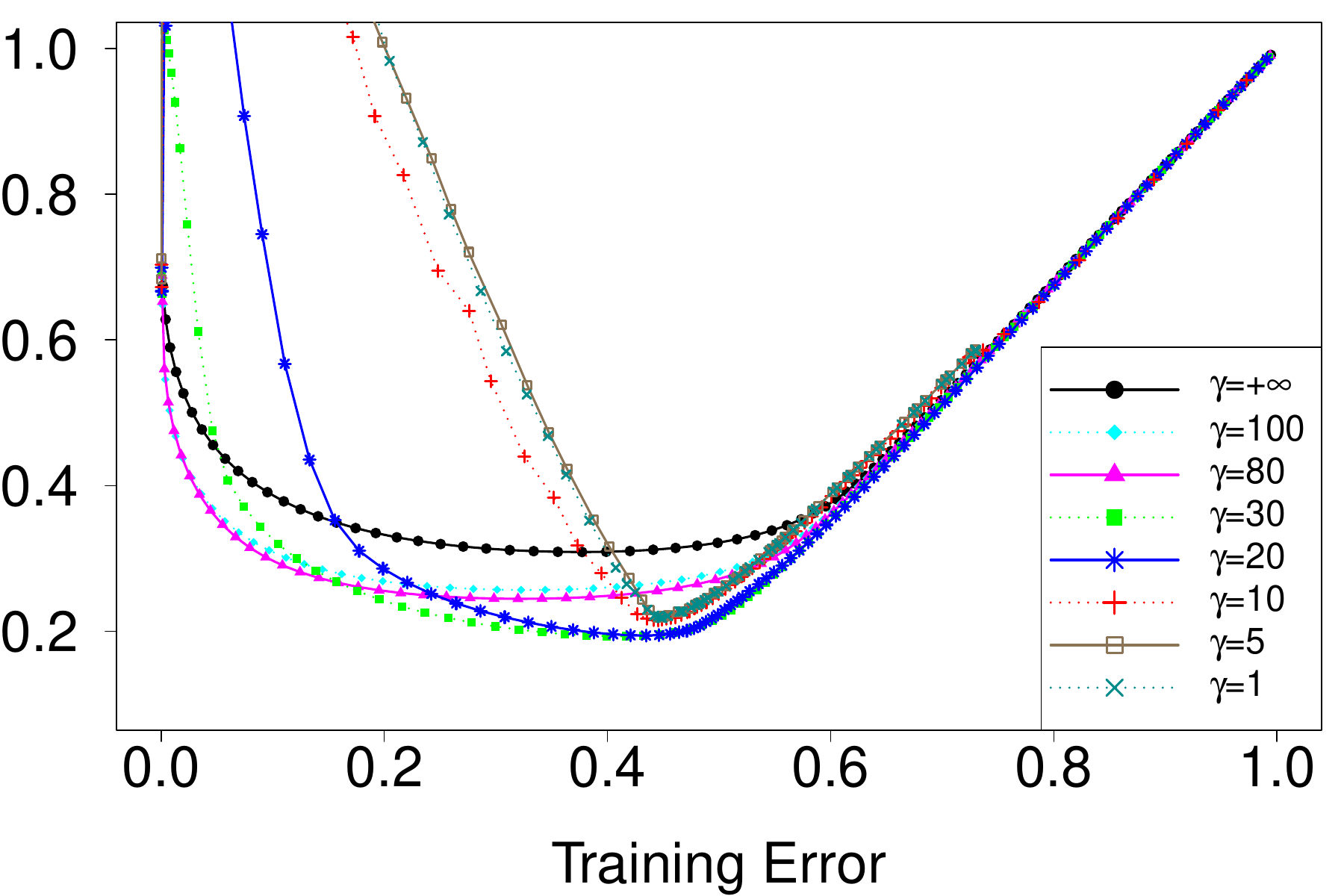}}}
\subfigure[ROM, $90$\% missing, $\text{SNR}=1$, $\text{true rank}=10$]{
\scalebox{.9}{\includegraphics[width=2.6in, height=2.in]{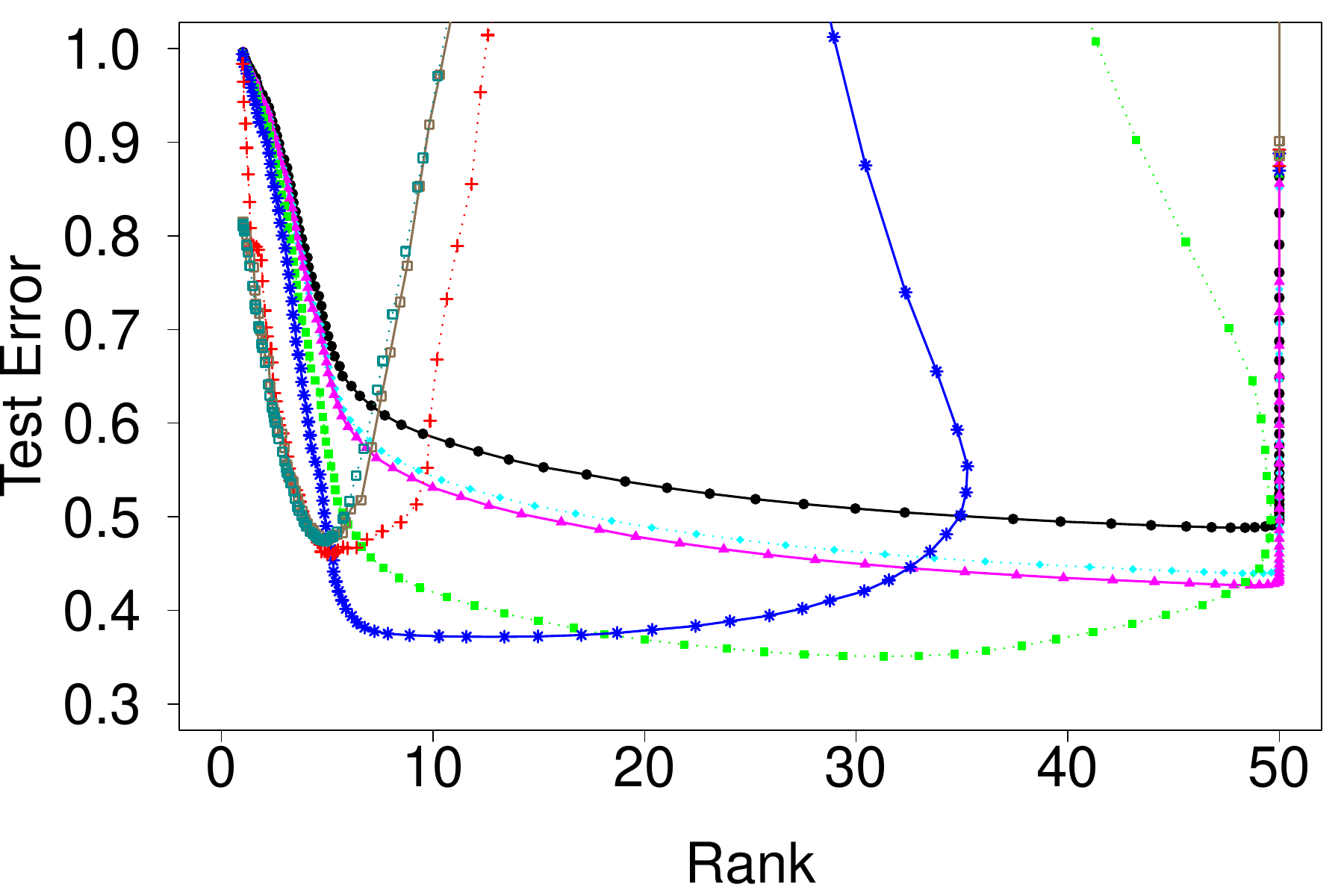}}}
\subfigure[ROM, $90$\% missing, $\text{SNR}=1$, $\text{true rank}=5$]{
\scalebox{.9}{\includegraphics[width=2.6in, height=2.in]{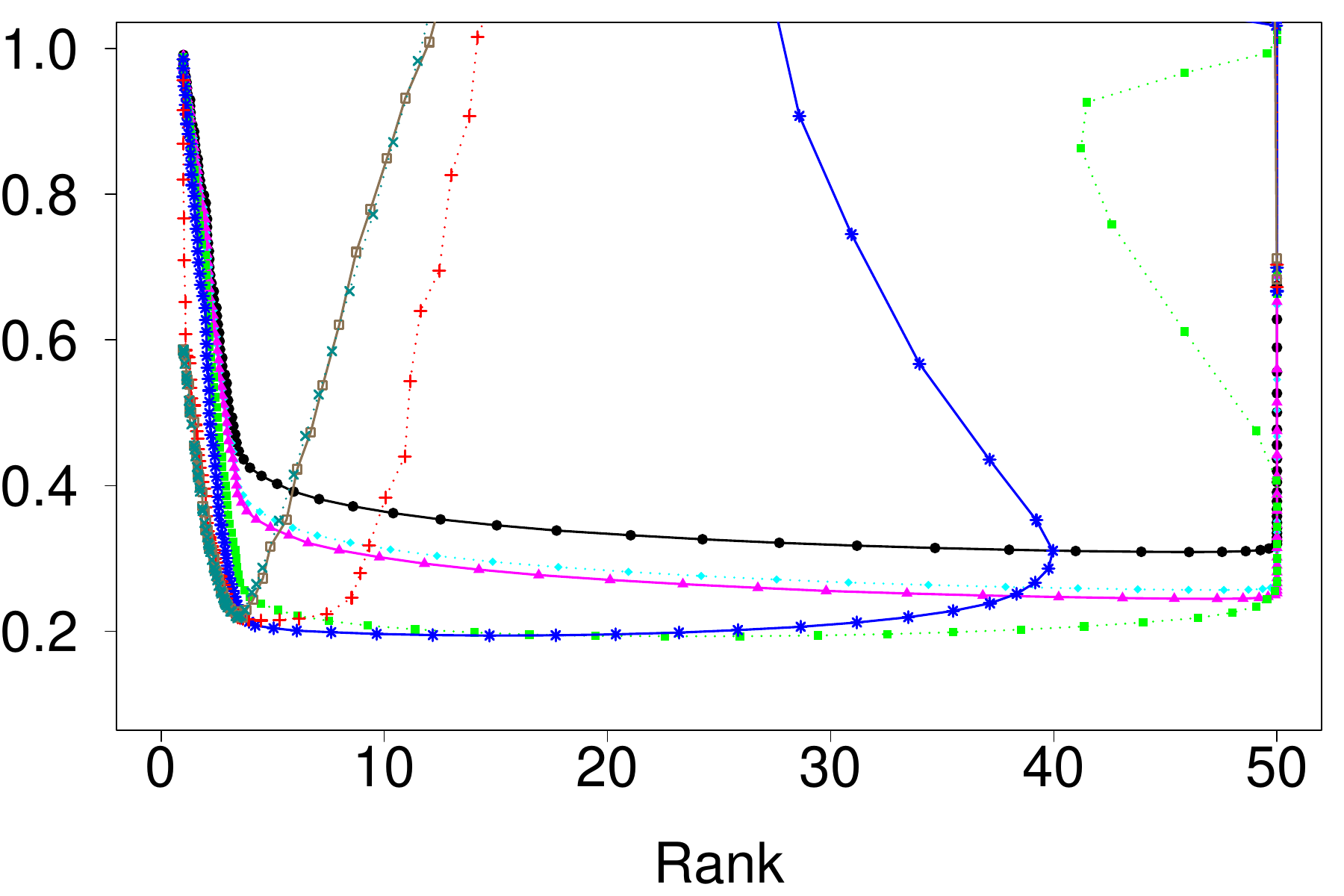}}}
\caption{\small (Color online) Random Orthogonal Model (ROM) simulations with $\text{SNR}=1$. The choice $\gamma=+\infty$ refers to nuclear norm regularization as provided by the \textsc{Soft-Impute} algorithm. We also include the choice $\gamma=1$ to represent the rank regularized approach. The least nonconvex alternatives at $\gamma=100$ and $\gamma=80$ behave similarly to nuclear norm, although with better prediction performance. The choices of $\gamma=1, 5, 10$ result in excessively aggressive fitting behavior for the $\text{true rank}=10$ case, but improve significantly in prediction error and recovering the true rank in the sparser $\text{true rank}=5$ setting. In both scenarios, the intermediate models with $\gamma=30$ and $\gamma=20$ fare the best, with the former achieving the smallest prediction error, while the latter estimates the actual rank of the matrix. Values of test error larger than one are not displayed in the figure.}\label{fig4}
\end{center}
\end{figure*}

\begin{figure*}[htb!]
\begin{center}
{\bf {Example-A}} (High SNR, more missing entries)\\
\subfigure[ROM, $95$\% missing, $\text{SNR}=5$, $\text{true rank}=10$]{
\scalebox{.9}{\includegraphics[width=2.6in, height=2.in]{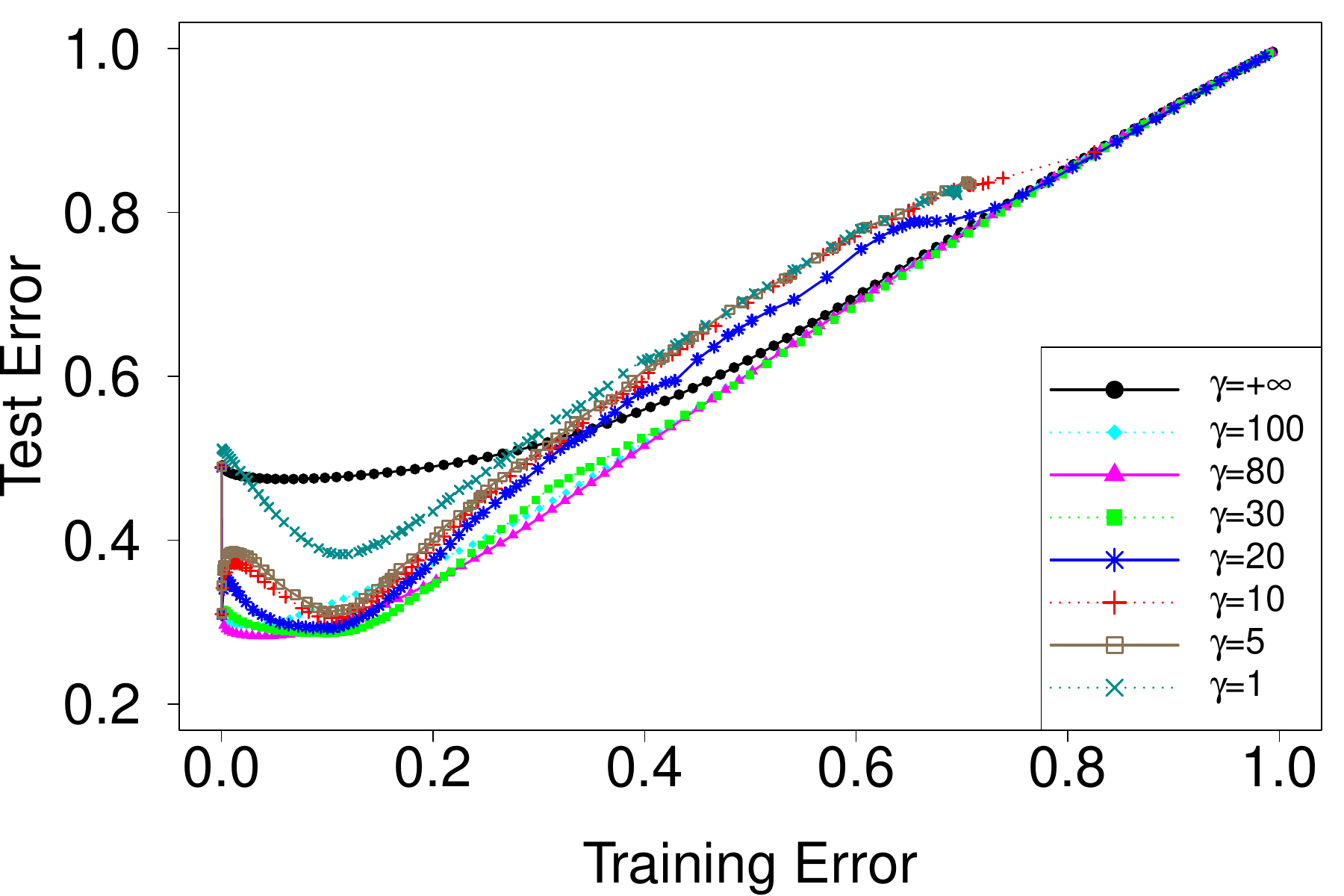}}}
\subfigure[ROM, $95$\% missing, $\text{SNR}=5$, $\text{true rank}=5$]{
\scalebox{.9}{\includegraphics[width=2.6in, height=2.in]{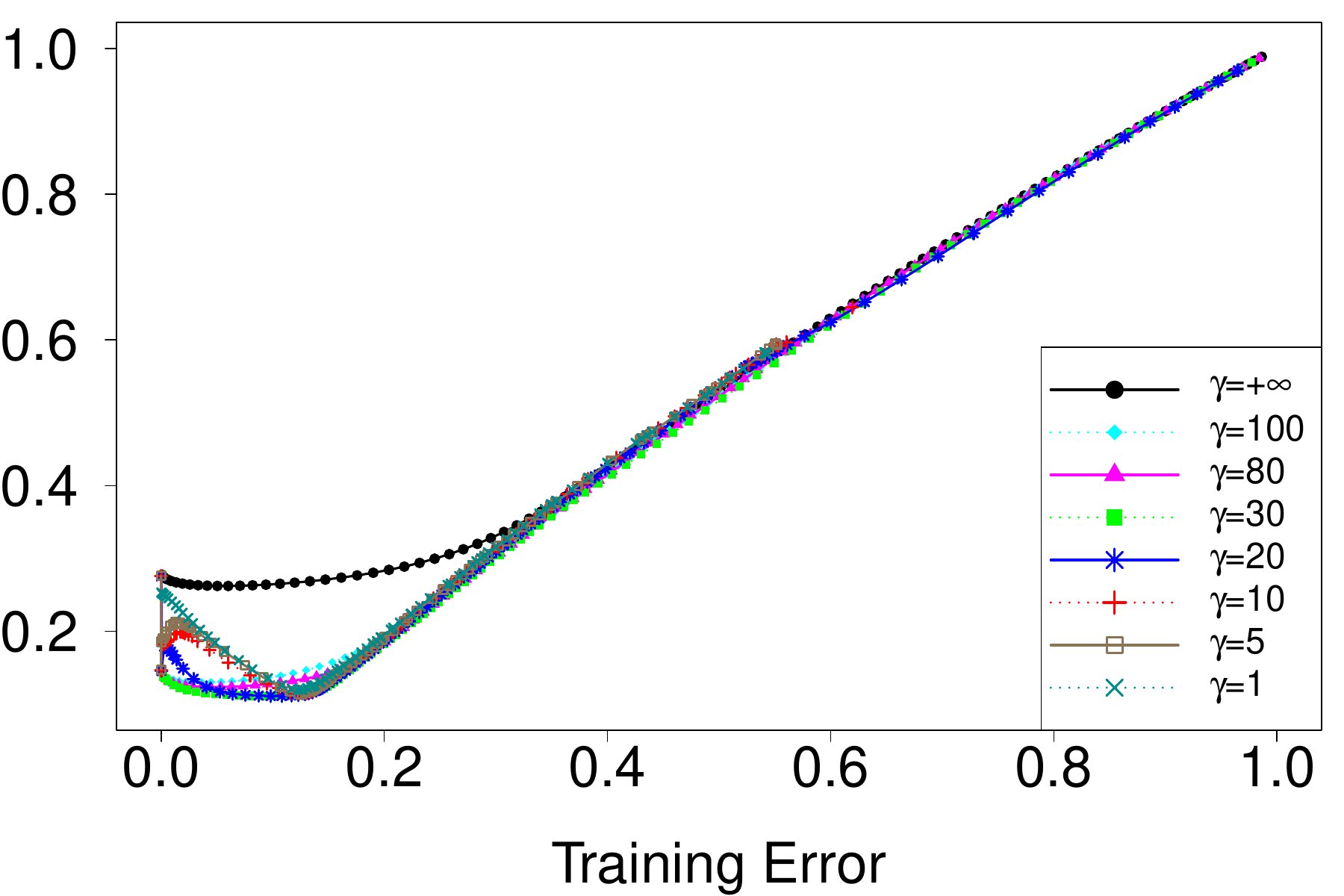}}}
\subfigure[ROM, $95$\% missing, $\text{SNR}=5$, $\text{true rank}=10$]{
\scalebox{.9}{\includegraphics[width=2.6in, height=2.in]{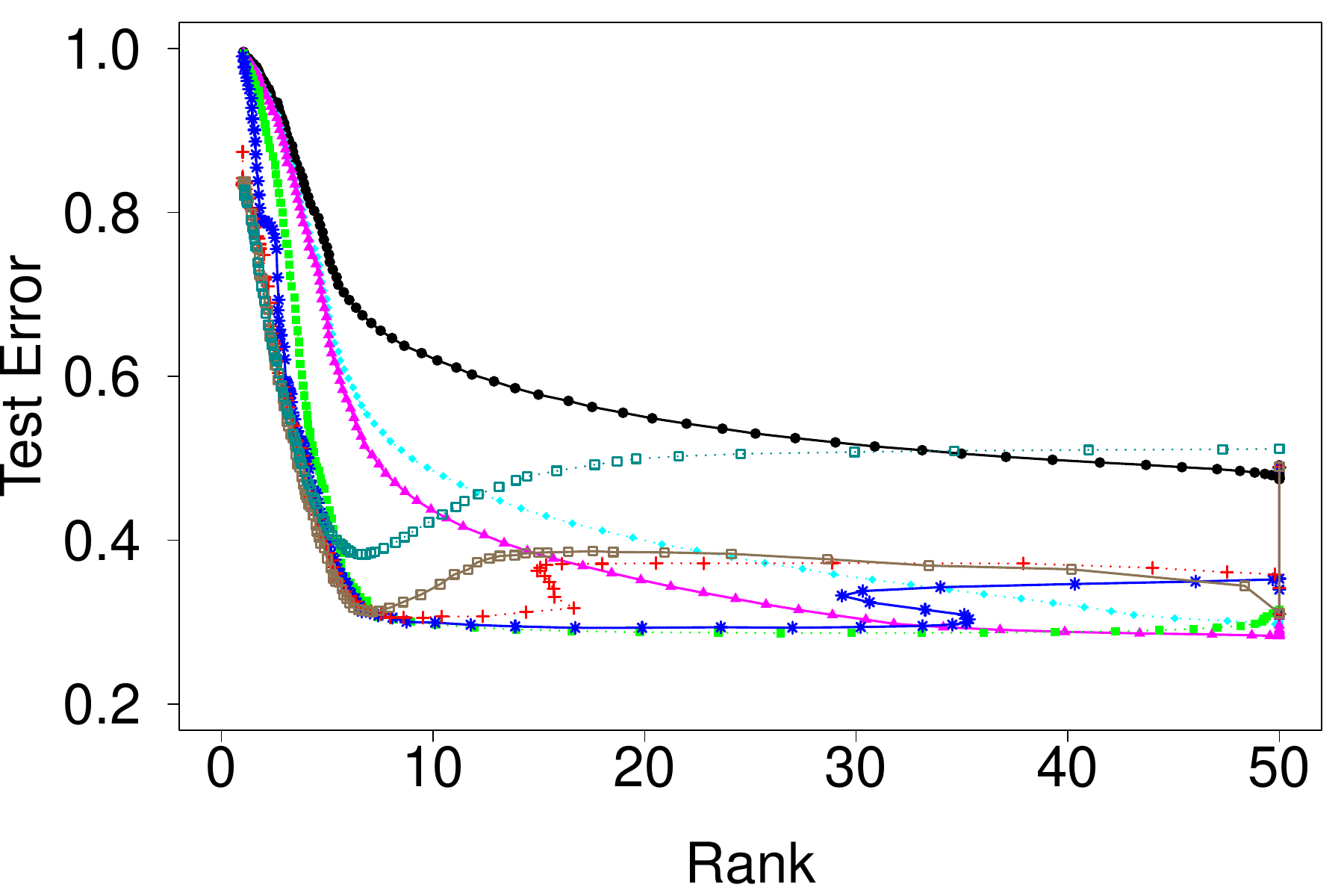}}}
\subfigure[ROM, $95$\% missing, $\text{SNR}=5$, $\text{true rank}=5$]{
\scalebox{.9}{\includegraphics[width=2.6in, height=2.in]{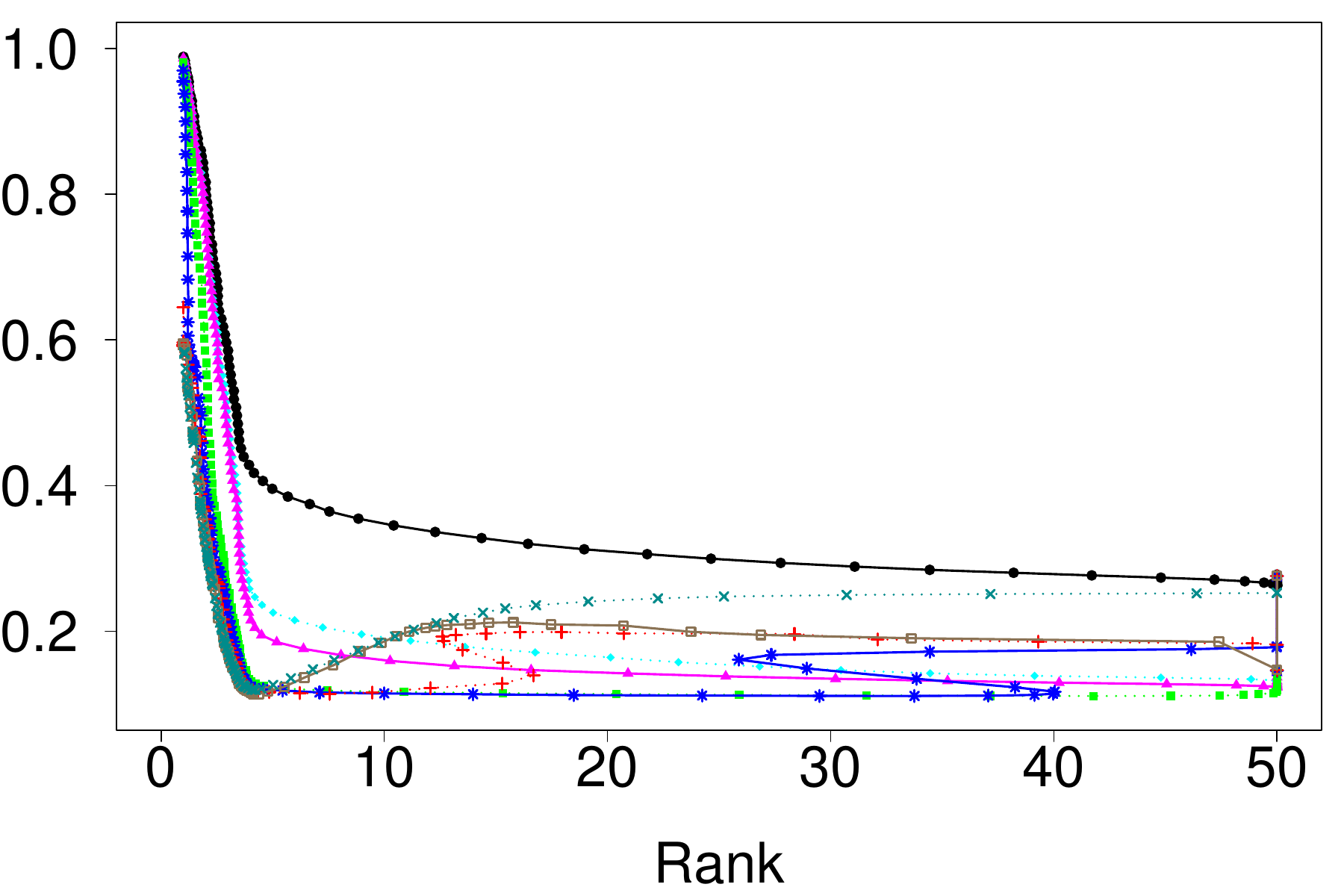}}}
\caption{\small (Color online) Random Orthogonal Model (ROM) simulations with $\text{SNR}=5$. The benefits of nonconvex regularization are more evident in this high-sparsity, high-missingness scenario. While the $\gamma=100$ and $\gamma=80$ models distance themselves more from nuclear norm, the remaining members of the MC+ family essentially minimize prediction error while correctly estimating the true rank. This is especially true in panel (d), where the best predictive performance of the model $\gamma=5$ at the correct rank is achieved under a low-rank truth and high SNR setting.}\label{fig5}
\end{center}
\end{figure*}

\begin{figure*}[htb!]
\begin{center}
{\bf {Example-B}} \hspace{4.5cm} {\bf {Example-C}} \\
\subfigure[\scriptsize Coherent, $90$\% missing, $\text{SNR}=10$, $\text{true rank}=10$]{
\scalebox{.9}{\includegraphics[width=2.6in, height=2.in]{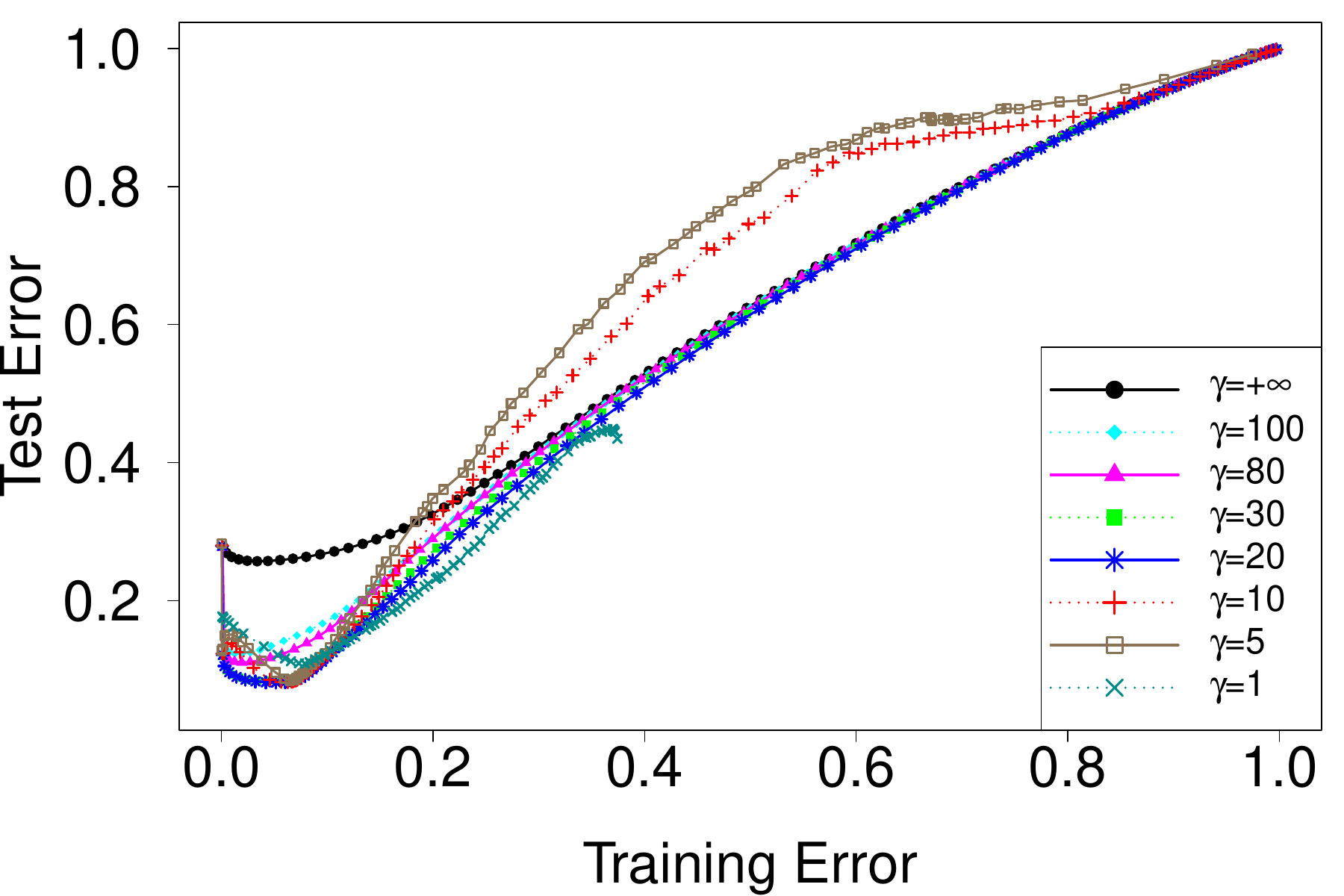}}}
\subfigure[\scriptsize NUS, $25$\% missing, $\text{SNR}=10$, $\text{true rank}=10$]{
\scalebox{.9}{\includegraphics[width=2.6in, height=2.in]{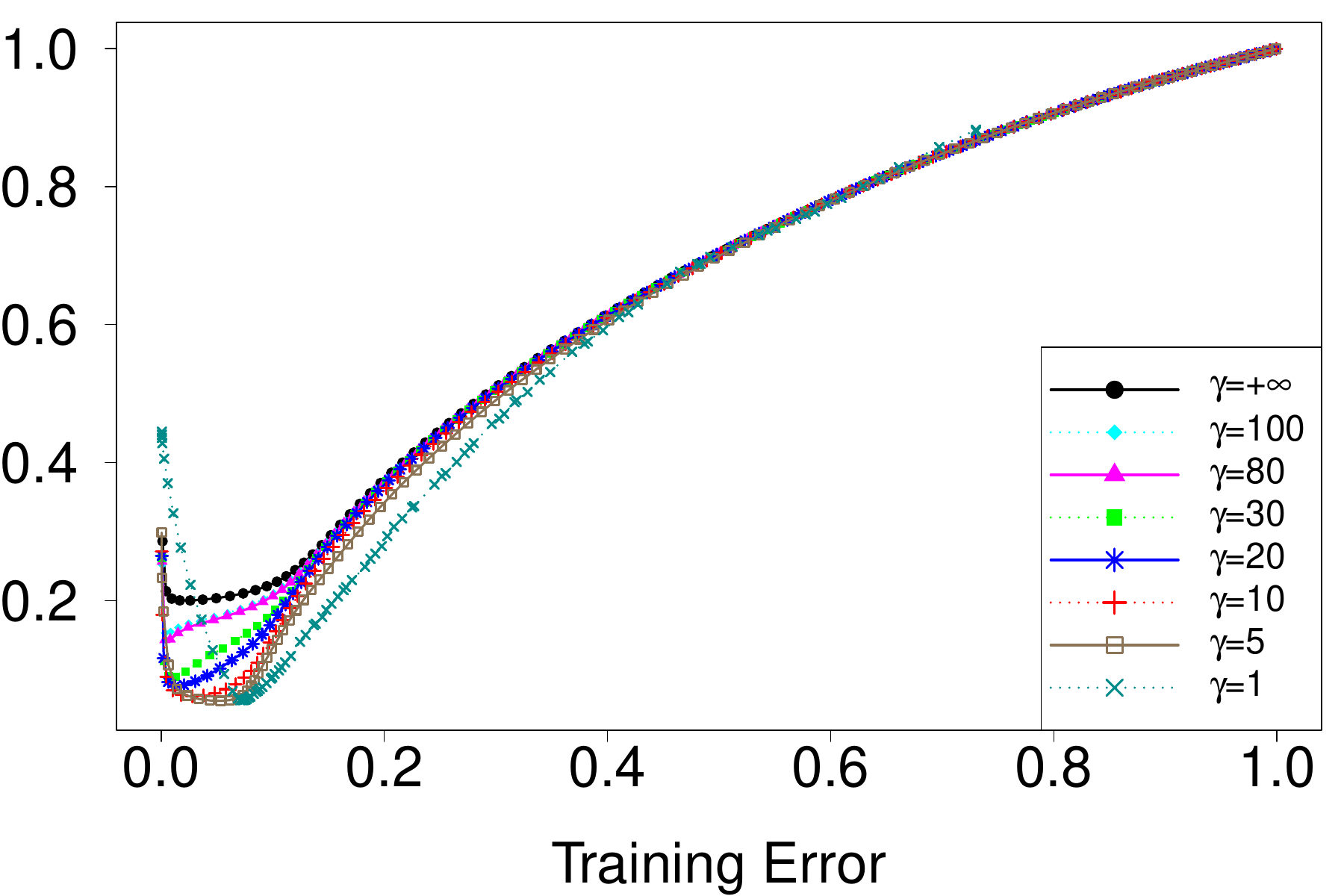}}}
\subfigure[\scriptsize Coherent, $90$\% missing, $\text{SNR}=10$, $\text{true rank}=10$]{
\scalebox{.9}{\includegraphics[width=2.6in, height=2.in]{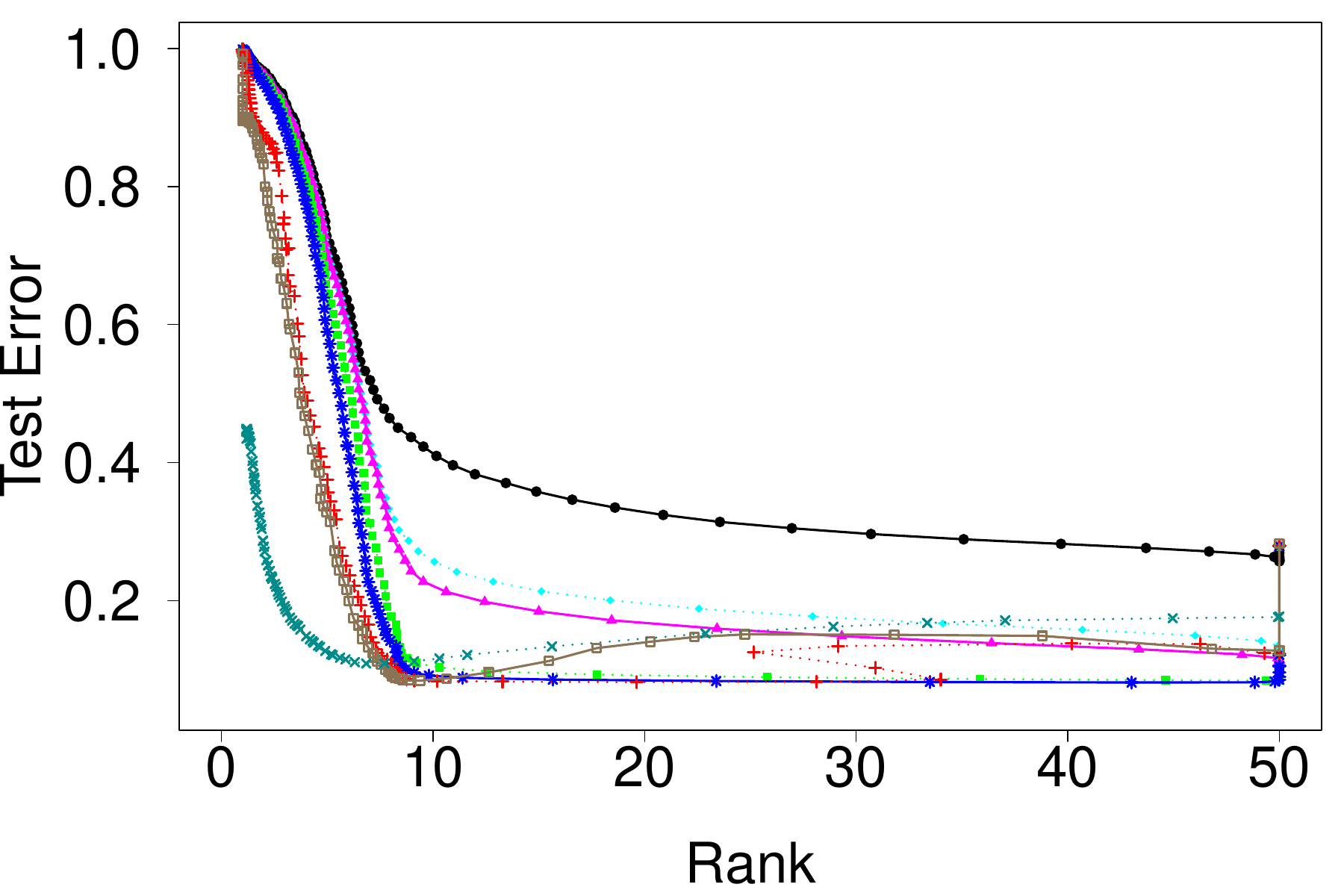}}}
\subfigure[\scriptsize NUS, $25$\% missing, $\text{SNR}=10$, $\text{true rank}=10$]{
\scalebox{.9}{\includegraphics[width=2.6in, height=2in]{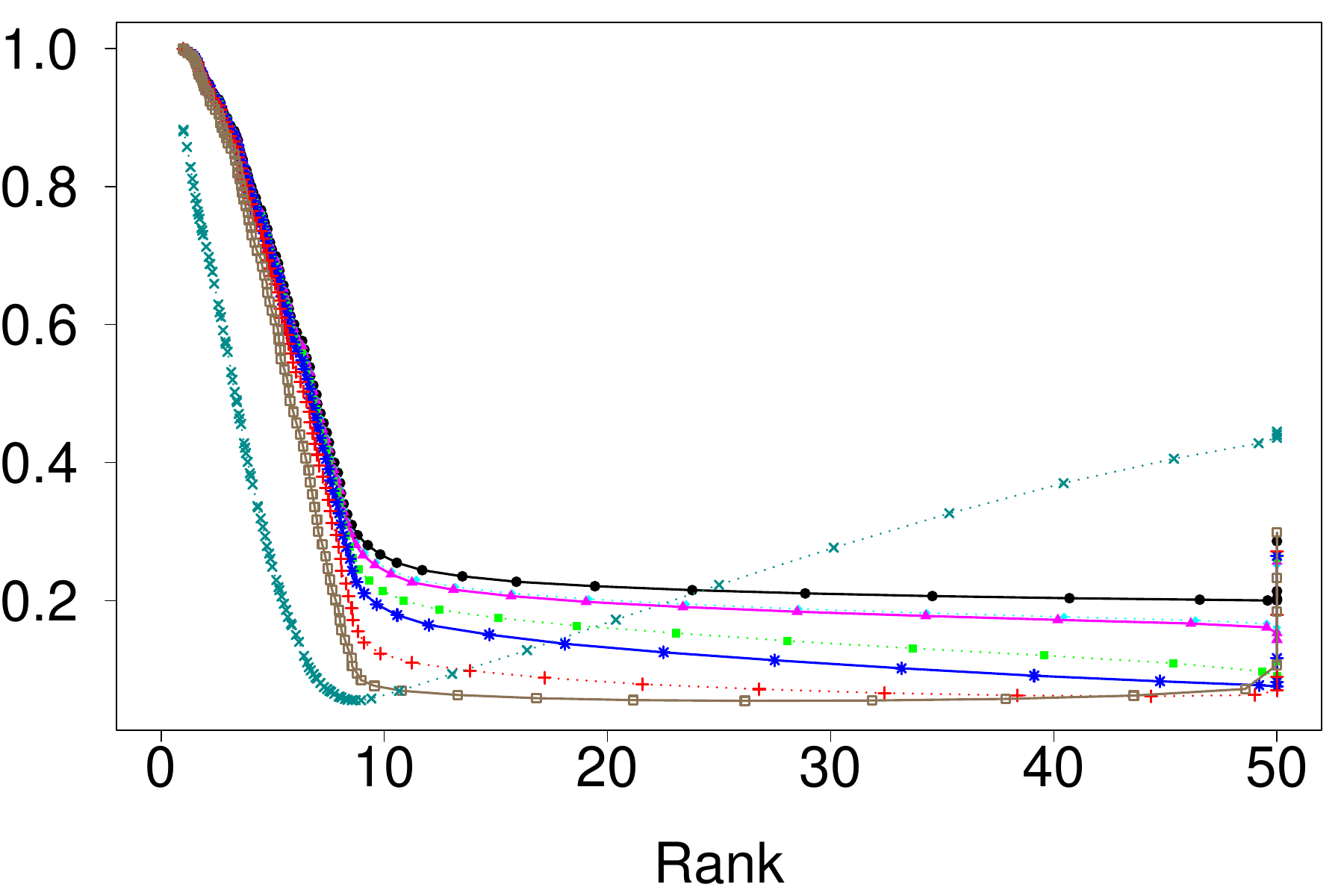}}}
\caption{\small (Color online) Coherent and Nonuniform Sampling (NUS) simulations with $\text{SNR}=10$. nonconvex regularization also proves to be a successful strategy in these challenging scenarios, particularly in the nonuniform sampling setting where the MC+ family exhibits a monotone decrease in prediction error as $\gamma$ approaches $1$. Again, the model $\gamma=5$ estimates the correct rank under high SNR settings. Although nuclear norm achieves a relatively small prediction error, compared with previous simulation settings, the MC+ family still provides a superior and more robust mechanism for regularization.}\label{fig6}
\end{center}
\end{figure*}

\begin{figure*}[htb!]
\begin{center}
{\sf {Real Data Example: MovieLens}}\\
\subfigure[MovieLens100k, $20$\% test data]{
\scalebox{.9}{\includegraphics[width=2.6in, height=2.in]{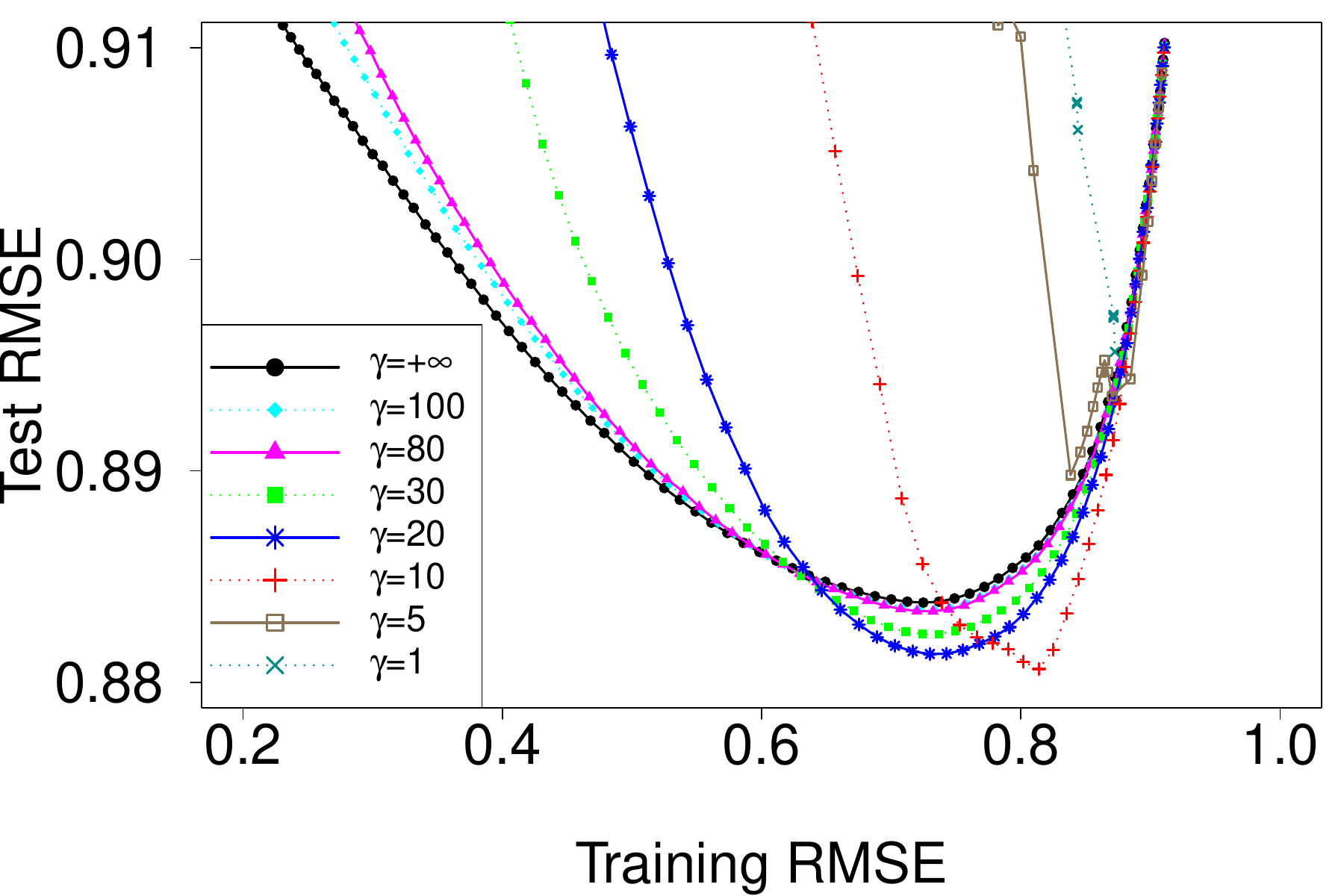}}}
\subfigure[MovieLens1m, $20$\% test data]{
\scalebox{.9}{\includegraphics[width=2.6in, height=2.in]{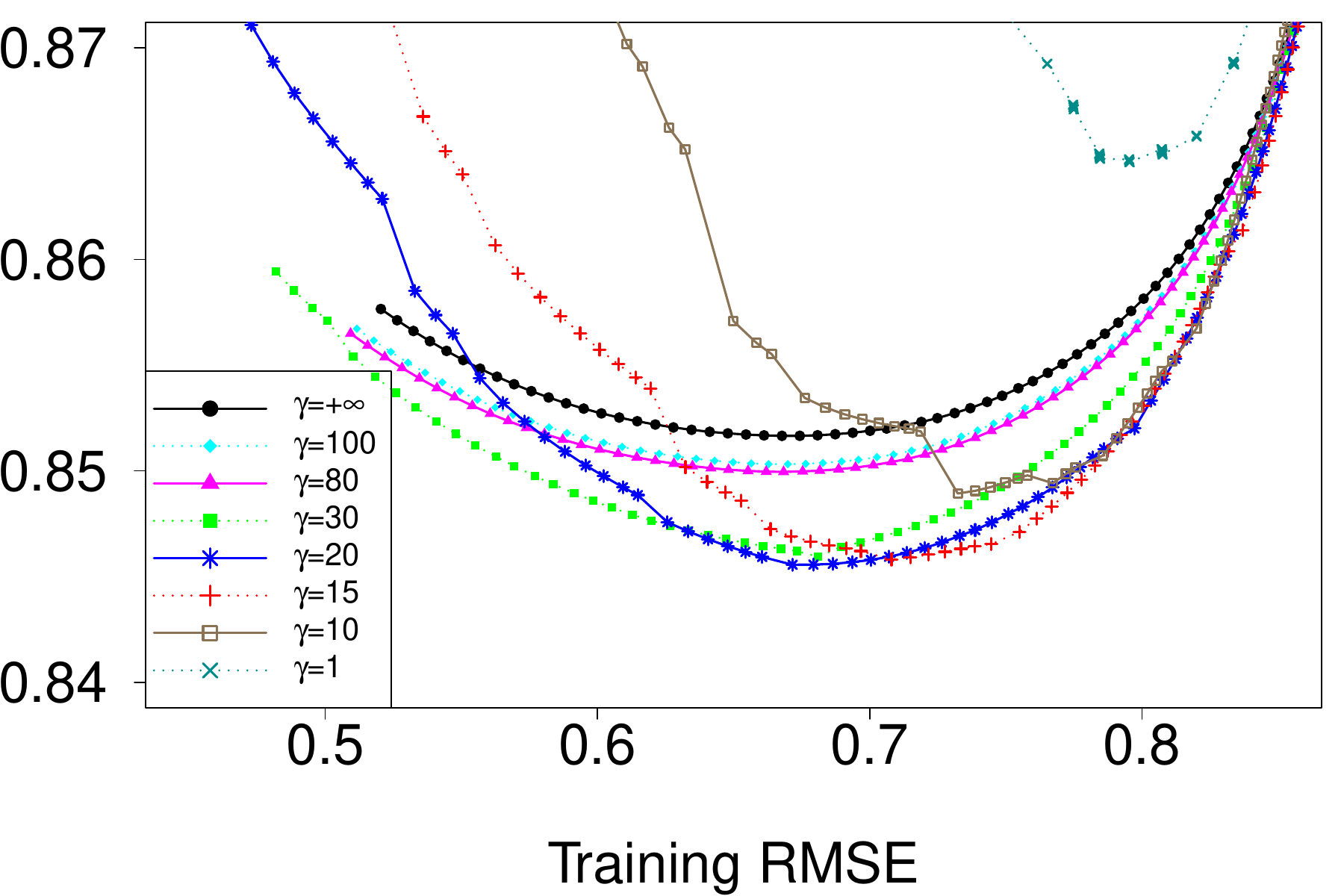}}}
\subfigure[MovieLens100k, $20$\% test data]{
\scalebox{.9}{\includegraphics[width=2.6in, height=2.in]{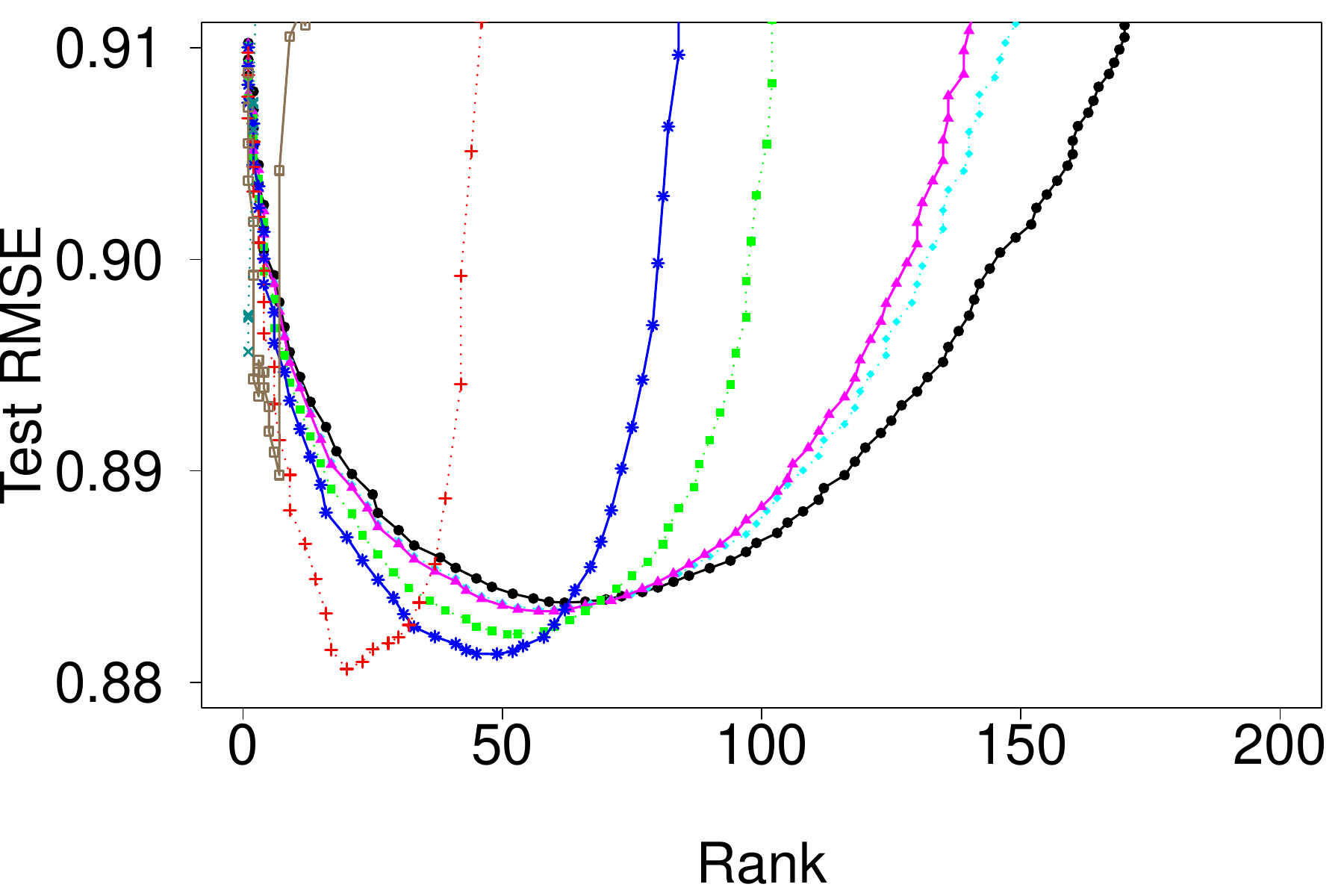}}}
\subfigure[MovieLens1m, $20$\% test data]{
\scalebox{.9}{\includegraphics[width=2.6in, height=2.in]{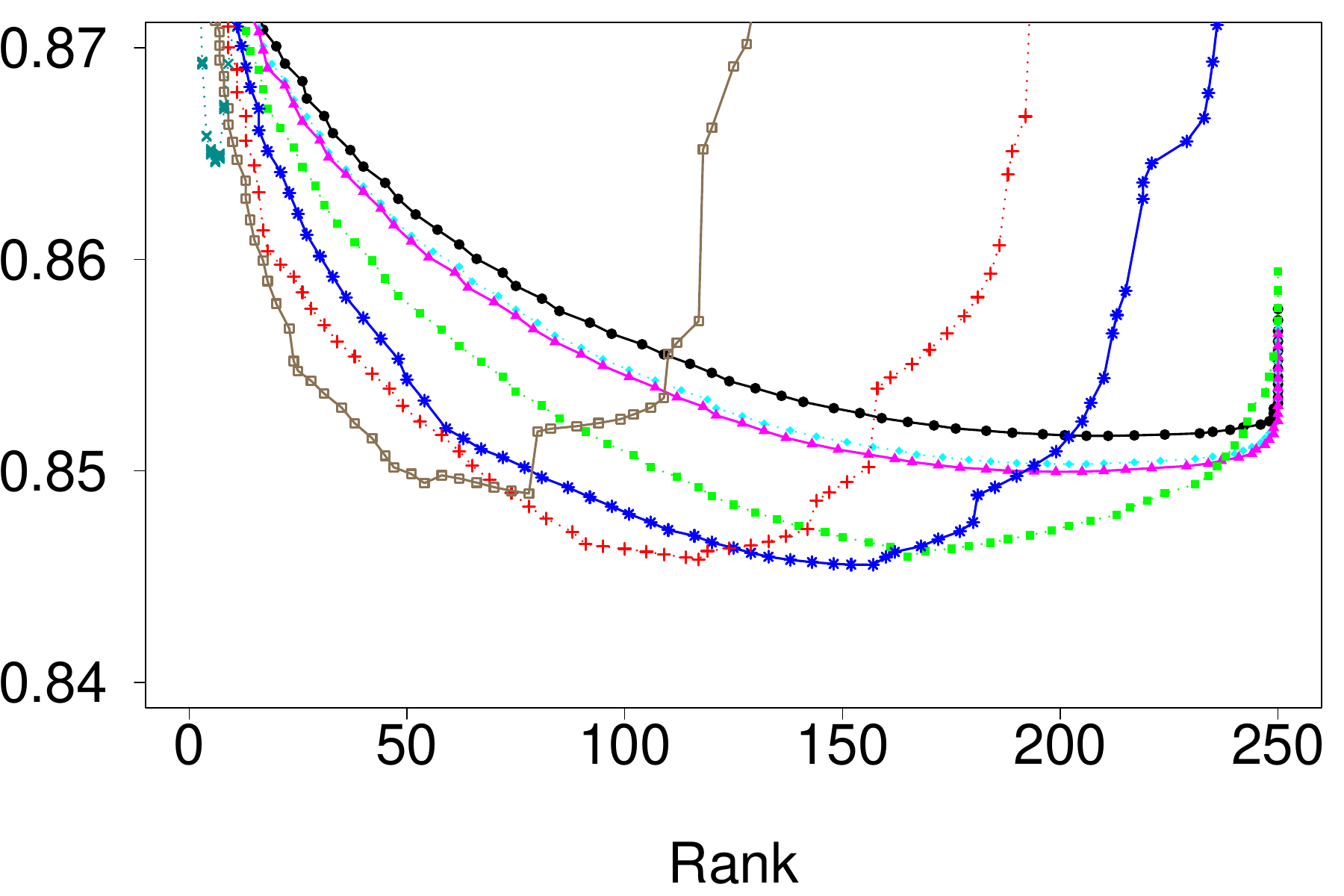}}}
\caption{\small (Color online) MovieLens 100k and 1m data. For each value of $\lambda$ in the solution path, an operating rank threshold (capped at 250) larger than the rank of the previous solution was employed.}\label{fig7}
\end{center}
\end{figure*}

We study three different examples, where, for the true low-rank matrix $M=L \Phi R'$, we vary both the structure of the left and right singular vectors in $L$ and $R$, as well as the sampling scheme used to obtain the observed entries in $\Omega$. Our basic model is $Y_{ij} = M_{ij} + \B\varepsilon_{ij}$, where we observe entries $(i,j) \in \Omega$. 
We consider different types of missing patterns for $\Omega$, and various signal-to-noise (SNR) ratios for the Gaussian error term $\B\varepsilon$, defined here to be:
$$
\text{SNR} = \frac{\var(\text{vec}{(M)})}{\var (\text{vec}(\B\varepsilon))} \,.
$$
Accordingly, the (standardized) training and test error for the model are defined as:
$$\text{Training Error}=\frac{\| \mathcal{P}_{\Omega}(Y-\hat{M}) \|^2_F}{\| \mathcal{P}_{\Omega}(Y) \|^2_F}, ~~\text{Test Error}=\frac{\| \mathcal{P}_{\Omega}^{\perp}(L\Phi R' - \hat{M}) \|^2_F}{\| \mathcal{P}_{\Omega}^{\perp}(L\Phi R') \|^2_F},$$
where a value greater than one for the test error indicates that the computed estimate $\hat{M}$ does a worse job at estimating $M$ than the zero solution, and the training error corresponds to the fraction of the error explained on the observed entries by the estimate $\hat{M}$ relative to the zero solution. 

\paragraph{Example-A:} In our first simulation setting, we use the model 
$$Y_{m \times n}= L_{m \times r} \Phi_{r \times r} R_{r \times n}' + \B\varepsilon_{m \times n},$$
where $L$ and $R$ are matrices generated from the \textit{random orthogonal model} \citep{CandesRecht2009}, and the singular values $\Phi=\text{diag}(\phi_1,\ldots,\phi_r)$ are randomly selected as $\phi_1,\ldots,\phi_r \overset{\text{iid}}{\sim} \text{Uniform}(0,100)$. The set $\Omega$ is sampled uniformly at random. Recall that for this model, exact matrix completion in the noiseless setting is guaranteed as long as $|\Omega| \geq Cmr\log^4m$, for some universal constant $C$ \citep{Recht2011}. Under the noisy setting, \citet{MazumderEtal2010} show superior performance of nuclear norm regularization \textit{vis-\`a-vis} other matrix recovery algorithms \citep{CaiEtal2010,KeshavanEtal2010} in terms of achieving smaller test error. For the purposes herein, we fix $(m,n)=(800,400)$ and set the fraction of missing entries to $|\Omega^c| / mn = 0.9$ and $|\Omega^c| / mn = 0.95$.

\paragraph{Example-B:} In our second setting, we also consider the model 
$$Y_{m \times n}=L_{m \times r} \Phi_{r \times r} R_{r \times n}' +\B\varepsilon_{m \times n}, $$
but we now select matrices $L$ and $R$ which do not satisfy the incoherence conditions required for full matrix recovery. Specifically, for the choices of $(m,n,r)=(800,400,10)$ and $|\Omega^c| / mn = 0.9$, we select $L$ and $R$ to be block-diagonal matrices of the form
\[
L = \diag(L_1, \ldots, L_5), \quad \quad R = \diag(R_1, \ldots, R_5),
\]
where $L_i \in \mathbb{R}^{160 \times 2}$ and $R_i \in \mathbb{R}^{80 \times 2}$, $i=1,\ldots,5$, are random matrices with scaled Gaussian entries. The singular values are again sampled as $\phi_1,\ldots,\phi_r \overset{\text{iid}}{\sim} \text{Uniform}(0,100)$ with $\Omega$ being uniformly random over the set of indices. For this model, successful matrix completion is not guaranteed even for the noiseless problem with the nuclear norm relaxation, as the left and right singular vectors are not sufficiently spread. We observe the usefulness of the nonconvex regularized estimators in this regime, in our experimental results. 

\paragraph{Example-C:} In our third simulation setting, for the choice of $(m,n,r)=(100,100,10)$, we also generate $Y_{m \times n}= L_{m \times r} \Phi_{r \times r} R_{r \times n}' + \B\varepsilon_{m \times n}$ from the random orthogonal model as in our first setting, but we now allow the observed entries in $\Omega$ to follow a nonuniform sampling scheme. In particular, we fix $\Omega^c=\{1 \leq i,j \leq 100: 1 \leq i \leq 50, 51 \leq j \leq 100 \}$ so that 
$$
\mathcal{P}_{\Omega}(Y)=
\begin{bmatrix}
Y_{11} & 0 \\
Y_{21} & Y_{22}
\end{bmatrix} \,\;\;\text{where,}\;\;\; Y=
\begin{bmatrix}
Y_{11} & Y_{12} \\
Y_{21} & Y_{22}
\end{bmatrix},
$$
with the fraction of missing entries thus being $|\Omega^c| / mn = 0.25$. This is again a challenging simulation setting in which both the uniform \citep{CandesRecht2009} and independent \citep{ChenEtal2014} sampling scheme assumptions in $\Omega$ are violated. 
Our aim again is to explore whether the nonconvex MC+ family is able to outperform nuclear norm regularization in this regime. 

For all three settings above,  we choose a $100 \times 25$ grid of $(\lambda,\gamma)$ values as follows. In each simulation instance we fix $\lambda_1=\|\mathcal{P}_{\Omega}(Y)\|_2$, the smallest value of $\lambda$ for which the nuclear norm regularized solution is zero, and set $\lambda_{100}=0.001 \cdot \lambda_1$. Keeping in mind that \textsc{NC-Impute} benefits greatly from using warm starts, we construct an equally spaced sequence of $100$ values of $\lambda$ decreasing from $\lambda_1$ to $\lambda_{100}$. We choose 25 $\gamma$-values in a logarithmic grid from 5000 to 1.1.
The results displayed in Figures \ref{fig4} -- \ref{fig6} show averages of training and test errors, as well as recovered ranks of the solution matrix $\hat{M}_{\lambda,\gamma}$ 
for the values of $(\lambda,\gamma)$, taken over 50 simulations under all three problem instances. The plots including rank reveal how effective the MC+ family is at recovering the true rank while minimizing prediction error. 
Throughout the simulations we keep an upper bound of the operating rank as  50.

\subsubsection{Discussion of Experimental Results}

We devote Figures \ref{fig4} and \ref{fig5} to analyze the simpler random orthogonal model (Example-A), leaving the more challenging coherent and nonuniform sampling settings (Example-B and Example-C) for Figure \ref{fig6}. In each case, the captions detail the results which we summarize here. The noise is quite high in Figure \ref{fig4} with $\text{SNR}=1$ and $90$\% of the entries missing in both displayed settings, while the model complexity decreases from a true rank of $10$ to $5$. The underlying true ranks remain the same in Figure \ref{fig5}, but the noise level has decreased to $\text{SNR}=5$ with the missing entries increasing to $95$\%. For each model setting considered, all nonconvex methods from the MC+ family outperform nuclear norm regularization in terms of prediction performance, while members of the MC+ family with smaller values of $\gamma$ are better at estimating the correct rank. The choices of $\gamma=30$ and $\gamma=20$ have the best performance in Figure \ref{fig4} (best prediction errors around the true ranks), while more nonconvex alternatives fare better in the high-sparsity, low-noise setting of Figure \ref{fig5}. In both figures, the performance of nuclear norm regularization is somewhat similar to the least nonconvex alternative displayed at $\gamma=100$, however, the bias induced in the estimation of the singular values of the low-rank matrix $M$ leads to the worst bias-variance trade-off among all training versus test error plots for the settings considered.

While the nuclear norm relaxation provides a good convex approximation for the rank of a matrix~\citep{RechtEtal2010}, these examples show that nonconvex regularization methods provide a superior mechanism for rank estimation. This is reminiscent of the performance of the MC+ penalty in the context of variable selection within high-dimensional sparse regression models. Although the $\ell_1$ penalty function represents the best convex approximation to the $\ell_0$ penalty, the gap bridged by the nonconvex MC+ penalty family provides a better basis for model selection, and hence rank estimation in the low-rank matrix completion setting.

For the coherent and nonuniform sampling settings of Figure \ref{fig6}, we choose the small noise scenario $\text{SNR}=10$ in order to favor all considered models. Despite the absence of any theoretical guarantees for successful matrix recovery, the nuclear norm regularization approach achieves a relatively small prediction error in all displayed instances. Nevertheless, the nonconvex MC+ family of penalties seems empirically more adept at overcoming the limitations of nuclear norm penalizated matrix completion in these challenging simulation settings.
In particular, the most aggressive nonconvex fitting behavior at $\gamma=5$ achieves excellent prediction performance in the nonuniform sampling setting \\ while correctly estimating the true rank of the coherent model. 

\begin{figure*}[tb!]
\begin{center}
{\sf {Real Data Example: Netflix}}\\
\subfigure[Netflix, test data=$1,500,000$ ratings]{
\scalebox{.9}{\includegraphics[width=2.9in, height=2.7in]{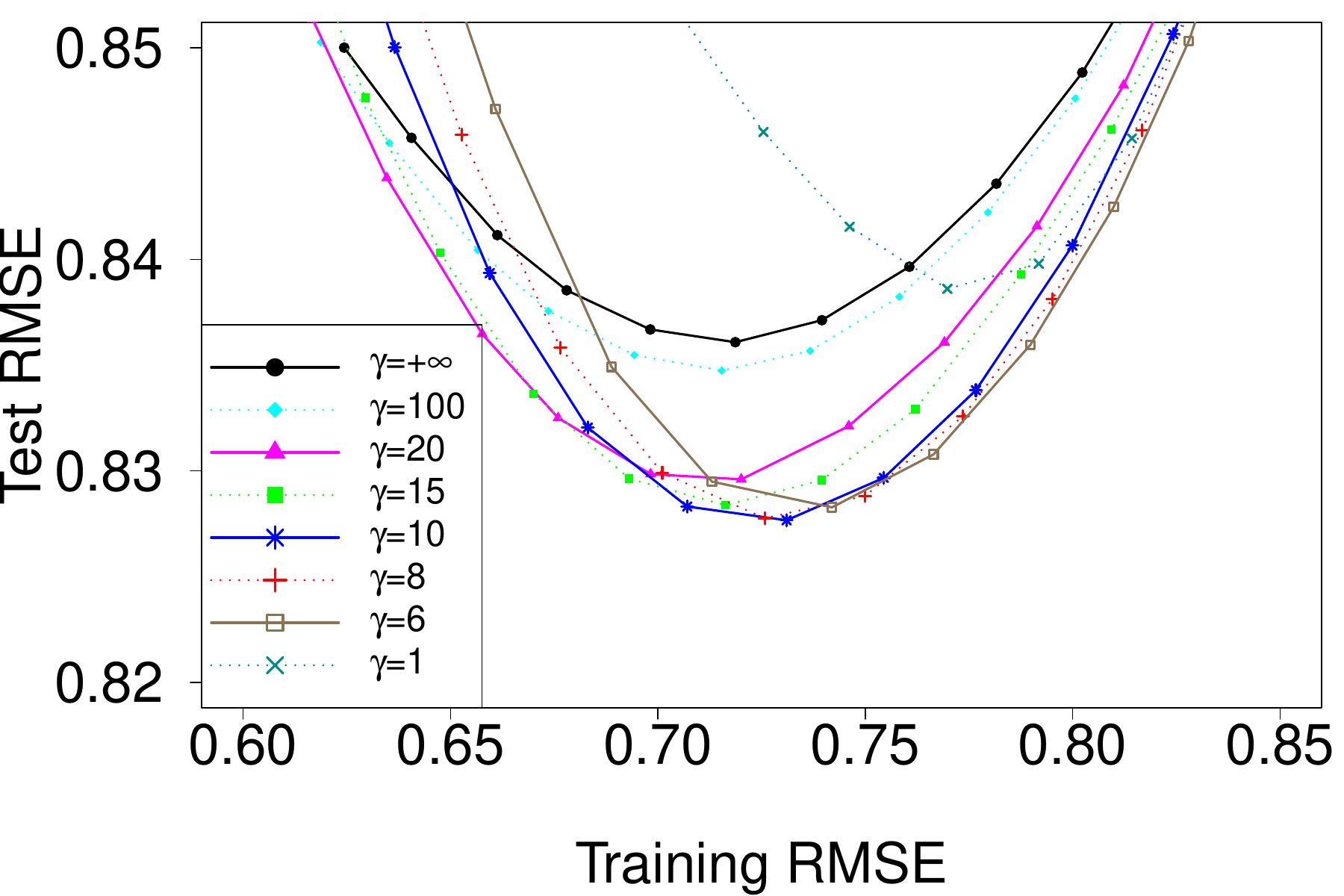}}}
\subfigure[Netflix, test data=$1,500,000$ ratings]{
\scalebox{.9}{\includegraphics[width=2.9in, height=2.7in]{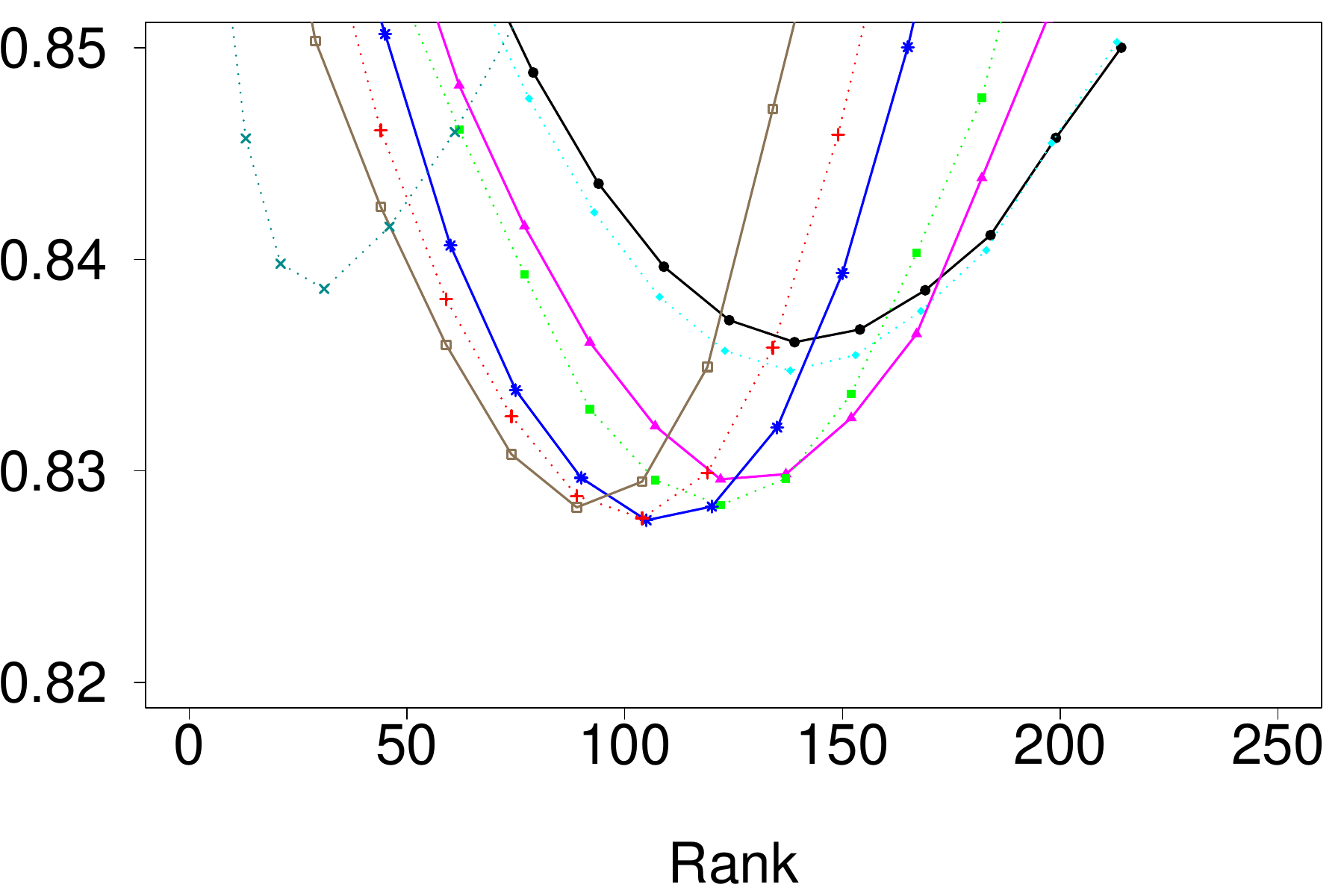}}}
\caption{\small (Color online) Netflix competition data. The model $\gamma=10$ achieves optimal test set RMSE of $0.8276$ for a solution rank of 105.}\label{fig8}
\end{center}
\end{figure*}

\subsection{Real Data Examples: MovieLens and Netflix datasets}

We now use the real world recommendation system datasets \verb"ml100k" and \verb"ml1m" provided by MovieLens\footnote{Available at
\url{http://grouplens.org/datasets/movielens/}}, as well as the famous Netflix competition data to compare the usual nuclear norm approach with the MC+ regularizers. The dataset \verb"ml100k" consists of $100,000$ movie ratings (1--5) from $943$ users on $1,682$ movies, whereas \verb"ml1m" includes $1,000,209$ anonymous ratings from $6,040$ users on $3,952$ movies. In both datasets, for all regularization methods considered, a random subset of $80$\% of the ratings were  used for training purposes; the remaining were used as the test set.

We also choose a similar $100 \times 25$ grid of $(\lambda,\gamma)$ values, but for each value of $\lambda$ in the decreasing sequence, we use an ``operating rank" threshold somewhat larger than the rank of the previous solution, with the goal of always obtaining solution ranks smaller than the operating threshold. Following the approach of \citet{HastieEtal2015}, we perform row and column centering of the corresponding (incomplete) data matrices as a preprocessing step.

Figure \ref{fig7} compares the performance of nuclear norm regularization with the MC+ family of penalties on these datasets, in terms of the prediction error (RMSE) obtained from the left out portion of the data. While the fitting behavior at $\gamma=5$ seems to be overly aggressive in these instances, the choice $\gamma=10$ achieves the best test set RMSE with a minimum solution rank of 20 for the \verb"ml100k" data. With a higher test RMSE, nuclear norm regularization achieves its minimum with a less parsimonious model of rank 62. Similar results hold for the \verb"ml1m" data, where the model $\gamma=15$ achieves near optimal test RMSE at a solution rank of 115, while the best estimation accuracy of \textsc{Soft-Impute} occurs for ranks well over 200.

The Netflix competition data consists of $100,480,507$ ratings from $480,189$ users on $17,770$ movies. A designated probe set, a subset of $1,408,395$ of these ratings, was distributed to participants for calibration purposes, leaving $99,072,112$ for training. We did not consider the probe set as part of this numerical experiment, instead choosing \\ $1,500,000$ randomly selected entries as test data with the remaining $97,572,112$ used for training purposes. Similar to the MovieLens data, we select a $20 \times 25$ grid of $(\lambda,\gamma)$ values which adaptively chooses an operating rank threshold, and also remove row and columns means for prediction purposes.

As shown in Figure \ref{fig8}, the MC+ family again yields better prediction performance under more parsimonious models. On average, and for a convergence tolerance of $0.001$ in Algorithm \ref{aone}, the sequence of twenty models took under $10.5$ hours of computing on an Intel  E5-2650L cluster with 2.6 GHz processor. 
We note that our main goal here is to show the feasibility of applying \textsc{NC-Impute} to the MC+ family on a Netflix sized dataset, and further reductions in 
computation time may be possible with specialized implementations. It seems that using a family of nonconvex penalties leads to 
models with better statistical properties, when compared to the nuclear norm regularized problem and the rank constrained problem (obtained via \textsc{Hard-Impute}, for example).

\section{Conclusions and Discussions}\label{sec-concl}
In this paper we present a computational study for the noisy matrix completion problem with nonconvex spectral penalties --- we consider a family of spectral penalties that bridge the convex nuclear norm penalty and the rank penalty, leading to a family of estimators with varying degrees of shrinkage and nonconvexity. 
We propose~\NSI --- an algorithm that appropriately modifies and enhances the EM-stylized procedure \SI~\citep{MazumderEtal2010}, to compute a two dimensional 
family of solutions with specialized warm-start strategies. The main computational bottleneck of our algorithm is a low-rank SVD of a structured matrix, which is performed using a block QR stylized strategy that makes effective use of singular subspace warm-start information across iterations.  We discuss computational guarantees of our algorithm, including a finite time complexity analysis to a first order stationary point. We present a systematic study of various statistical and structural properties of spectral thresholding functions, which form a main building block in our algorithm. We demonstrate the impressive gains in statistical properties of our framework on a wide array of synthetic and real datasets. The current work leaves open several important directions for future research.
\begin{itemize}
\item \emph{Statistical guarantees for the nonconvex method}. In addition to the comprehensive algorithmic analysis presented in the paper, it is of great importance to establish statistical theories such as estimation error bounds to shed new lights on the empirical success of the proposed nonconvex method. In particular, given that our algorithm returns stationary points, it would be interesting to obtain statistical errors for these local optima. In fact, such type of results have been derived for regularized M-estimators in some general multivariate analysis settings when the penalty function is separable \citep{loh2015regularized}. A notable difficulty in the matrix completion problem is that the penalty is imposed on the singular values and is hence a non-separable function of the matrix, which will require more delicate analyses. Moreover, we should point out that statistical analysis of nonconvex optimization methods for matrix completion has been actively investigated in recent years. See the works surveyed in the last paragraph of Introduction and \cite{chi2019nonconvex} for a thorough review. However, the nonconvex methods studied in this line of research are exclusively based on low-rank matrix factorization formulation, and the regularization from these methods is less general than the ones in our method. The nonconvexity of the former methods is largely due to the matrix factorization which enables the reduction of memory and computation costs, while the nonconvexity of our method arises from the nonconvex penalties that aim to attenuate the bias. 
\item \emph{Sharp comparison between nonconvex and convex methods}. Referring to both simulation study and real data analysis in Section \ref{sec4}, we observe that the value of $\gamma$ leading to the optimal matrix completion performance lies between 5 and 30. Recall that the penalty parameter $\gamma$ in the MC+ penalties controls the amount of nonconvexity in the regularization. As $\gamma$ decreases from $\infty$ down to $1$, the penalty behaves closer to $\ell_0$ and farther away from $\ell_1$. Hence, the empirical results in Section \ref{sec4} demonstrate that neither the convex approach ($\gamma=+\infty)$ nor the most aggressive nonconvex one ($\gamma=1$) is the optimal choice. This phenomenon is in fact a manifestation of bias-variance tradeoff. A smaller value of $\gamma$ brings more ``nonconvexity" to the regularization and hence induces less bias as expected. On the other hand, more ``nonconvexity" means more aggressiveness in the selection of low rank matrices and thus results in larger variance. Consequently, for a given level of signal-to-noise ratio in the observations, the optimal $\gamma$ is the one that strikes the best balance between bias and variance. This point of view lends further support to our proposed nonconvex method which incorporates the entire family of nonconvex penalties instead of some particular instantiations. Recent works by a subset of the authors~\citep{zheng2017does, weng2018overcoming, wang2019bridge} have given sharp theoretical characterizations of such a phenomenon in the high-dimensional sparse regression and variable selection problems. The results there reveal that among the $\ell_q$ penalties for $q\in [0,2]$, as the signal-to-noise ratio (SNR) decreases, the optimal value of $q$ will move from $0$ towards $2$. See also the work of~\cite{hazimeh2019fast, mazumder2017subset} for similar observations regarding the overfitting of $\ell_0$-based estimators for low SNR regimes. For the MC+ penalties in the matrix completion problem, $\gamma$ plays a similar role as $q$ does in the regression problem. It would be of great interest to derive analogue theories for the matrix completion problem and establish a sharp characterization of the proposed nonconvex method. 


\end{itemize}

\section{Appendix}
\label{last:appen}

\subsection{Additional Technical Material}
\label{app}

\begin{lemma}\label{lemma3}(Marchenko-Pastur law \citep{BaiSilverstein2010}).
Let $X\in \mathbb{R}^{m \times n}$, where $X_{ij}$ are \text{iid} with $\mathbb{E}(X_{ij})=0, \mathbb{E}(X_{ij}^2)=1$, and $m>n$. Let $\lambda_1\leq \lambda_2 \leq \dots \leq \lambda_n$ be the eigenvalues of $Q_m=\frac{1}{m}X'X$. Define the random spectral measure
$$\mu_n=\frac{1}{n}\sum_{i=1}^n\delta_{\lambda_i}\,.$$
Then, assuming $n/m \rightarrow \alpha \in (0,1]$, we have
$$\mu_n(\cdot, \omega) \rightarrow \mu~~a.s.,$$
where $\mu$ is a deterministic measure with density
\begin{equation*}\label{app2}
\frac{d\mu}{dx}=\frac{\sqrt{(\alpha_+-x)(x-\alpha_-)}}{2\pi \alpha x}I(\alpha_-\leq x \leq \alpha_+).
\end{equation*}
Here, $\alpha_+=(1+\sqrt{\alpha})^2\,$ and $\, \alpha_-=(1-\sqrt{\alpha})^2$.
\end{lemma}

\subsubsection{Proof of Proposition \ref{propfive}.} \label{proof-prop-five}
\begin{proof}

In the following proof, we make use of the notation: $\Theta_1(\cdot)$ and $\Theta_2(\cdot)$, defined as follows. 
For two positive sequences $a_{k}$ and $b_k$, we say $a_k= \Theta_2(b_k)$
if there exists a constant $c>0$ such that
$a_k \geq c b_k$ and we say $a_k = \Theta_1(b_k)$, whenever,
$a_k = \Theta_2(b_k)$ and $b_k=\Theta_2(a_k)$.

We first consider the case $\lambda_n=\Theta_1(\sqrt{m})$ 
For simplicity, we assume $\lambda_n=\zeta\sqrt{m}$ for some constant $\zeta>0$. Denote $df(S_{\lambda_n,\gamma}(Z))= D_{\lambda_n,\gamma}$, and use $\mathcal{T}_{t_1,t_2}$ to represent
\begin{align*}
\frac{\sqrt{mt_1} s_{\lambda_n,\gamma}(\sqrt{mt_1}) - \sqrt{mt_2} s_{\lambda_n,\gamma}(\sqrt{mt_2})}{mt_1-mt_2} \mathbbm{1}(t_1\neq t_2).
\end{align*}
Adopting the notation from Lemma \ref{lemma3}, it is not hard to verify that
\begin{align*}
D_{\lambda_n,\gamma} = n \mathbb{E}_{\mu_n} \bigg\{ s'_{\lambda_n,\gamma}(\sqrt{mt_1}) + |m-n|\frac{s_{\lambda_n,\gamma}(\sqrt{mt_1})}{\sqrt{mt_1}} \bigg\} + n^2 \mathbb{E}_{\mu_n} (\mathcal{T}_{t_1,t_2})\,,
\end{align*}
where $t_1, t_2 \overset{\text{iid}}{\sim} \mu_n$. A quick check of the relation between $s_{\lambda_n,\gamma}$ and $g_{\zeta,\gamma}$ yields
\begin{eqnarray*}
\frac{D_{\lambda_n,\gamma}}{mn}=\frac{1}{m}\mathbb{E}_{\mu_n}s'_{\lambda_n,\gamma}(\sqrt{mt_1})+\left(1-\frac{n}{m}\right)\mathbb{E}_{\mu_n}g_{\zeta,\gamma}(t_1)+\frac{n}{m} \mathbb{E}_{\mu_n} \left\{\frac{t_1g_{\zeta,\gamma}(t_1)-t_2g_{\zeta,\gamma}(t_2)}{t_1-t_2}\mathbbm{1}(t_1\neq t_2) \right\}\,.
\end{eqnarray*}
Due to the Lipschitz continuity of the functions $s_{\lambda_n,\gamma}(x)$ and $xg_{\zeta, \gamma}(x)$, we obtain
\begin{eqnarray*}
\Big| \frac{D_{\lambda_n,\gamma}}{mn} \Big| \leq \frac{\gamma}{m(\gamma-1)}+\left(1-\frac{n}{m}\right) + \frac{n}{m}\left(\frac{2\gamma-1}{2\gamma-2}\right) \,.
\end{eqnarray*}
Hence, there exists a positive constant $C_{\alpha}$, such that for sufficiently large $n$,
$$
\Big| \frac{D_{\lambda_n,\gamma}}{mn} \Big| \leq C_{\alpha}, \quad \, a.s.
$$
Let $T_1,T_2$ be two independent random variables generated from the Marchenko-Pastur distribution $\mu$. If we can show 
\begin{align*}
\frac{D_{\lambda_n,\gamma}}{mn} \overset{a.s.}{\rightarrow} (1-\alpha)\mathbb{E}g_{\zeta,\gamma}(T_1) + \alpha \mathbb{E}\left(\frac{T_1g_{\zeta,\gamma}(T_1)-T_2g_{\zeta,\gamma}(T_2)}{T_1-T_2}\right),
\end{align*}
then by the Dominated Convergence Theorem (DCT), we conclude the proof in the $\lambda_n=\Theta_1(\sqrt{m})$ regime. Note immediately that
\begin{eqnarray}\label{mndfone}
\frac{1}{m} \mathbb{E}_{\mu_n}s'_{\lambda_n,\gamma}(\sqrt{mt_1}) \rightarrow 0 \quad a.s.
\end{eqnarray}
Moreover, given that $g_{\zeta,\gamma}(\cdot)$ is bounded and continuous, the Marchenko-Pastur theorem in Lemma \ref{lemma3} implies
\begin{eqnarray}\label{mndftwo}
\left(1-\frac{n}{m}\right) \mathbb{E}_{\mu_n} g_{\zeta,\gamma}(t_1) \rightarrow (1-\alpha) \mathbb{E}_{\mu} g_{\zeta,\gamma}(T_1) \quad a.s.
\end{eqnarray}
Since $(t_1, t_2) \overset{d}{\rightarrow} (T_1, T_2)$, and the discontinuity set of the function $\frac{t_1g_{\zeta,\gamma}(t_1)-t_2g_{\zeta,\gamma}(t_2)}{t_1-t_2}\mathbbm{1}(t_1\neq t_2)$ has zero probability under the measure induced by $(T_1,T_2)$, by the continuous mapping theorem,
\begin{align*}
\frac{t_1g_{\zeta,\gamma}(t_1)-t_2g_{\zeta,\gamma}(t_2)}{t_1-t_2}\mathbbm{1}(t_1\neq t_2) \overset{d}{\rightarrow} \frac{T_1g_{\zeta,\gamma}(T_1)-T_2g_{\zeta,\gamma}(T_2)}{T_1-T_2}\mathbbm{1}(T_1 \neq T_2) \quad \text{as } \, n \rightarrow \infty\,.
\end{align*}
Also, due to the boundedness of $\frac{t_1g_{\zeta,\gamma}(t_1)-t_2g_{\zeta,\gamma}(t_2)}{t_1-t_2}\mathbbm{1}(t_1\neq t_2)$, it holds that
\begin{align}\label{mndfthree}
\mathbb{E}_{\mu_n} \left\{ \frac{t_1g_{\zeta,\gamma}(t_1)-t_2g_{\zeta,\gamma}(t_2)}{t_1-t_2}\mathbbm{1}(t_1\neq t_2) \right\} \overset{a.s.}{\rightarrow}  \mathbb{E}_{\mu} \left\{ \frac{T_1g_{\zeta,\gamma}(T_1)-T_2g_{\zeta,\gamma}(T_2)}{T_1-T_2}\mathbbm{1}(T_1 \neq T_2)\right\}.
\end{align}
Combining \eqref{mndfone} -- \eqref{mndfthree} completes the proof for the $\lambda_n=\Theta_1(\sqrt{m})$ case.

When $\lambda_n=o(\sqrt{m})$, we can readily see that $$\mathbb{E}_{\mu_n}\mathbbm{1}(\sqrt{mt_1} \geq \lambda_n \gamma) \rightarrow 1, a.s.$$ Using that both $\frac{s_{\lambda_n,\gamma}(\sqrt{mt_1})}{\sqrt{mt_1}}\,$ and $\mathcal{T}_{t_1,t_2}$ are bounded, we have, almost surely
\begin{align*}
\mathbb{E}_{\mu_n}\frac{s_{\lambda_n,\gamma}(\sqrt{mt_1})}{\sqrt{mt_1}} = &\mathbb{E}_{\mu_n}\mathbbm{1}(\sqrt{mt_1}\geq \lambda_n \gamma) + \mathbb{E}_{\mu_n} \left\{ \frac{s_{\lambda_n, \gamma}(\sqrt{mt_1})}{\sqrt{mt_1}} \mathbbm{1}(\sqrt{mt_1}< \lambda_n \gamma) \right\} \rightarrow 1
\end{align*}
and
\begin{align*}
 \mathbb{E}_{\mu_n}(\mathcal{T}_{t_1,t_2}) = \mathbb{E}_{\mu_n} \mathbbm{1}(\sqrt{mt_1} \geq \lambda_n \gamma) \mathbbm{1}(\sqrt{mt_2} \geq \lambda_n \gamma) + o(1) \rightarrow 1.
\end{align*}
Invoking DCT completes the proof. Similar arguments hold for the case $\lambda_n=\Theta_2(\sqrt{m})$. 
\end{proof}

\subsubsection{Proof of Proposition~\ref{prop-subsp-stab1}}\label{proof-prop-subsp-stab1}
\begin{proof}
Observe that $R$ as defined in Proposition~\ref{mat-pert-svd-2} can be written as:
\begin{equation}
\begin{aligned}
 R &=&  \widetilde{A}\widetilde{V}_{1}  - \widetilde{U}_{1}\widetilde{\Sigma}_{1}   +  (A - \widetilde{A})\widetilde{V}_{1} &=& (A - \widetilde{A})\widetilde{V}_{1}
 \end{aligned} 
\end{equation}
where, above we have used the fact that
$\widetilde{A}\widetilde{V}_{1}  = \widetilde{U}_{1}\widetilde{\Sigma}_{1}$, which follows from the 
definition of the SVD of $\widetilde{A}$. By a simple inequality it follows that 
\begin{equation}\label{bound-R1}
 \| R \|_2 \leq \|(A - \widetilde{A})\|_2 \| \widetilde{V}_{1}\|_2  =  \|(A - \widetilde{A})\|_2,
 \end{equation}
where we have used the 
 the fact that $\| \widetilde{V}_{1}\|_2 =  1$.  Similarly, we have an analogous result for $Q$:
\begin{equation}\label{bound-S1}
 \| Q \|_2 \leq \|(A - \widetilde{A})\|_2 \| \widetilde{U}_{1}\|_2 =   \|(A - \widetilde{A})\|_2.
 \end{equation}
Note that~\eqref{bound-R1} and~\eqref{bound-S1} together imply that if  $\|\widetilde{A} - A \|_2$ is small, then so are $\| R\|_2, \|Q\|_2$.  

We now apply~\eqref{ineq-rho-r} (Proposition~\ref{mat-pert-svd-2}) with $A = X_{k}$ and $\widetilde{A} = X_{k+1}$ and $r_{1} = p$, to arrive at the proof of 
Proposition~\ref{prop-subsp-stab1}.
\end{proof}

\subsubsection{Proof of Proposition~\ref{prop-asympt-conv}}\label{pf-asympt-conv}
\begin{proof}
Proof of Part ({{a}}):\\
Let us write the stationary conditions for every update: $$X_{k+1} = \argmin_{X} \; F_{\ell}(X;X_{k}).$$ We set the subdifferential of the map
$X \mapsto  F_{\ell}(X;X_{k})$ to zero at $X = X_{k+1}$:
\begin{align}\label{deriv-1-1}
\left( X_{k+1} -  \left(\mathcal{P}_{\Omega}(Y)+ \mathcal{P}_{\Omega}^\perp(X_{k}) \right) \right) + \ell (X_{k+1} - X_{k}) +  U_{k+1} \nabla_{k+1} V_{k+1}' = 0,
\end{align}
where $X_{k+1} = U_{k+1} \diag(\B\sigma_{k+1})V'_{k+1}$ is the SVD of $X_{k+1}$. Note that the term: $U_{k+1} \nabla_{k+1} V_{k+1}'$ in~\eqref{deriv-1-1}, 
is a subdifferential~\citep{Lewis1995} of the spectral function: 
$$X \mapsto \sum_{i} P(\sigma_{i}(X); \lambda, \gamma),$$ where
$\nabla_{k+1}$ is a diagonal matrix with the $i$th diagonal entry being a derivative of the map $\sigma_{i} \mapsto P(\sigma_{i}; \lambda,\gamma)$  (on $\sigma_{i} \geq 0$), denoted by $\partial P(\sigma_{k+1, i}; \lambda,\gamma)/\partial \sigma_{i}$ for all $i$. Note that~\eqref{deriv-1-1} can be rewritten as:
\begin{align*}\label{deriv-1-2}
\mathcal{P}_{\Omega}(X_{k+1}) - \mathcal{P}_{\Omega}(Y)  +U_{k+1} \nabla_{k+1} V_{k+1}' +  \underbrace{\left ( \mathcal{P}_{\Omega}^\perp(X_{k+1} - X_{k})  + \ell (X_{k+1} - X_{k}) \right)}_{(a)}  = 0.
\end{align*}
As $k \rightarrow \infty$, term (a) converges to zero (See Proposition~\ref{conv-rate-1}) and thus, we have:
\begin{equation*}\label{deriv-1-3}
\mathcal{P}_{\Omega}(X_{k+1}) - \mathcal{P}_{\Omega}(Y)  +  U_{k+1} \nabla_{k+1} V_{k+1}' \rightarrow  0.
\end{equation*}
Let us denote the $i$th column of $U_{k}$ by ${u}_{k,i}$, and use a similar notation for $V_{k}$ and $v_{k,i}$. Let $r_{k+1}$ denote the rank of $X_{k+1}$.
Hence, we have:
$$  \sum_{i=1}^{r_{k+1}}\sigma_{k+1,i} \mathcal{P}_{\Omega}(u_{k+1,i} v_{k+1, i}') - \mathcal{P}_{\Omega}(Y)  +  U_{k+1} \nabla_{k+1} V_{k+1}' \rightarrow  0.$$
Multiplying the left and right hand sides of the above by $u'_{k+1,j}$ and $v_{k+1,j}$, we have the following:
\begin{align*}\label{deriv-1-4}
\sum_{i=1}^{r_{k+1}} \sigma_{k+1,i} u'_{k+1,j}\mathcal{P}_{\Omega}(u_{k+1,i}v'_{k+1,i})v_{k+1,j} - u'_{k+1,j}\mathcal{P}_{\Omega}(Y)v_{k+1,j} +  \nabla_{k+1,j} \rightarrow  0,
\end{align*}
 for  $j = 1, \ldots, r_{k+1}.$
Let  $\left \{ \bar{U}, \bar{V} \right\}$ denote a limit point of the sequence $\left \{U_{k},V_{k}\right\}$ (which exists since the sequence is bounded); and let $r$ be the rank of $\bar{U}$ and $\bar{V}$. Let us now study the following equations\footnote{Note that we do not assume that the sequence $\B\sigma_{k}$ has a limit point.}:
\begin{equation}\label{deriv-1-5}
\sum_{i=1}^{r} \bar{\sigma}_{j} \bar{u}'_{j}\mathcal{P}_{\Omega}(\bar{u}_{i}\bar{v}'_{i})\bar{v}_{j} - \bar{u}'_{j}\mathcal{P}_{\Omega}(Y)\bar{v}_{j}  +  \bar{\nabla}_{j} =  0, \;\;\; j = 1, \ldots, r.
\end{equation}
Using the notation  $\bar{\theta}_{j} = \text{vec} \left( \mathcal{P}_{\Omega}(\bar{u}_{j}\bar{v}'_{j}) \right)$ and $\bar{y} = \text{vec}(\mathcal{P}_{\Omega}(Y))$,
we note that~\eqref{deriv-1-5} are the first order stationary conditions for a point $\bar{\B\sigma}$ for the following penalized regression problem:
\begin{equation}\label{deriv-1-6}
\mini_{\B\sigma} \;\; \frac12 \| \sum_{j=1}^{r} \sigma_{j} \bar{\theta}_{j} - \bar{y} \|_{2}^2  + \sum_{j=1}^{r} P(\sigma_{j}; \lambda,\gamma),
\end{equation}
with $\B\sigma \geq \M{0}$. If the matrix $\bar{\Theta} = [\bar\theta_{1}, \ldots, \bar\theta_{r}]$ (note that $\bar{\Theta} \in \mathbb{R}^{mn \times r}$) has rank $r$, then any $\B\sigma$ that satisfies~\eqref{deriv-1-5} is finite --- in particular, the sequence $\B\sigma_{k}$ is bounded and has a limit point: $\bar{\B\sigma}$ which satisfies the first order stationary condition~\eqref{deriv-1-5}.

Proof of Part ({{b}}):\\
Furthermore, if we assume that
$$\lambda_{\min}( \bar{\Theta}'\bar{\Theta}) + \phi_P > 0,$$
then~\eqref{deriv-1-6} admits a unique solution $\bar{\B\sigma}$, which implies that $\B\sigma_{k}$ has a unique limit point, and hence the sequence $\B\sigma_{k}$ necessarily converges.


\end{proof}

\subsection{Additional Simulation Results}
\label{add:simu}

This section contains additional numerical results from the simulation study in Section \ref{simulation:study}. 
\begin{itemize}
\item To demonstrate the variation of the procedures in the experiments, we plot the averaged value and standard error of both test error and rank for some representative nonconvex penalty functions. Specifically, under each scenario considered in Section \ref{simulation:study}, we pick the nonconvex penalty that yields the best prediction and rank estimation performance. For each picked penalty, we plot the averaged value of test error and rank along with the associated standard error, against the tuning parameter $\lambda$. The results are shown in Figures \ref{figadd1}, \ref{figadd2}, and \ref{figadd3}. As is clear form the figures, the standard error is typically (at least) one order of magnitude smaller than the average. Moreover, the general patterns of test error and rank on the solution path are expected, except for a few points corresponding to very small values of $\lambda$. The irregularity of these few points occurs probably because the solutions are getting unstable as the nonconvex regularization becomes weak when $\lambda$ is significantly small. 
\item To examine the rank dynamics of the updates in \textsc{NC-Impute}, we compute the number of iterations that the algorithm takes for the convergence of the rank. We choose the same six non-convex penalties as above and evaluate the rank stabilization for several values of $\lambda$. The results are summarized in Figure \ref{figadd4}. One clearly observes that except for few instances, it takes less than 10 iterations for the rank to stabilize. Moreover, when the penalty is more ``nonconvex" (i.e., $\gamma$ is smaller), the rank stabilization occurs earlier. These empirical results provide complementary information on rank stabilization that has been theoretically investigated in \ref{rankstable:sec}.
\end{itemize}

\begin{figure*}[tb!]
\begin{center}
{\bf {Example-A}} (Low SNR, less missing entries) \\
\subfigure[ROM, $90$\% missing, $\text{SNR}=1$, $\text{true rank}=10$]{
\scalebox{.9}{\includegraphics[width=2.8in, height=2.3in]{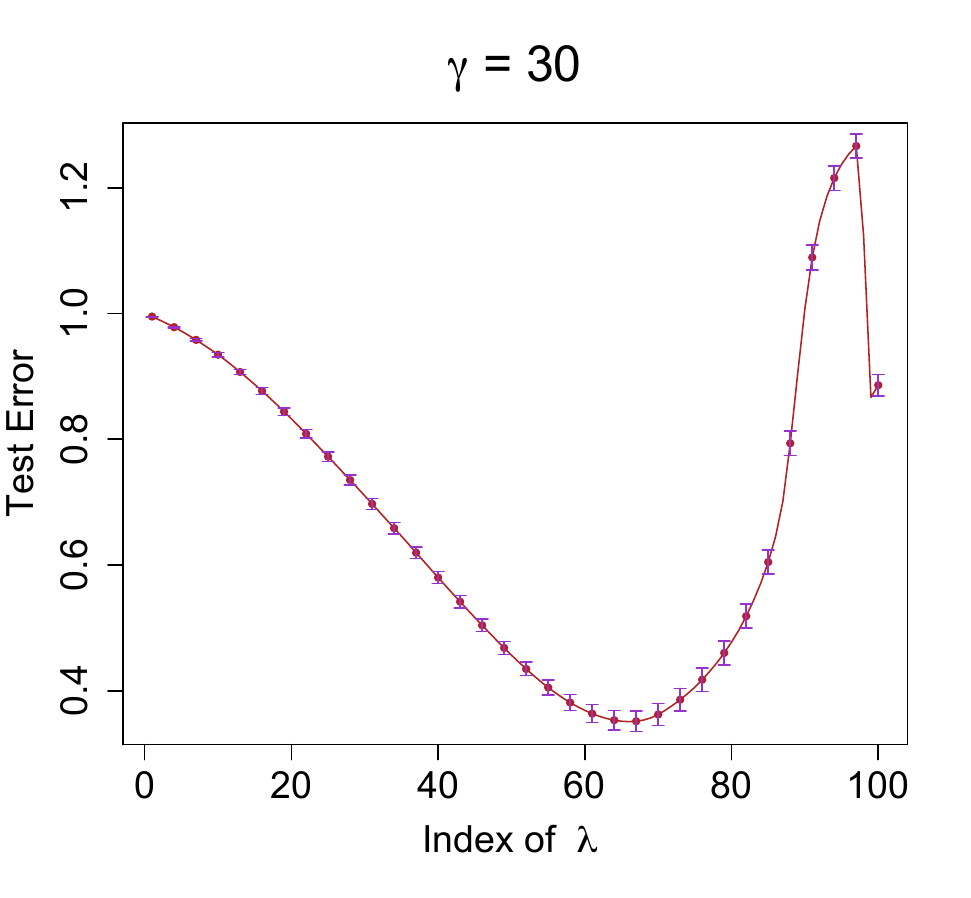}}}
\subfigure[ROM, $90$\% missing, $\text{SNR}=1$, $\text{true rank}=5$]{
\scalebox{.9}{\includegraphics[width=2.8in, height=2.3in]{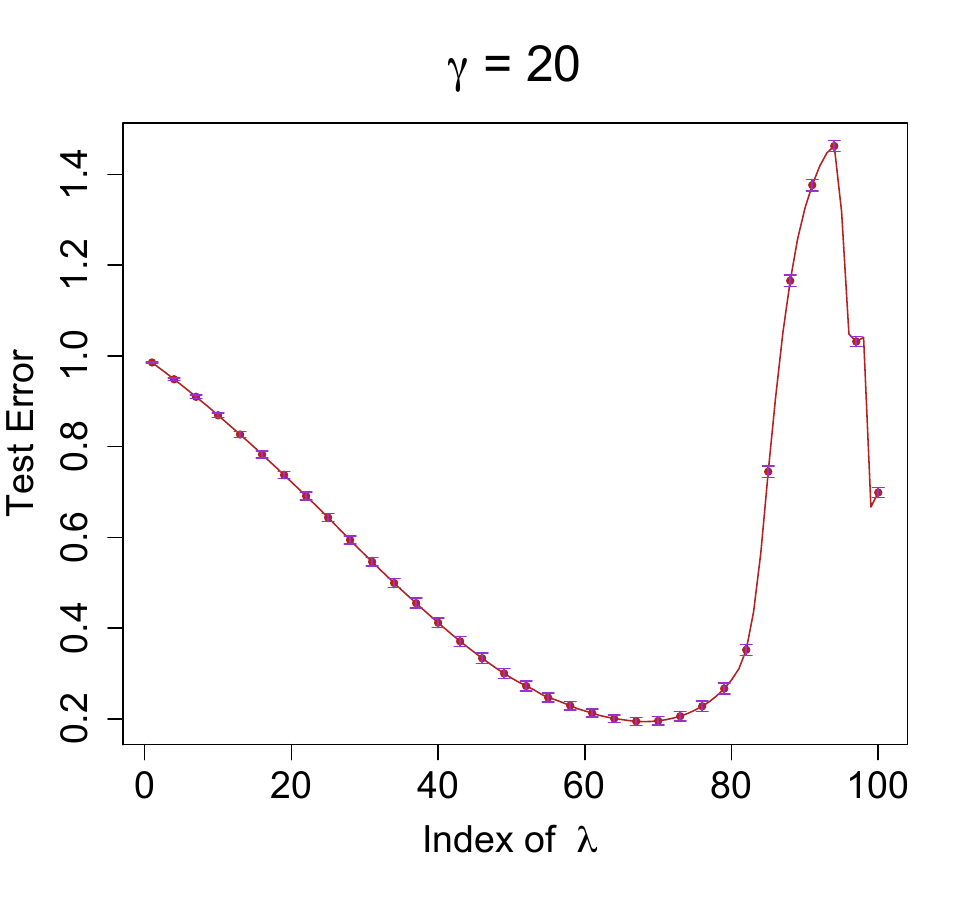}}}
\subfigure[ROM, $90$\% missing, $\text{SNR}=1$, $\text{true rank}=10$]{
\scalebox{.9}{\includegraphics[width=2.8in, height=2.3in]{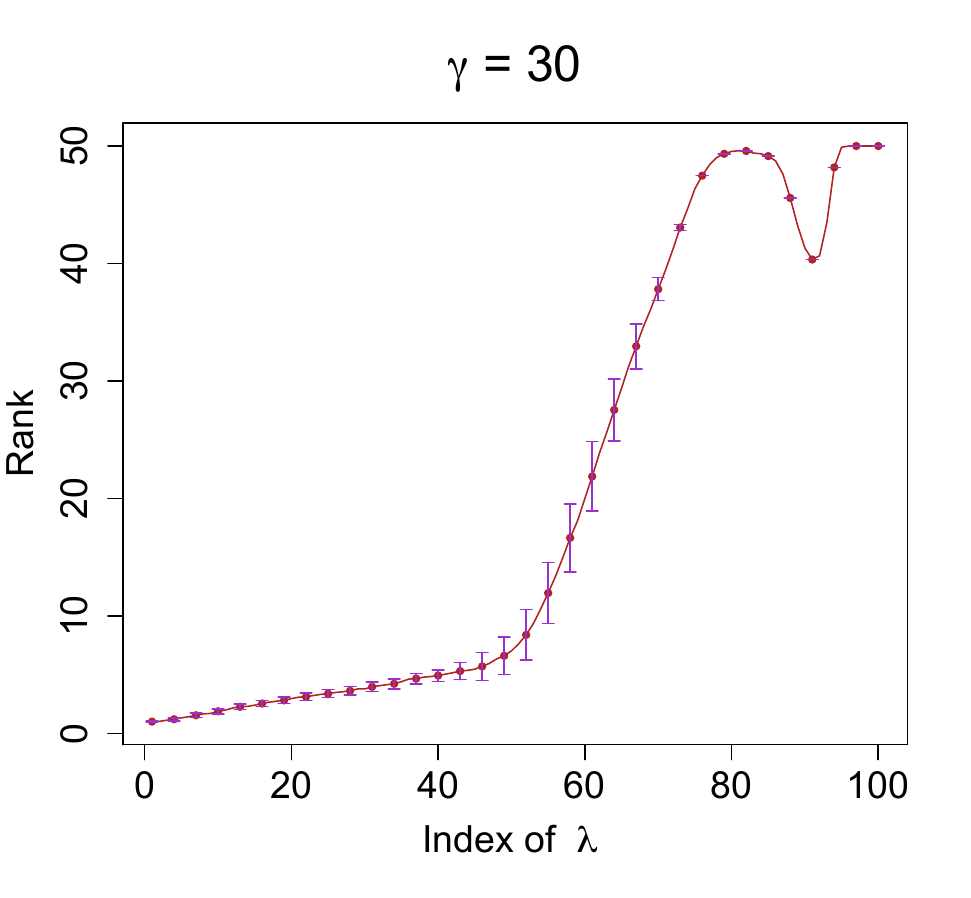}}}
\subfigure[ROM, $90$\% missing, $\text{SNR}=1$, $\text{true rank}=5$]{
\scalebox{.9}{\includegraphics[width=2.8in, height=2.3in]{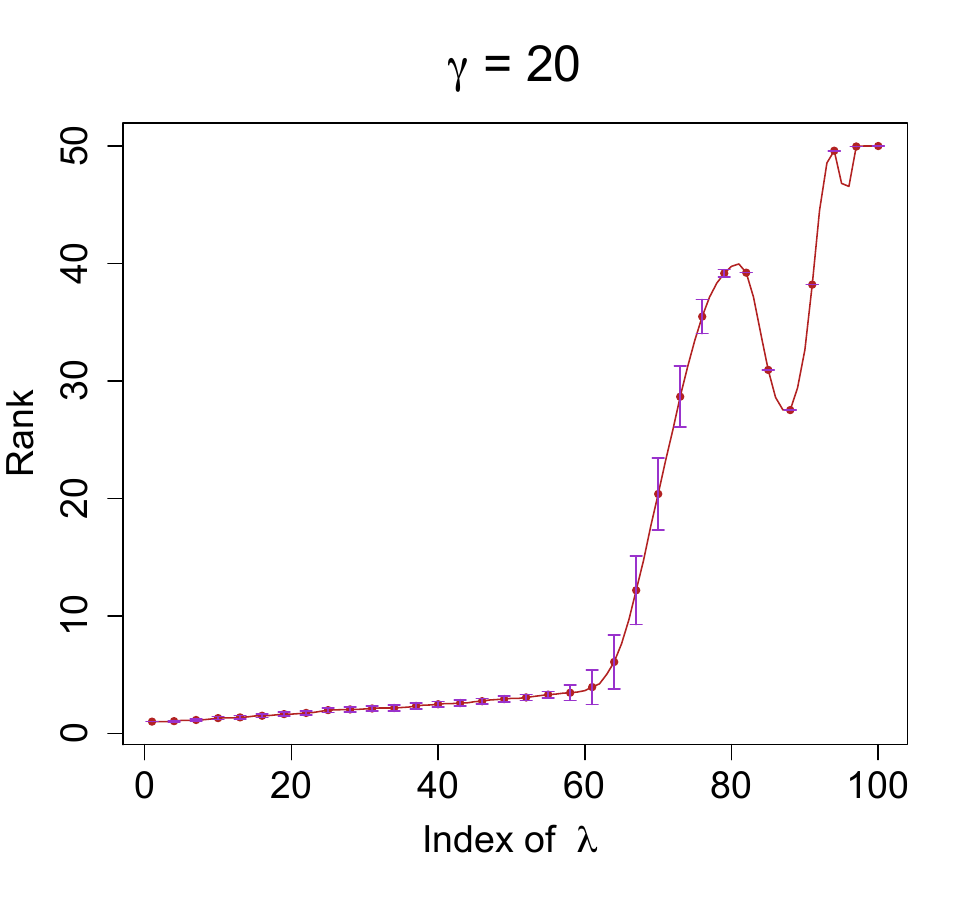}}}
\caption{\small Random Orthogonal Model (ROM) simulations with $\text{SNR}=1$. The optimal nonconvex penalties are obtained at $\gamma=30$ and $\gamma=20$ under the two scenarios respectively. The integers from 1 to 100 on the x-axis index the grid of 100 values of $\lambda$ (from largest to smallest) as described in Section \ref{simulation:study}. }\label{figadd1}
\end{center}
\end{figure*}

\begin{figure*}[tb!]
\begin{center}
{\bf {Example-A}} (High SNR, more missing entries)\\
\subfigure[ROM, $95$\% missing, $\text{SNR}=5$, $\text{true rank}=10$]{
\scalebox{.9}{\includegraphics[width=2.8in, height=2.3in]{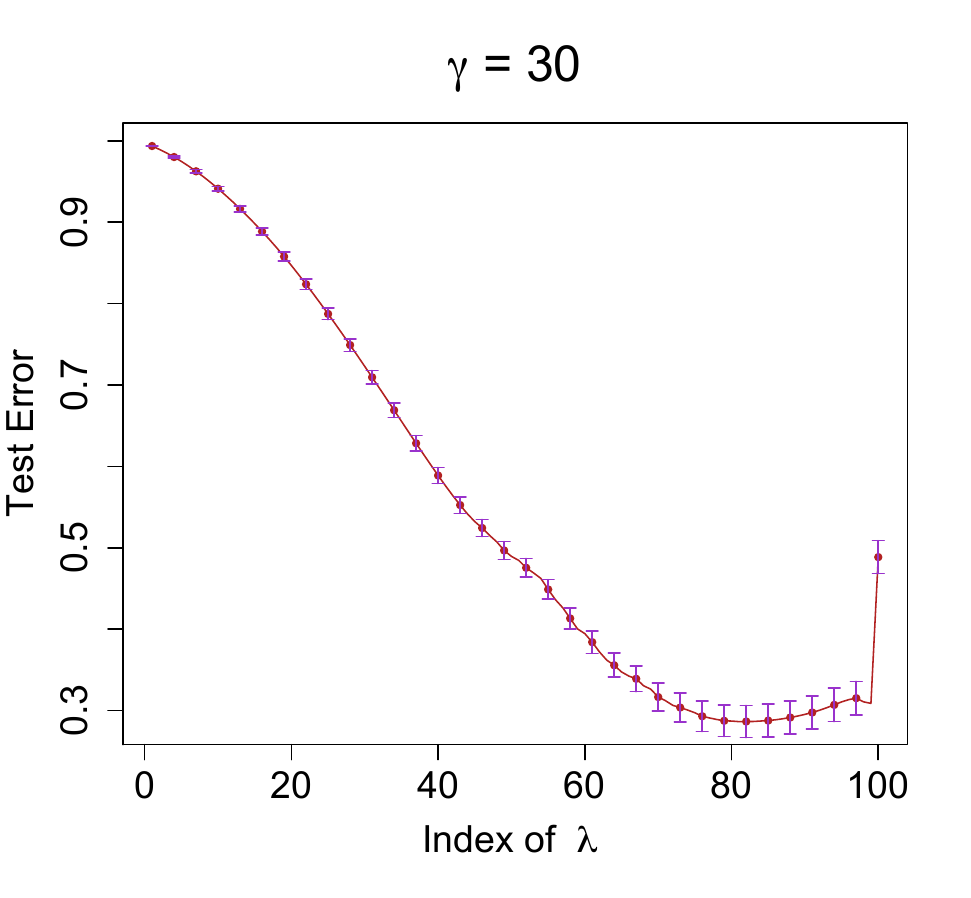}}}
\subfigure[ROM, $95$\% missing, $\text{SNR}=5$, $\text{true rank}=5$]{
\scalebox{.9}{\includegraphics[width=2.8in, height=2.3in]{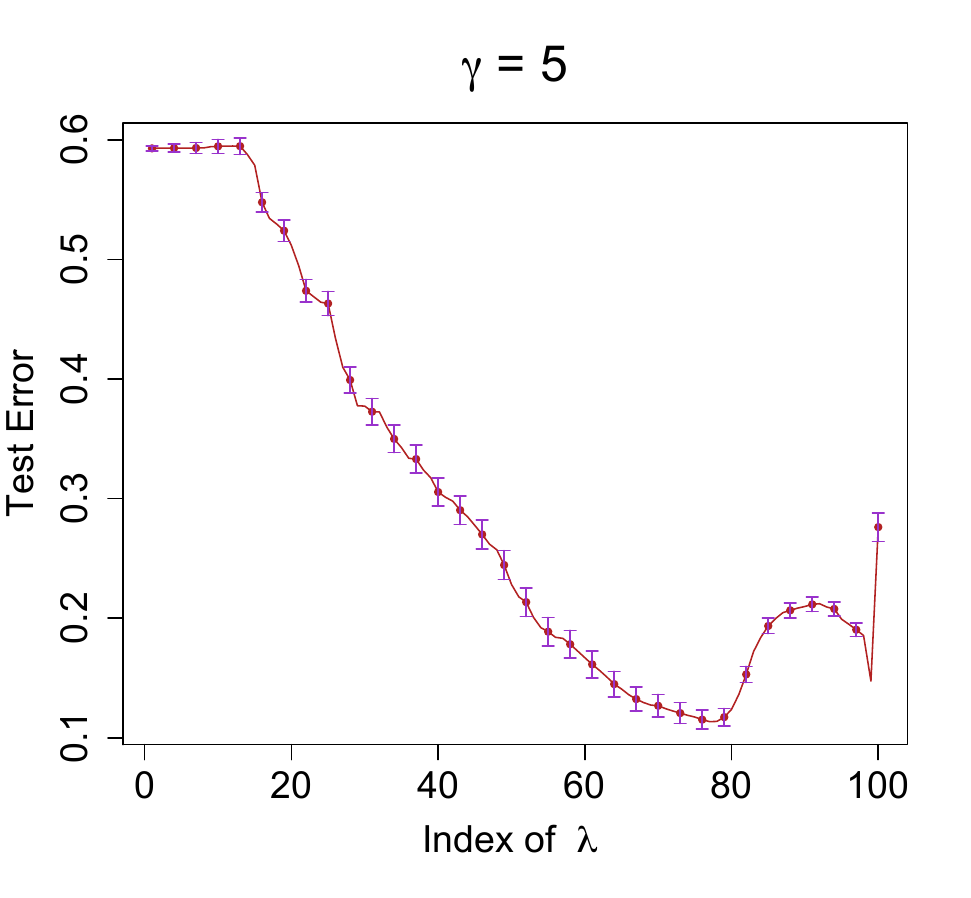}}}
\subfigure[ROM, $95$\% missing, $\text{SNR}=5$, $\text{true rank}=10$]{
\scalebox{.9}{\includegraphics[width=2.8in, height=2.3in]{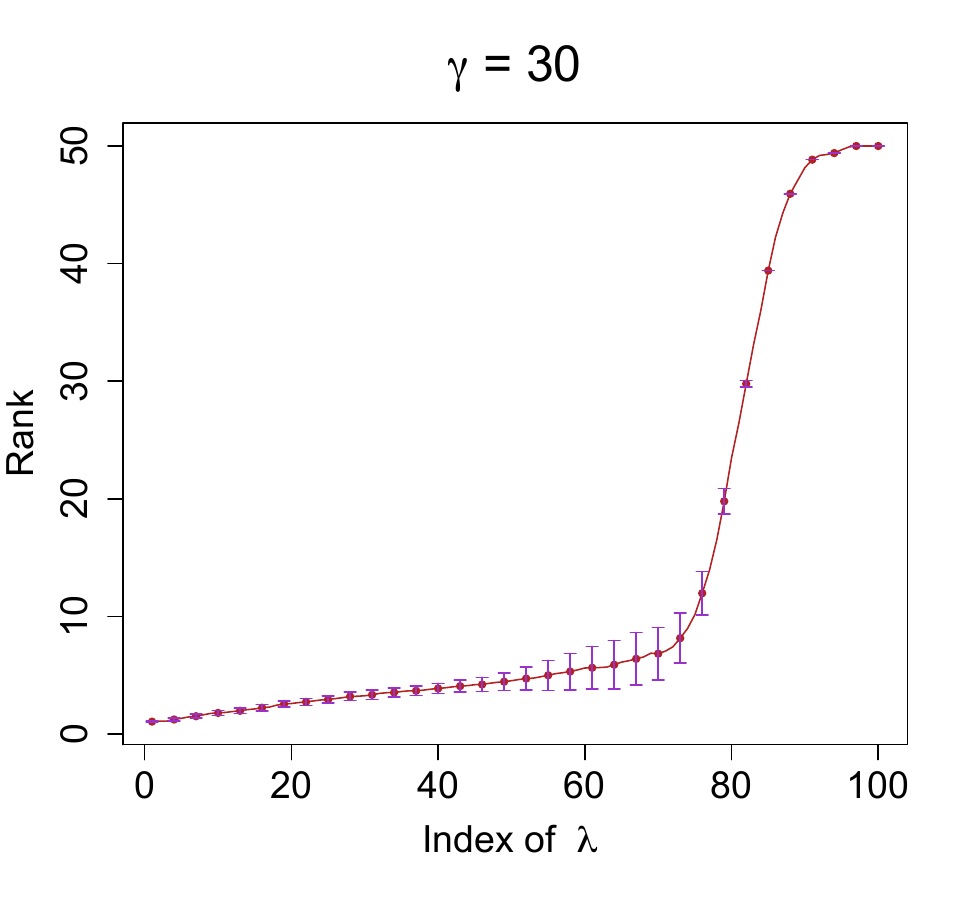}}}
\subfigure[ROM, $95$\% missing, $\text{SNR}=5$, $\text{true rank}=5$]{
\scalebox{.9}{\includegraphics[width=2.8in, height=2.3in]{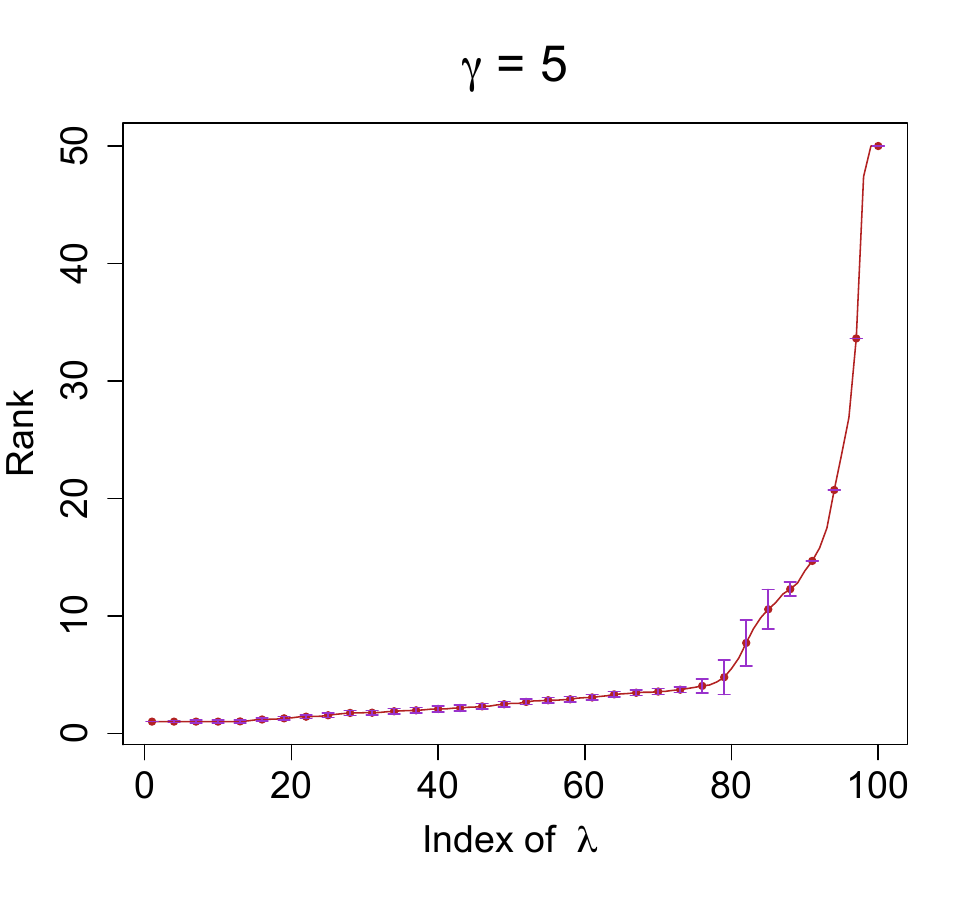}}}
\caption{\small Random Orthogonal Model (ROM) simulations with $\text{SNR}=5$. The optimal nonconvex penalties are obtained at $\gamma=30$ and $\gamma=5$ under the two scenarios respectively. The integers from 1 to 100 on the x-axis index the grid of 100 values of $\lambda$ (from largest to smallest) as described in Section \ref{simulation:study}.}\label{figadd2}
\end{center}
\end{figure*}

\begin{figure*}[htb!]
\begin{center}
{\bf {Example-B}} \hspace{5cm} {\bf {Example-C}} \\
\subfigure[\scriptsize Coherent, $90$\% missing, $\text{SNR}=10$, $\text{true rank}=10$]{
\scalebox{.9}{\includegraphics[width=2.8in, height=2.3in]{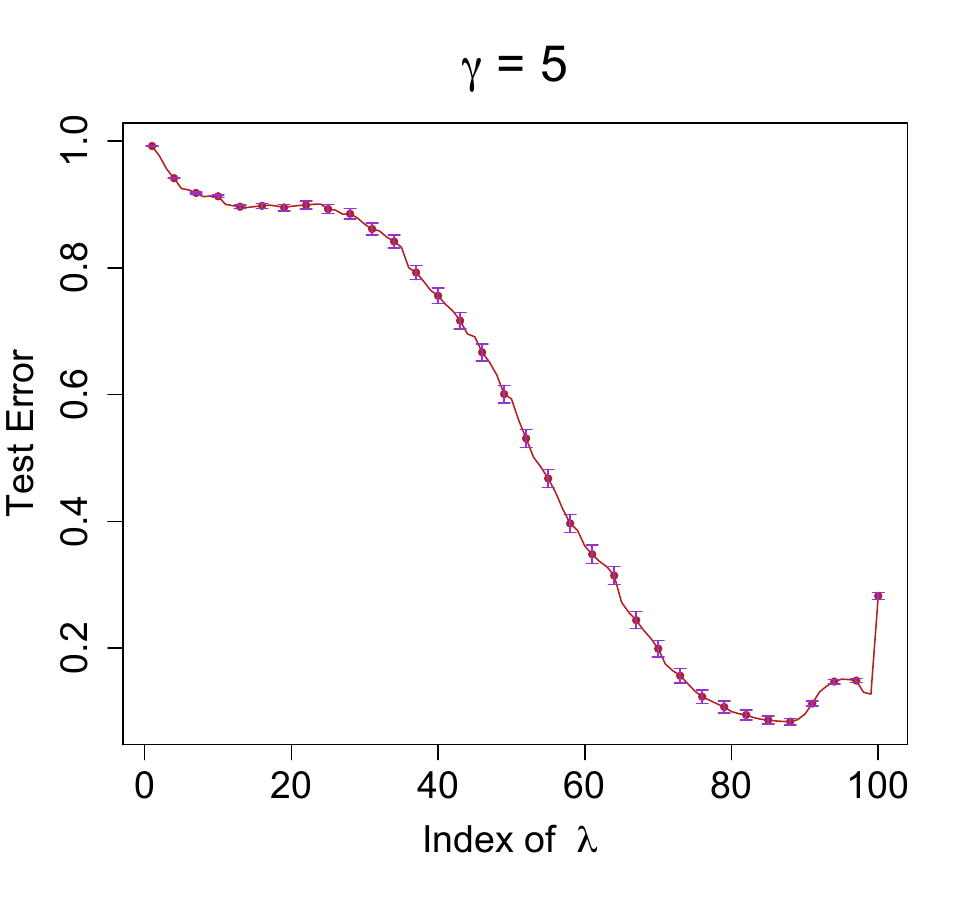}}}
\subfigure[\scriptsize NUS, $25$\% missing, $\text{SNR}=10$, $\text{true rank}=10$]{
\scalebox{.9}{\includegraphics[width=2.8in, height=2.3in]{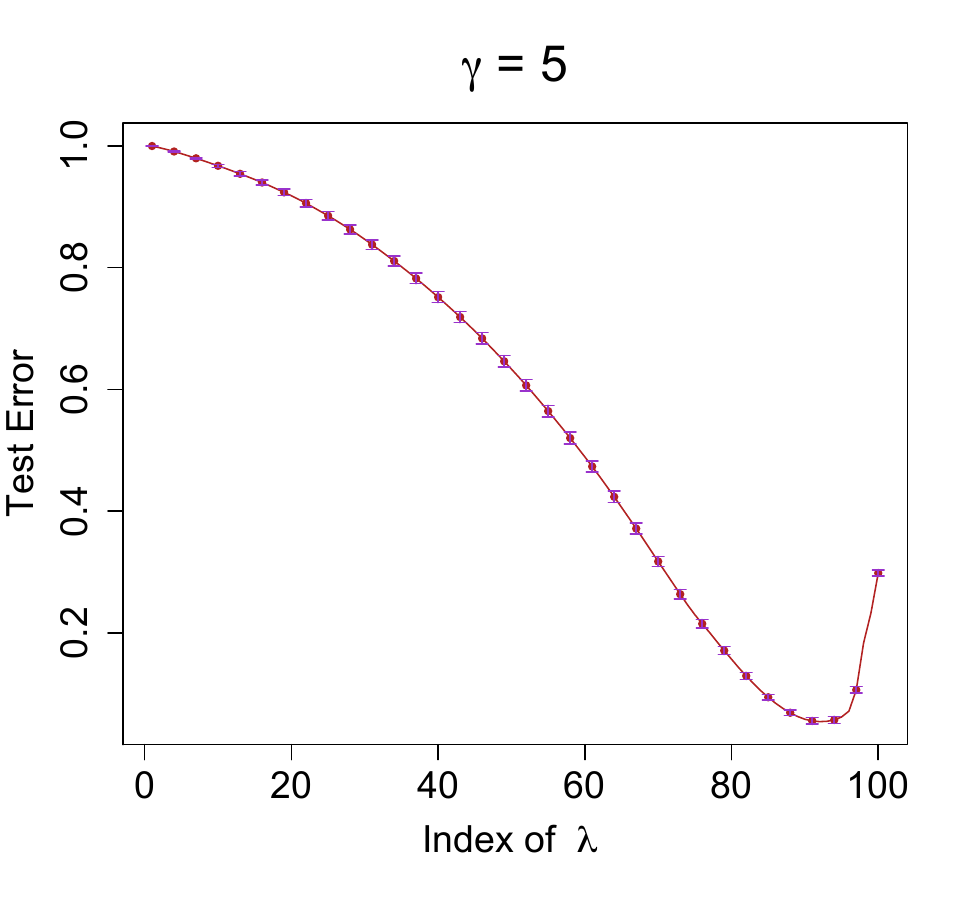}}}
\subfigure[\scriptsize Coherent, $90$\% missing, $\text{SNR}=10$, $\text{true rank}=10$]{
\scalebox{.9}{\includegraphics[width=2.8in, height=2.3in]{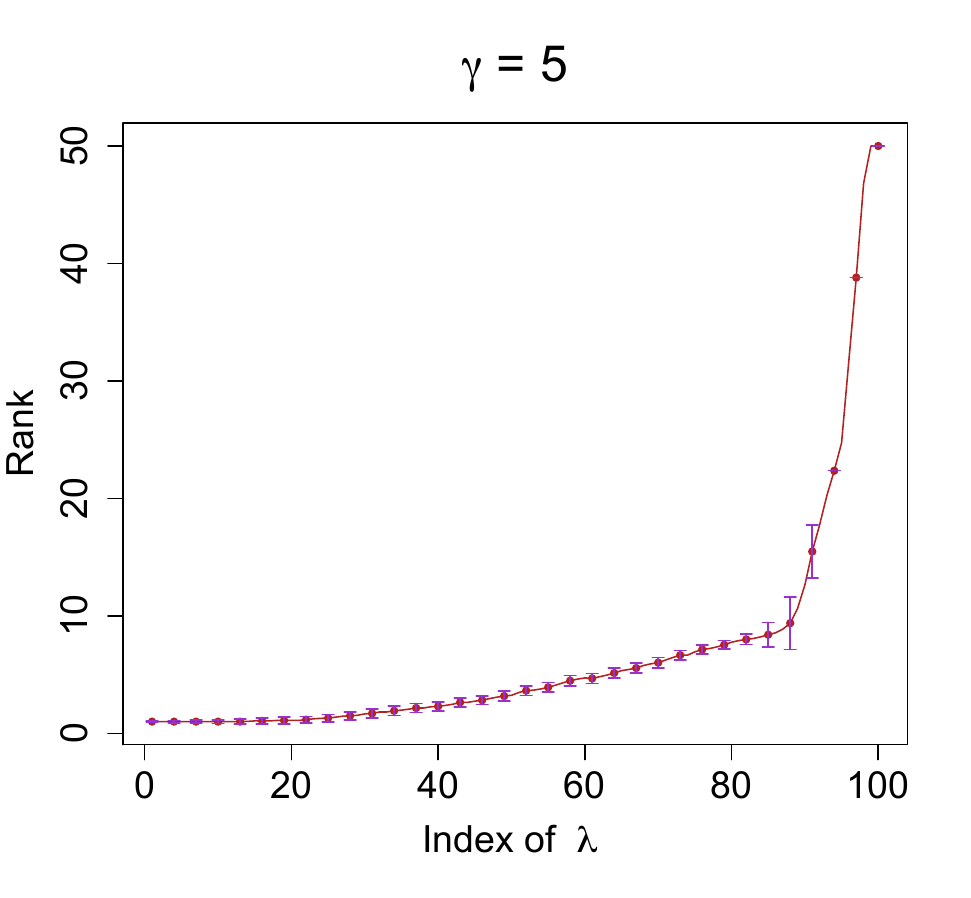}}}
\subfigure[\scriptsize NUS, $25$\% missing, $\text{SNR}=10$, $\text{true rank}=10$]{
\scalebox{.9}{\includegraphics[width=2.8in, height=2.3in]{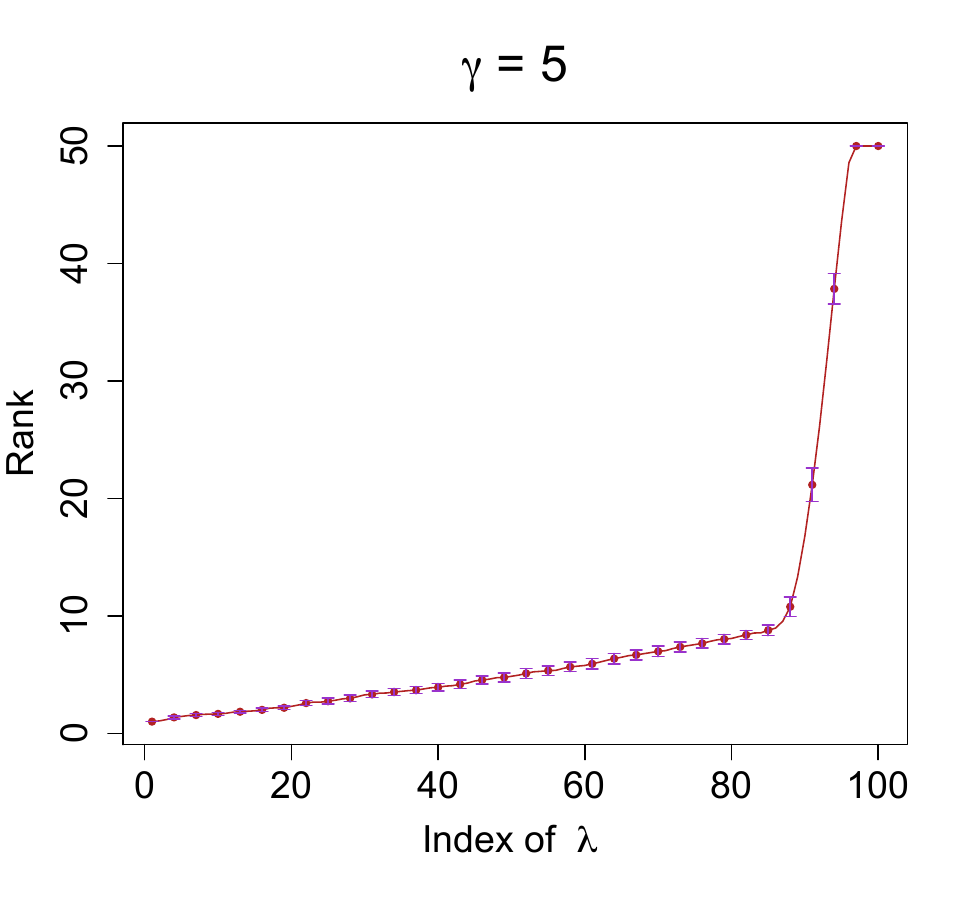}}}
\caption{\small Coherent and Nonuniform Sampling (NUS) simulations with $\text{SNR}=10$. The optimal nonconvex penalties are both obtained at $\gamma=5$ under the two scenarios respectively. The integers from 1 to 100 on the x-axis index the grid of 100 values of $\lambda$ (from largest to smallest) as described in Section \ref{simulation:study}.}\label{figadd3}
\end{center}
\end{figure*}

\begin{figure*}[htb!]
\begin{center}

\subfigure[ROM, $90$\% missing, $\text{SNR}=1$, $\text{true rank}=10$]{
\scalebox{.9}{\includegraphics[width=2.3in, height=2.2in]{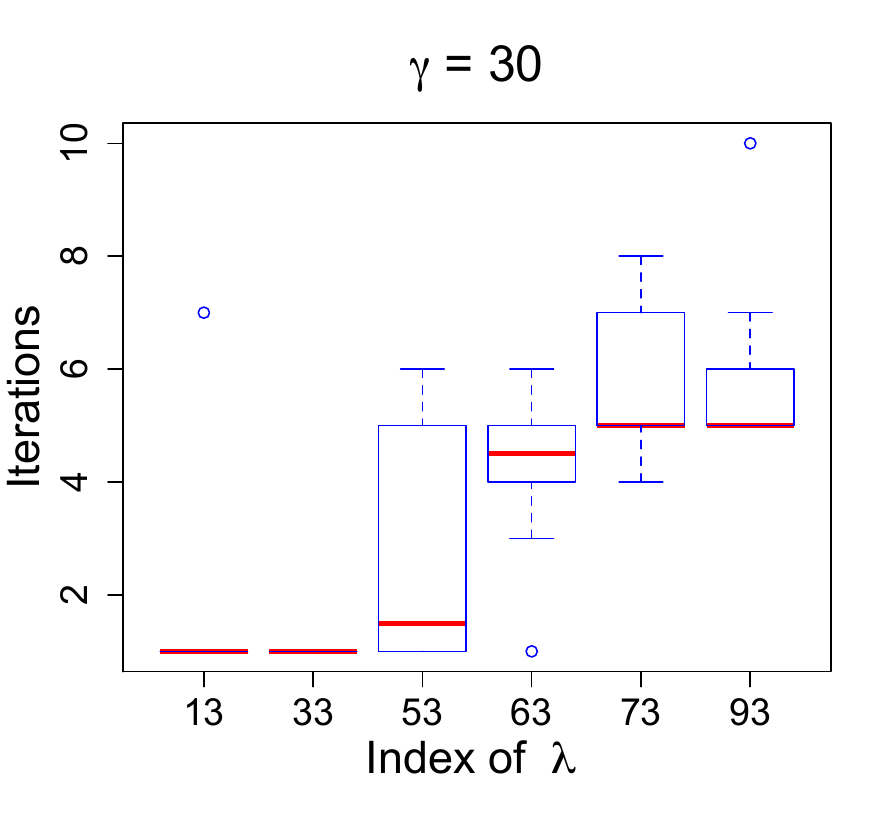}}}
\subfigure[ROM, $90$\% missing, $\text{SNR}=1$, $\text{true rank}=5$]{
\scalebox{.9}{\includegraphics[width=2.3in, height=2.2in]{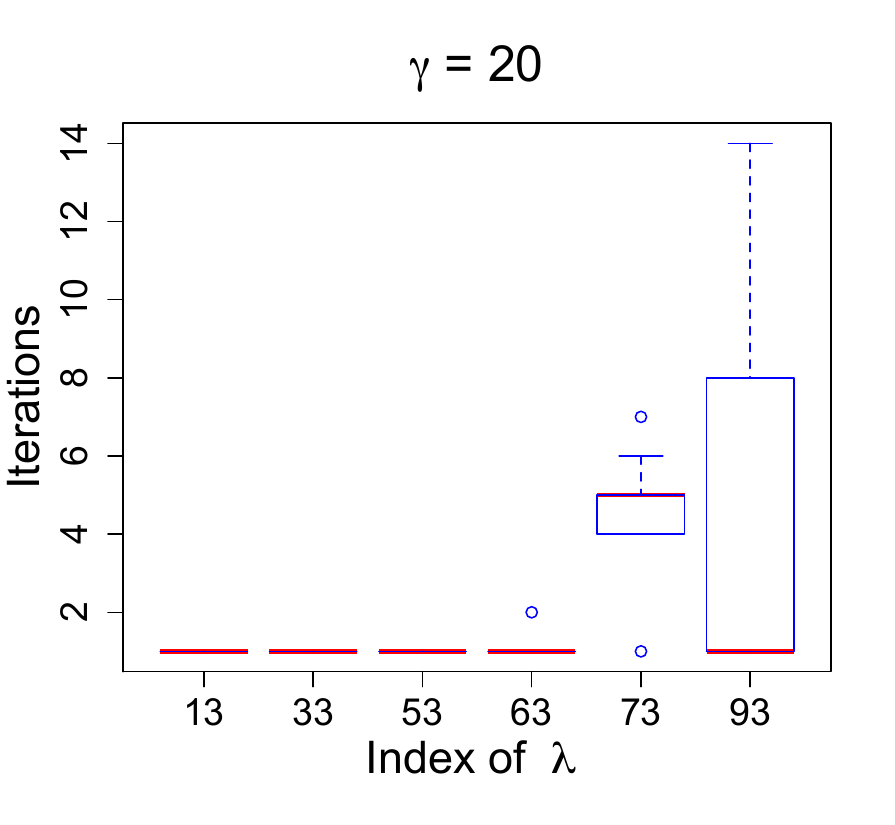}}}
\subfigure[ROM, $95$\% missing, $\text{SNR}=5$, $\text{true rank}=10$]{
\scalebox{.9}{\includegraphics[width=2.3in, height=2.2in]{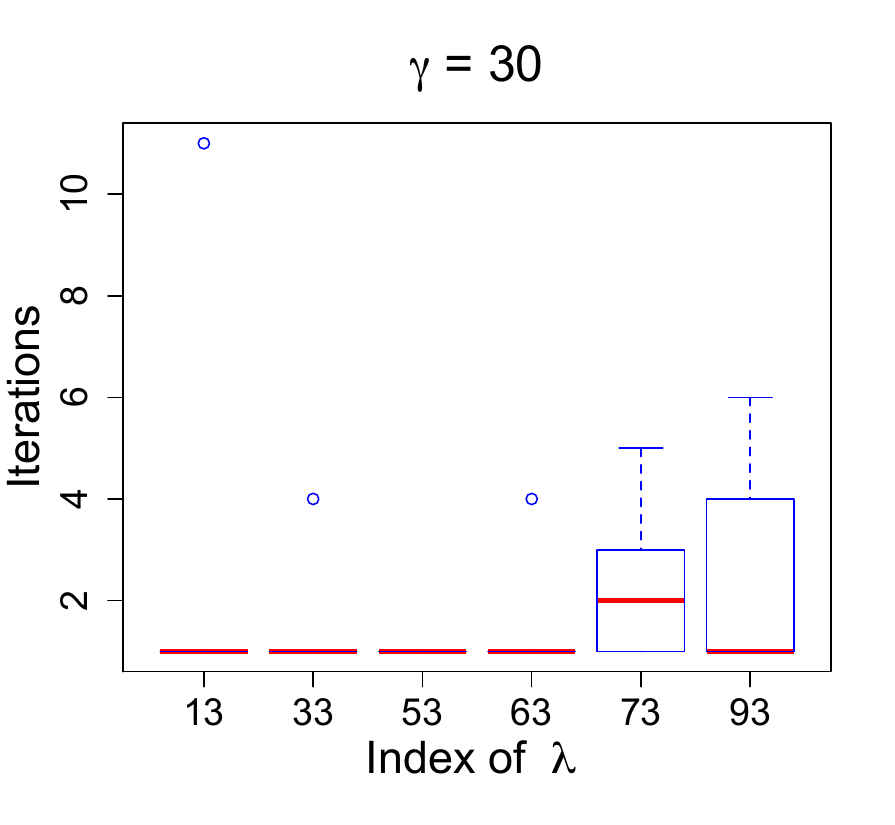}}}

\subfigure[ROM, $95$\% missing, $\text{SNR}=5$, $\text{true rank}=5$]{
\scalebox{.9}{\includegraphics[width=2.3in, height=2.2in]{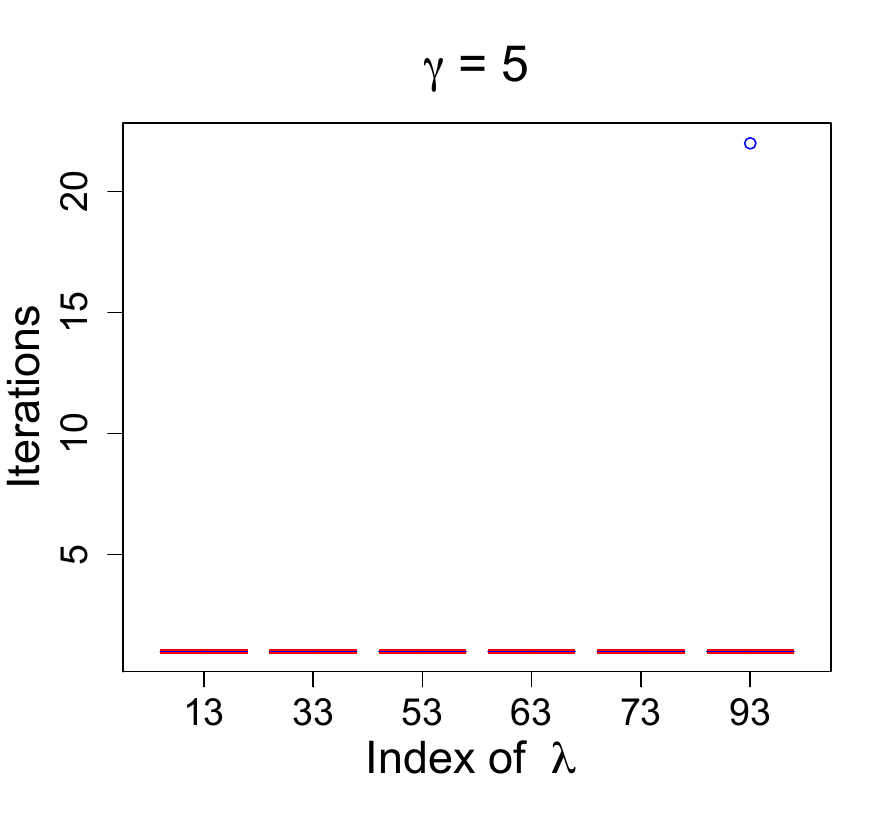}}}
\subfigure[Coherent, $90$\% missing, $\text{SNR}=10$, $\text{true rank}=10$]{
\scalebox{.9}{\includegraphics[width=2.3in, height=2.2in]{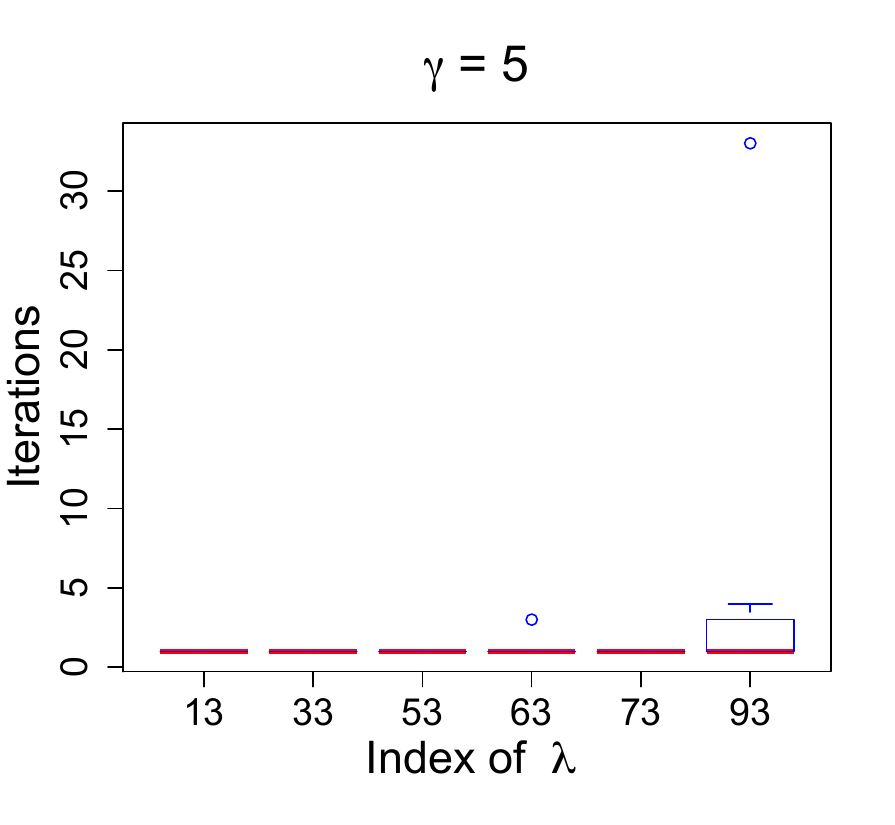}}}
\subfigure[NUS, $25$\% missing, $\text{SNR}=10$, $\text{true rank}=10$]{
\scalebox{.9}{\includegraphics[width=2.3in, height=2.2in]{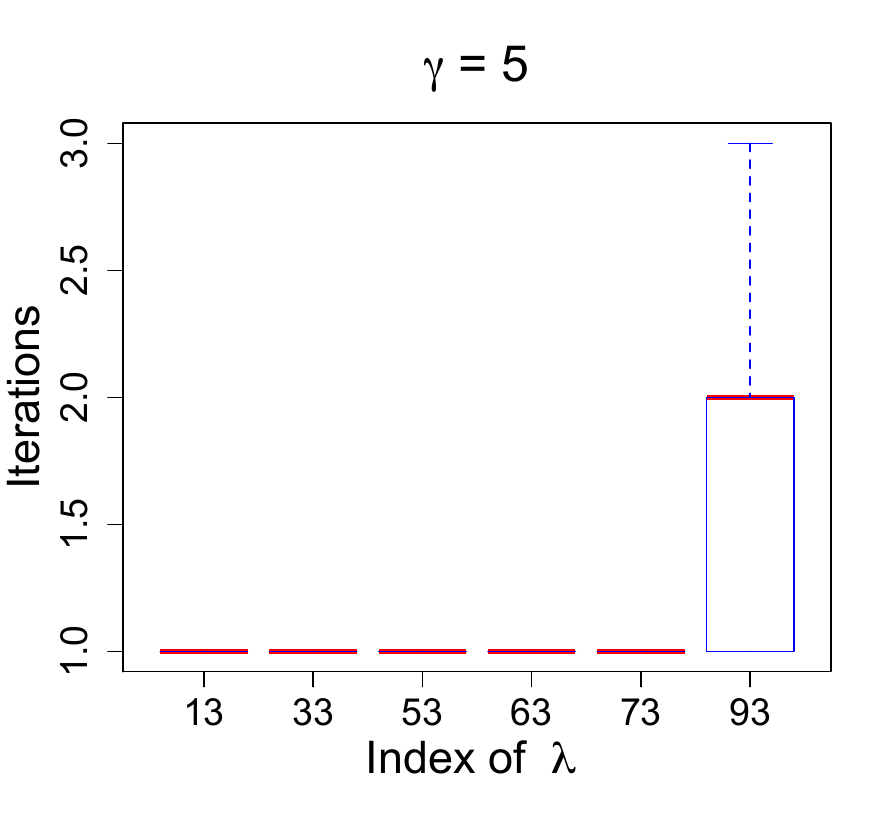}}}
\caption{\small The y-axis denotes the number of iterations \textsc{NC-Impute} takes to stabilize the rank. The integers on the x-axis index some values on a grid of $\lambda$ (from largest to smallest) as described in Section \ref{simulation:study}. The six plots represent the six scenarios considered in Section \ref{simulation:study}: (a)-(d) correspond to the four scenarios of Example-A; (e) covers Example-B; (f) is for Example-C. Each procedure is repeated 10 times.}\label{figadd4}
\end{center}
\end{figure*}


%
%

\bibliographystyle{plainnat_my}
\bibliography{mc}

\end{document}